\documentclass[aps,prx,reprint,showpacs,twocolumn,superscriptaddress,eqsecnum]{revtex4-2}
\usepackage{xr}
\usepackage{amsmath}
\usepackage{amsfonts,amssymb,amsthm,amsxtra}
\usepackage{algorithm} 
\usepackage{algpseudocode} 
\usepackage[]{graphicx}
\usepackage{mathrsfs}
\usepackage{bbm}
\usepackage{bm}
\usepackage{braket}
\usepackage{xcolor}
\usepackage{mathtools}
\usepackage{tikz-cd}
\usepackage{lipsum}
\usepackage{booktabs}
\usepackage[version=4]{mhchem}
\usepackage{float}
\usepackage[bookmarks=false,colorlinks=true,urlcolor=blue,citecolor=blue,linkcolor=blue]{hyperref}
\usepackage{quantikz}

\newtheorem{theorem}{Theorem}
\newtheorem*{theorem*}{Theorem}
\newtheorem{lemma}[theorem]{Lemma}

\newtheorem{procedure}[theorem]{Procedure}
\newtheorem{definition}[theorem]{Definition}

\newcommand{\F}{\mathbb{F}}
\newcommand{\cP}{\mathcal{P}}
\newcommand{\cS}{\mathcal{S}}

\newcommand{\cC}{\mathcal{C}}
\newcommand{\cQ}{\mathcal{Q}}

\begin{document}

\title{Adaptive Syndrome Extraction}

\author{Noah Berthusen}
\email{nfbert@umd.edu}
\author{Shi Jie Samuel Tan}
\author{Eric Huang}
\author{Daniel Gottesman}
\affiliation{Joint Center for Quantum Information and Computer Science, NIST/University of Maryland, College Park, Maryland 20742, USA}

\begin{abstract}
Device error rates on current quantum computers have improved enough to where demonstrations of error correction below break-even are now possible. Still, the circuits required for quantum error correction introduce significant overhead and sometimes inject more errors than they correct. 
In this work, we introduce adaptive syndrome extraction as a scheme to improve code performance and reduce the quantum error correction cycle time by measuring only the stabilizer generators that are likely to provide useful syndrome information. 
We provide a concrete example of the scheme through the $[[4,2,2]]$ code concatenated with a hypergraph product code and a syndrome extraction cycle that uses quantum error detection to modify the syndrome extraction circuits in real time.
Compared to non-concatenated codes and non-adaptive syndrome extraction, we find that the adaptive scheme achieves over an order of magnitude lower logical error rates while requiring fewer CNOT gates and physical qubits.
Furthermore, we show how to achieve fault-tolerant universal logical computation with $[[4,2,2]]$-concatenated hypergraph product codes.
\end{abstract}


\maketitle


\section{Introduction}

Recent improvements in quantum computing hardware have allowed for demonstrations of quantum error correcting codes (QECCs) operating below threshold~\cite{acharya2024, reichardt2024, dasilva2024, reichardt2024_2, lacroix2024scalinglogiccolorcode}. 
However, performing the syndrome extraction circuits can introduce new errors through faulty gates, idling errors, and incorrect corrections.  
Indeed, it is sometimes the case that bare qubits can idle longer than a logical qubit of a QECC that undergoes several rounds of syndrome extraction and error correction~\cite{Xu_2018}. This issue can be further exacerbated when, for example, implementing nonlocal QECCs on hardware which only have access to local gates~\cite{Delfosse_Beverland_Tremblay_2021, baspin2023lower}.

Quantum error correction (QEC) is inherently adaptive in the sense that after extracting the syndrome and decoding, the appropriate correction is applied to return the system to the codespace. 
Performing this process in real time is necessary when the quantum circuit contains non-Clifford operations.
Making changes to the syndrome extraction circuit based on measurement results obtained during a QEC cycle is a technique that has also seen use in a number of contexts including repeat-until-success protocols~\cite{brown2023} and detection of hook errors~\cite{Reichardt_2020, ryananderson2021}. And while not in real time, Ref.~\cite{berthusen20242dlocalimplementationquantum} provided a scheme for altering syndrome extraction circuits by occasionally neglecting to measure geometrically nonlocal stabilizer generators. 
In this work, we generalize this idea and present a framework that can `short-circuit' syndrome extraction and skip measuring generators that are unlikely to be useful for decoding. 
This decision process is carried out in real time based on the results of the syndrome measurement from the current QEC cycle.

\begin{figure}[t]
    \centering
    \includegraphics[width=\linewidth]{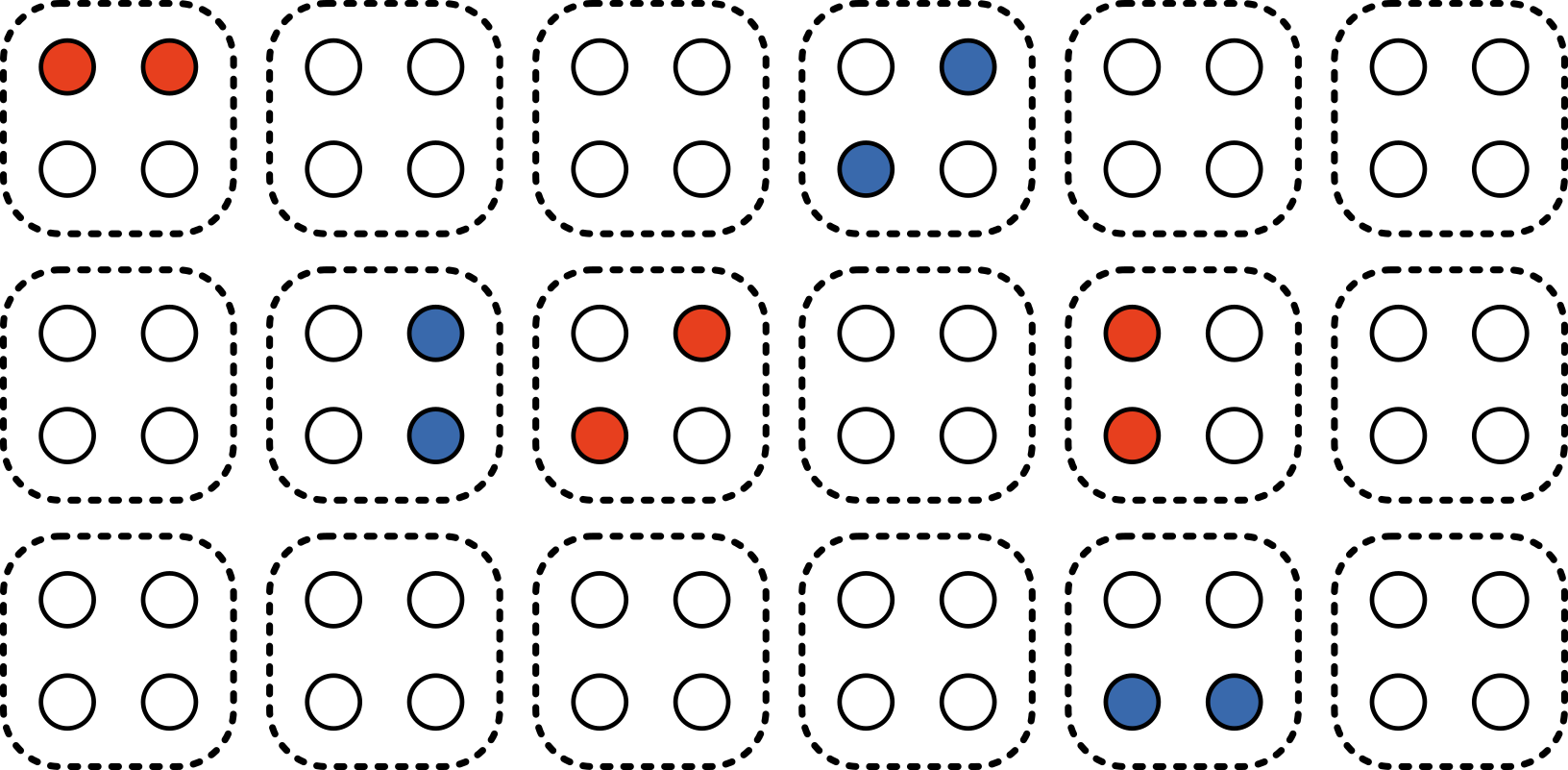}
    \caption{A QED+QEC concatenated code. Physical qubits are first encoded in a small error detecting code, such as the $[[4,2,2]]$ Iceberg code. The logical qubits of many Iceberg code blocks then act as the physical qubits of some high-rate qLDPC code.}
    \label{fig:example-code}
\end{figure}

The main observation motivating the investigation of this adaptive scheme is that in most error correction rounds, much of the syndrome information is unnecessary. For example, in a system without errors, the syndrome does nothing but confirm that the code remains in the correct subspace. 
In this specific case, a high-distance code and an error detecting code would function identically, although the high-distance code would require longer syndrome extraction circuits and introduce more noise.
In general, generators with support on qubits with no errors provide little useful information.
So, how can we cheaply determine approximate error location before committing to an expensive syndrome extraction?
The approach we take in this work is to use a concatenated code where the inner code is a small error detecting code, and the outer code is some high-rate, high-distance quantum low-density parity-check (qLDPC) code. 
The adaptive scheme for this code construction is as follows: in each error correction cycle, the stabilizer generators of the underlying error detecting codes are measured first. 
The resulting syndrome reveals which blocks contain errors; however, since we are using error detecting codes, there is not enough information to correct the error, and so we look for more syndromes by measuring the generators of the outer level of the concatenated code. By the previous argument, we should only measure those generators which provide useful syndrome information. As such, we measure only the stabilizer generators of the outer code which have support on qubits located in QED blocks that have errors in them. 
Generators with no overlapping support are assumed to yield a $+1$ measurement result.

Our adaptive syndrome extraction scheme can be interpreted as a quantum weight reduction protocol in the space-time picture.
Hastings first introduced quantum weight reduction as a toolkit to transform QECCs as to reduce the number of qubits that each checks acts on and the number of checks that each qubit is involved in~\cite{hastings2016weight, hastings2023quantumweightreduction}. 
Reducing the check weights and the qubit connectivity helps to limit the propagation of errors in the QECC as well as reduce the syndrome extraction circuit depth.
Prior to Hastings' work, quantum weight reduction has already been studied in the form of subsystem codes~\cite{poulin2005stabilizer, bacon2006operator, bombin2010topological, baspin2024wire}.
Since then, other works have considered weight reduction in the space-time picture of error correction~\cite{bacon2017sparse, hastings2021dynamically, gottesman2022opportunities, delfosse2023spacetime, de2024dynamical}.
In our scheme, we consider a single space-time block to contain $T$ rounds of syndrome extraction.
By choosing to measure generators of the outer code adaptively, each physical qubit of our concatenated code is now involved in fewer checks ``on average''.
Similarly, the generators of the concatenated code act on fewer qubits ``on average'' in a single space-time block.
While the concatenation scheme increases the actual weight of the checks, we can ensure an overall weight reduction in the space-time picture as long as the outer code generators are measured sufficiently infrequently.

We provide a concrete example of adaptive syndrome extraction facilitated by code concatenation by considering a quantum code consisting of the $[[4,2,2]]$ Iceberg code~\cite{Steane_1996, gottesman1997stabilizercodesquantumerror, Self_2024} concatenated with a hypergraph product (HGP) code~\cite{Tillich_2014}. We show that these $[[4,2,2]]$-concatenated HGP codes function as better quantum memories than non-adaptive and non-concatenated codes, while using fewer CNOT gates and physical qubits.
For these codes, we also provide fault-tolerant logical gadgets with which we can obtain universal logical computation. 
In particular, we adapt the grid Pauli product measurement (GPPM) scheme of Ref.~\cite{xu2024fast} and the single-shot state preparation scheme of Ref.~\cite{hong2024single} to our concatenated codes, allowing us to obtain the same space-time cost of logical computation as non-concatenated HGP codes.

The rest of the work is structured as follows. In Section~\ref{sec:background} we give an introduction to quantum error correction and describe the constructions of the codes used in this work. Section~\ref{sec:adaptive} motivates adaptive syndrome extraction and details how it may provide performance boosts as well as reductions in circuit depth. 
In Section~\ref{sec:canonical} we provide an explicit concatenation structure for $[[4,2,2]]$-concatenated HGP codes that facilitates convenient logical gates and improved numerical performance. 
We then show in Section~\ref{sec:gates} that we can achieve universal logical computation with $[[4,2,2]]$-concatenated HGP codes, with the details deferred to Appendix~\ref{sec:clifford_logical_ops_concatenated_HGP_code}. 
In Section~\ref{sec:numerics} we provide circuit-level simulations comparing the performance of the adaptive scheme with non-adaptive syndrome extraction as well as non-concatenated HGP codes.  
We conclude in Sections~\ref{sec:applications}-\ref{sec:discussion} with a discussion on potential specialized applications for the adaptive scheme, the implications of this work, and potential follow-up research.

\section{Background}
\label{sec:background}

\subsection{Quantum error correction}

An $[n,k,d]$ classical error correcting code $\cC$ is a $k$-dimensional subspace of the full $n$-dimensional vector space of $n$-bit strings. Vectors in this subspace are called codewords. The code $\cC$ can be defined by its parity check matrix (pcm) $H$, where the codewords satisfy $Hv = 0$, that is, $v \in \ker H$. The distance $d$ is the minimum Hamming weight of a codeword.
We say that a code is $(\Delta_V, \Delta_C)$-LDPC if the weight of every column and row of $H$ is bounded by $\Delta_V, \Delta_C \in O(1)$, respectively. 
We can also represent the code $\cC$ with its Tanner graph, a bipartite graph $G = (V \sqcup C, E)$ whose biadjacency matrix is $H$.

An $[[n,k,d]]$ quantum error correcting code (QECC) is a $2^k$-dimensional subspace of the full $2^n$-dimensional Hilbert space of $n$ qubits.
Stabilizer codes~\cite{gottesman1997stabilizercodesquantumerror} are a class of quantum error correcting codes defined to be the joint +1-eigenspace of their stabilizer $\cS$, an Abelian subgroup of the Pauli group on $n$ qubits, $\cP_n$. For an $[[n,k,d]]$ stabilizer code, $\cS$ is fully specified by $m = n-k$ stabilizer generators. The distance of a stabilizer code the minimum Hamming weight of a nontrivial logical operator, that is, Pauli operators which commute with everything in $\cS$ but are not in $\cS$ themselves.
A stabilizer code is said to be a Calderbank-Shor-Steane (CSS)~\cite{shor1996_2, steane1996_2} code if each generator is a tensor product of X and I or a tensor product of Z and I. This then allows us to describe a CSS code using two binary pcms, $H_X$ and $H_Z$, denoting the locations of $X$ and $Z$ Pauli operators, respectively. To ensure that the $X$- and $Z$-type generators commute, we require that $H^T_Z H_X = 0$.
An $[[n,k,d]]$ stabilizer code is considered a $(\Delta_V, \Delta_C)$-qLDPC code if both the $X$- and $Z$-type pcms are $(\Delta_V, \Delta_C)$-LDPC.
Stabilizer codes also have a Tanner graph representation: $H_X$ and $H_Z$ can be individually interpreted as biadjacency matrices; alternatively, the two bipartite graphs could be combined in a single graph $G = (V \sqcup C_X \sqcup C_Z, E)$.

With this redundant encoding, we can detect and correct errors in the system by measuring the eigenvalues of the $m$ stabilizer generators. In the absence of errors, all generators will produce a measurement result of +1, whereas generators that anticommute with an error $e$ will produce a measurement result of -1. This list of measurement results, called a syndrome, is then fed to a classical decoder which attempts to output a correction that returns the system to the codespace. 
An $[[n,k,d]]$ QECC is able to correct for errors on up to $\lfloor(d-1)/2\rfloor$ qubits, and it is able to detect errors of weight up to $d-1$. 
Thus, $[[n,k,d=2]]$ codes can correct no errors and can only detect single-qubit errors; as such, they are called error detecting codes.

\subsubsection{Iceberg codes}

Iceberg codes~\cite{Steane_1996, gottesman1997stabilizercodesquantumerror, Self_2024}, also known as quantum parity codes, are a family of CSS codes with parameters $[[n,n-2,2]]$. The two weight-$n$ stabilizer generators, 
\begin{equation}
    \label{eq:iceberg_stabilizers}
    S_Z = \bigotimes_{i\in [n]} Z_i \quad\text{~and~}\quad S_X = \bigotimes_{i \in [n]} X_i,
\end{equation}
allow the code to detect a single bit-flip and phase-flip error. Here, $[n]$ denotes the sequence $\{1,\ldots,n\}$. Technically, any error of odd parity is detectable, but the code is unable to differentiate between the error weights. Its $n-2$ weight-2 logical operators can be defined as 
\begin{equation}
    \label{eq:iceberg_logicals}
    \overline{X}_i = X_1 X_{i+1} \quad\text{~and~}\quad \overline{Z}_i = Z_{i+1} Z_n,
\end{equation}
and it is easily verified that the commutation relations between logical operators are satisfied. 



\subsubsection{Hypergraph product codes}

Hypergraph product (HGP) codes~\cite{Tillich_2014} are CSS codes which are formed from the hypergraph, or homological~\cite{Bravyi2014} product of two (potentially identical) classical codes. 
A HGP code which is formed from two copies of a single classical code $H$ is called a \textit{square} HGP code, $HGP(H,H)$, and has the following pcms:
\begin{align}
    \label{eq:hgppcms}
    H_X &= \left(H \otimes I_{n},\; I_{m} \otimes H^T\right) \\
    \label{eq:hgppcm2}
    H_Z &= \left(I_{n} \otimes H,\;H^T \otimes I_{m}\right).
\end{align}
When $H$ is the full-rank pcm for a binary linear code with parameters $[n,k,d]$, the resulting HGP code, $HGP(H,H)$, has parameters $[[n^2+(n-k)^2,k^2,d]]$. In this work, we exclusively investigate square HGP codes where $H$ is full-rank.
A notable family of HGP codes are those formed from classical expander codes~\cite{Sipser_Spielman_1996}. The resulting HGP codes are deemed \textit{quantum expander codes}~\cite{Leverrier_2015} and have parameters scaling like $[[n, O(n), O(\sqrt{n})]]$.

\subsubsection{Concatenated codes}

Quantum code concatenation~\cite{gottesman1997stabilizercodesquantumerror} is a procedure that takes two stabilizer codes and produces a larger stabilizer code. In a concatenated code, the physical qubits are first encoded in some $[[n_1, k_1, d_1]]$ stabilizer code encoding $k_1$ logical qubits. These $k_1$ logical qubits then serve as the \textit{physical} qubits for an $[[n_2, k_2, d_2]]$ stabilizer code. This process can be repeated arbitrarily many times, with each encoding step generally increasing the error correction properties of the code. In this work, we focus on concatenated codes consisting of two levels of encoding and use the concatenation procedure defined below. Note that the scheme requires that $k_1$ divides $n_2$.

\begin{procedure}
    \label{proc:concat}
    {\normalfont (Concatenation of quantum codes, see e.g., Section 3.5 of Ref.~\cite{gottesman1997stabilizercodesquantumerror})}. Consider stabilizer codes $\cQ_1$ and $\cQ_2$ with parameters $[[n_1, k_1, d_1]]$, $[[n_2, k_2, d_2]]$ and stabilizers $\cS_1$, $\cS_2$, respectively. Then a concatenated code $\cQ$ with parameters $[[n_1n_2/k_1, k_2, d\ge d_1d_2/k_1]]$ can be constructed as follows:

    Divide the $n_2$ qubits into $n_2/k_1$ blocks of $k_1$ qubits, $B(b), b \in \{1, ..., n_2/k_1\}$. Each block of $k_1$ qubits is then encoded into $n_1$ qubits using $\cQ_1$.
    \begin{itemize}
        \item For each generator $M \in \cS_1$ and for each block $B(b)$ of $n_1$ qubits, include in $\cS$ the Pauli $M_b$ acting on the block tensored with the identity on all other blocks.
        \item Also include in $\cS$ every generator $M \in S_2$, where each single-qubit Pauli $P_i$ is replaced with the corresponding logical Pauli operator according to the mapping of $B(b)$.
    \end{itemize}
\end{procedure}

\section{Adaptive syndrome extraction}
\label{sec:adaptive}


Using Procedure~\ref{proc:concat}, we can obtain a concatenated code that is well-suited for the adaptive syndrome extraction scheme. Let $\cQ_2$ be a $[[n_2, k_2, d_2]]$, $(\Delta_V, \Delta_C)$-qLDPC square quantum expander code, and let $\cQ_1$ be a $[[n_1, k_1, d_1]] = [[n_1, n_1-2, 2]]$ Iceberg code such that $n_1-2~|~n_2$. The resulting concatenated code then has parameters $[[n,k,d]] = [[n_1 n_2 / (n_1-2), k_2, d\ge d_1 d_2 / (n_1 -2) ]]$. Out of the $(n_1n_2/(n_1-2))-k_2$ generators, $2(n_2/k_1)$ come from the Iceberg code blocks, while the outer HGP code supplies its own $n_1-k_1$ generators. For Iceberg codes, the distance lower bound provided in Procedure~\ref{proc:concat} can be improved:

\begin{theorem}
    \label{thm:naive_dist}
    Applying Procedure~\ref{proc:concat} with $\cQ_1$ as a $[[n_1, n_1-2, 2]]$ Iceberg code and $\cQ_2$ as a $[[n_2, k_2, d_2]]$ qLDPC code yields a concatenated code $\cQ$ with parameters $[[n_1n_2/(n_1-2), k_2, 2d_2 \ge d \ge d_2]].$
\end{theorem}
\begin{proof}
    First notice that the Iceberg code logical operators, Eq.~\eqref{eq:iceberg_logicals}, overlap in exactly one location, i.e., on qubit 1 for $X$-type logicals and on qubit $n$ for $Z$-type logicals.
    Consider a weight-$w$ $X$-type Pauli operator encoded in a $[[w+2, w, 2]]$ Iceberg code. Due to the structure of the Iceberg logical operators, the encoded operator will have weight $w$ if $w$ is even; otherwise it will have weight $w+1$.
    The maximum distance of $2d_2$ is hence achieved when the $d_2$ Pauli operators coming from the HGP logical operator fall into $d_2$ distinct Iceberg code blocks, whereas the minimum distance of $d_2$ is attained if $d_2$ is even and all Pauli operators are in the same Iceberg block.    
\end{proof}
By carefully choosing the assignment of qLDPC physical qubits to Iceberg logical qubits, it may be more likely to achieve the full distance of $2d_2$, see Section~\ref{sec:assignment}.
\begin{figure*}[t]
    \centering
    \includegraphics[width=\linewidth]{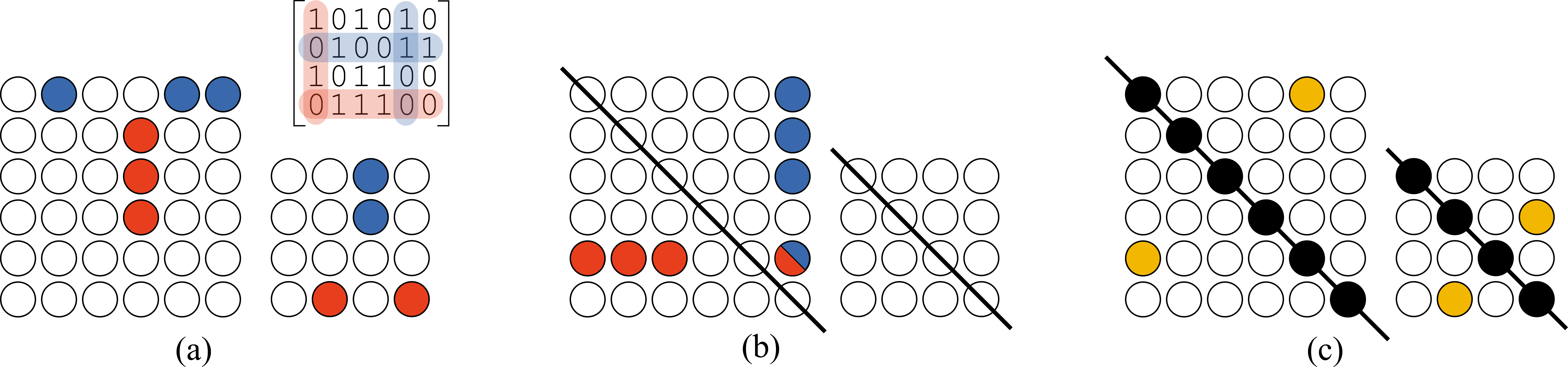}
    \caption{Partitioning the qubits of a $[[52,4,4]]$ HGP code formed from two copies of a $[6,2,4]$ classical code. (a) $Z$-type generator $S_Z(1,3)$ (blue) and $X$-type generator $S_X(4,4)$ (red) obtained by taking rows and columns of the (transpose) parity check matrix of the classical code. (b) Canonical logical operators $\overline{Z}_2$ and $\overline{X}_2$. The qubit highlighted both blue and red is the pivot qubit for that logical pair. (c) Diagonal qubits (black) and twin qubits (yellow). }
    \label{fig:canonical}
\end{figure*}

Briefly, the adaptive syndrome extraction scheme consists of two stages. The first stage measures the stabilizer generators of the inner Iceberg code blocks. Depending on the obtained measurement results, a potentially submaximal set of generators from the outer quantum expander code are measured. The remaining, unmeasured generators are assumed to yield a trivial syndrome. Decoding can then proceed in the normal manner.
Due to the fact that the syndromes of the inner code and outer code are correlated, the total number of generators that are measured can be much less than the number of generators for the concatenated code, especially in the low error rate regime.

Let $\cS_{IB}$ be the generators of the concatenated code coming from the Iceberg blocks, and let $\cS_{HGP}$ be the generators of the concatenated code coming from the HGP code. We can then restrict the full syndrome to these subsets $\sigma_{\cS_{IB}} \in \F_2^{2(n_2/k_1)} \subset \sigma$ or to a single generator, e.g. $\sigma_{\cS_{IB}}^{(i)} \in \{0, 1\}$. 
For any Pauli operator $L$, we denote its support $\text{supp}(L) \subset [n]$ as the set of qubits on which $L$ acts nontrivially. For every $S \in \cS_{IB}, ~|\text{supp}(S)| = n_1$, and for every $S \in \cS_{HGP}, ~|\text{supp}(S)| \le 2 \Delta_C$.

As the simplest example, consider a single $X$ error on the $i$th Iceberg block: assuming no measurement errors, the resulting syndrome when restricted to $\cS_{IB}$ is $\sigma_{\cS_{IB}} = \delta_{ij}$, i.e. all zeros with a single one in the $i$th bit. 
Considering the generators in $\cS_{HGP}$, we can say that
\begin{equation}
    \text{supp}(\cS_{IB}^{(i)}) \cap \text{supp}(\cS_{HGP}^{(j)}) = \emptyset \longrightarrow \sigma_{\cS_{HGP}}^{(j)} = 0,
\end{equation}
since any generator which does not have support on the qubits of the $i$th Iceberg block will not have support on the error. Only those $\cS_{HGP}^{(j)} \in \cS_{HGP}$ such that
\begin{equation}
    \text{supp}(\cS_{IB}^{(i)}) \cap \text{supp}(\cS_{HGP}^{(j)}) \neq \emptyset
\end{equation}
provide additional syndrome information which helps us locate the error within the $i$th Iceberg block. 
Hence in this example, we have determined that it is not necessary to measure all $(n_1n_2/(n_1-2))-k_2$ generators; instead, it suffices to measure only the $2(n_2/k_1)$ Iceberg generators and a small set of overlapping HGP generators.

To generalize, let us first define the following function:
\begin{definition}
    Consider a set of stabilizer generators $\cS$, we define $\varphi(\cS^{(i)})$ to be the set of stabilizer generators which share support with $\cS^{(i)}$    
    \begin{align}
        \label{eq:overlapping}
        \varphi(\cS^{(i)}) = \{ \cS^{(j)} ~|~ \text{supp}(\cS^{(i)}) \cap \text{supp}(\cS^{(j)}) \neq \emptyset \} \\ \nonumber
    \end{align}    
\end{definition}
When the outer HGP code is $(\Delta_V, \Delta_C)$-qLDPC, each Iceberg code block support overlaps with at most $n_1 \cdot \Delta_V$ HGP generators, i.e., $|\varphi(\cS_{IB}^{(i)})| \le n_1 \Delta_V$. Both $n_1$ and $\Delta_V$ are constants, so the number of HGP generators one needs to measure is proportional to the number of Iceberg blocks which have detected an error. Given the syndrome of all Iceberg blocks, $\sigma_{\cS_{IB}}$, the HGP generators which should be measured are then those which share support with Iceberg generators that have indicated an error is present in their block:
\begin{equation}
    \label{eq:hgp_set}
    \{ \varphi(\cS_{IB}^{(i)}) ~|~ \sigma_{\cS_{IB}}^{(i)} = 1 \}.
\end{equation}
This adaptive procedure makes it so that we only obtain the syndrome information necessary to diagnose the error, ultimately yielding shorter syndrome extraction circuits and less introduced noise. 
We note that while here we have described the outer qLDPC code as a HGP code, any other qLDPC code would function similarly.


\section{Canonical logical basis for Hypergraph Product Codes}
\label{sec:canonical}

The physical qubits of HGP($H,H$) can be arranged into two square grids of size $n \times n$ and $m \times m$, see Fig.~\ref{fig:canonical}. In this representation, the stabilizer generators and logical operators have a convenient geometric structure. For each physical qubit in the code, we assign a triplet $(i,j,L)$ or $(k,\ell,R)$ where $1 \le i,j \le n$ and $1 \le k,\ell \le m$. Here, $i~(k)$ denotes the row coordinate, $j~(\ell)$ denotes the column coordinate, and $L~(R)$ specifies whether the qubit is in the left $n \times n$ sector or the right $m \times m$ sector. We denote the physical qubits on the principal diagonal of each sector, i.e. $(r,r,\bullet)$, as \textit{diagonal} qubits. Additionally, for an $(r,c,\bullet)$ qubit, we identify the $(c,r,\bullet)$ qubit as its \textit{mirror} qubit, with the two together considered \textit{twin} qubits, see Fig.~\ref{fig:canonical}(c).

With this layout of the physical qubits, the stabilizer generators of the code have the property that their support is contained in a single row of one sector and a single column of the other. Specifically, $X$-type stabilizers have support on the $k$th row of the $R$ sector and the $j$th column of the $L$ sector. As such, each can be labeled by $S_X(k,j)$. Similarly, a $Z$-type generator labeled by $S_Z(i,\ell)$ has support on the $i$th row of the $L$ sector and the $\ell$th column of the $R$ sector.
The stabilizer generators as expressed in this representation can by constructed by indexing specific rows and columns of the underlying classical parity check matrix $H$~\cite{quintavalle2022reshape}, see Fig.~\ref{fig:canonical}(a), or by reshaping the HGP pcms shown in Eqs.~\eqref{eq:hgppcms},~\eqref{eq:hgppcm2}.

In this layout, we can also obtain a \textit{canonical basis} for hypergraph product code logical operators~\cite{quintavalle2023partitioning}. For these codes, a canonical basis is defined to be a set of logical operators such that they have support contained in a single row or column of a single sector. Additionally, $\overline{X}_i$ and $\overline{Z}_i$ share support on exactly one qubit, whereas $\overline{X}_i$ and $\overline{Z}_j$ for $i \neq j$ do not share any support. The physical qubit on which $\overline{X}_i$ and $\overline{Z}_i$ overlap is deemed the pivot qubit for that logical qubit. Similarly to the physical qubits, the logical qubits are either diagonal logical qubits or part of a twin logical qubit pair depending on the location of its pivot qubit.
We refer to Ref.~\cite{quintavalle2023partitioning} where they provide a method for constructing a canonical basis for square hypergraph product codes.

\subsection{Assignment of $[[4,2,2]]$ logical qubits}
\label{sec:assignment}
\begin{figure}
    \centering
    \includegraphics[width=0.65\linewidth]{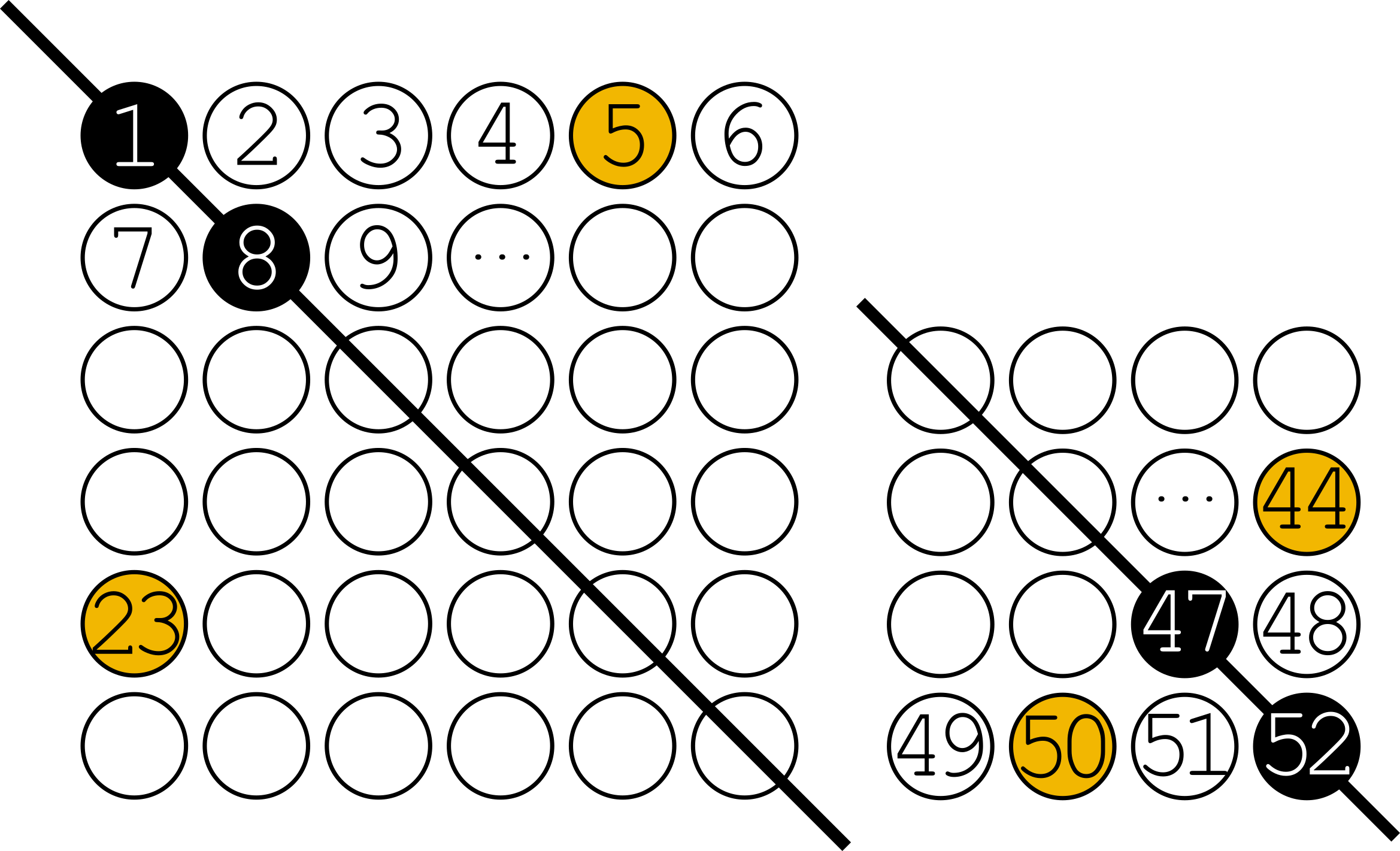}
    \caption{Assignment of physical HGP qubits to logical $[[4,2,2]]$ qubits. Physical HGP qubits are numbered starting from the top left qubit of the $L$ sector, going row-by-row, and finishing on the bottom right qubit of the $R$ sector.
    Consecutive diagonal qubits, e.g. qubits $\{1,8\}$ and $\{47,52\}$ are assigned to the same Iceberg code block. Twin qubits, e.g. qubits $\{5,23\}$ and $\{44,50\}$ are assigned to the same Iceberg block. }
    \label{fig:mapping}
\end{figure}

We now describe a method for assigning the $[[4,2,2]]$ Iceberg code logical qubits to physical qubits of a hypergraph product code which takes advantage of the aforementioned geometric structure of the stabilizer generators and logical operators. Additionally, we will later show that this assignment allows the concatenated code to inherit logical gates from the base HGP code.

The left and right sectors are labeled by the $n^2 + m^2$ physical qubits as shown in Fig.~\ref{fig:mapping}, with the first $n^2$ qubits labeling the left sector, and the remaining $m^2 = (n-k)^2$ qubits labeling the right sector. We then assign the physical qubits to the logical qubits of the $(n^2+m^2)/2$ $[[4,2,2]]$ Iceberg codes as follows: consecutive diagonal qubits, e.g. qubits $\{1,8\}$ and $\{47,52\}$ of Fig.~\ref{fig:mapping}, get mapped to the same Iceberg block. We also map twin qubits, e.g. $\{5,23\}$ and $\{44, 50\}$, to the same Iceberg code block. 

As discussed in the proof of Theorem~\ref{thm:naive_dist}, we obtain suboptimal distances whenever qubits supported on the same HGP logical operator get mapped into the same Iceberg code block. While it is possible for there to exist a low weight logical operator span several rows or columns and have support on both qubits of a single Iceberg code block, this mapping reduces the chance of that happening. 
Indeed, we verified the resulting distances of concatenated codes constructed with this mapping using QDistRnd~\cite{Pryadko_2022} and found that all instances achieved a distance of $2d_2$.
Additionally, using a square HGP code and this mapping ensures symmetric performance between the $X$ and $Z$ basis, where even if $d < 2d_2$, $d_X = d_Z = d$.

\section{Logical Computation for $[[4,2,2]]$-concatenated hypergraph product codes}
\label{sec:gates}
In Ref.~\cite{xu2024fast}, Xu \textit{et al.} developed grid Pauli product measurements (GPPMs) as a framework to perform fast and parallelizable logical computation for homological product codes.
Here, we show that this framework can be adapted for our $[[4,2,2]]$-concatenated HGP codes.
By combining the adapted GPPMs with a set of new logical gadgets that we obtain from the transversal gate sets and fold-transversal gate sets of the $[[4,2,2]]$ Iceberg code and HGP codes, respectively, we are able to show Theorem~\ref{thm:clifford_logical_ops_concatenated} and prove that we can efficiently simulate the full Clifford group on an arbitrary $[[4,2,2]]$-concatenated HGP code block.
In particular, we demonstrate how we can utilize inter- and intra-block CNOT gates, concatenated H-SWAPs, CZ-Ss, and GPPMs to fault-tolerantly simulate a layer of Clifford gates that include $H$, $S$, and intra-block CNOT gates acting on $O(k)$ logical qubits of our $[[4,2,2]]$-concatenated HGP code block with low space and time overhead.
We defer the definitions of the logical gadgets and the proof of Theorem~\ref{thm:clifford_logical_ops_concatenated} to Appendix~\ref{sec:clifford_logical_ops_concatenated_HGP_code}.

\begin{theorem}[Clifford Gates for Concatenated HGP Code]\label{thm:clifford_logical_ops_concatenated}
    A single layer of an ideal Clifford circuit on $k$ logical qubits with $\Theta(k)$ gates consisting of Hadamard, S, and CNOT gates can be simulated on a square HGP code concatenated with the $[[4,2,2]]$ Iceberg code blocks with either one of the following space-time costs:
    \begin{enumerate}
        \item $O(k)$ space and $O\left(k^{3/2}\right)$ time
        \item $O\left(k^{3/2}\right)$ space and $O(k)$ time.
    \end{enumerate}
    If the square HGP code was constructed from one-generator systematic circulant (OGSC) quasi-cyclic base codes, the concatenated code can simulate the layer with either one of the following space-time costs:
    \begin{enumerate}
        \item $O(k)$ space and $O\left(k^{5/4}\right)$ time
        \item $O\left(k^{3/2}\right)$ space and $O\left(k^{3/4}\right)$ time.
    \end{enumerate}
\end{theorem}

Our technical contributions include the following: adapting the GPPM scheme of Ref.~\cite{xu2024fast}, developing fault-tolerant logical gadgets, and proving the amenability of a single-shot state preparation scheme for our $[[4,2,2]]$-concatenated HGP codes.
The main challenge in adapting the GPPM scheme lies in resolving the tension between the punctures introduced in the GPPM scheme and the assignment of Iceberg code logical qubits as part of our code concatenation scheme.

The fault-tolerant logical gadgets have to be formulated with the constraints of both the HGP and Iceberg codes in mind.
In particular, we adapt the logical translation gadget introduced in Ref.~\cite{xu2024fast} for square HGP codes constructed from quasi-cyclic one-generator systematic circulant (OGSC) codes.
We reduce the problem to the permutation of arbitrary single logical qubits between Iceberg code blocks and use measurement-based logical SWAP gates to construct the logical translation gadget for our $[[4,2,2]]$-concatenated HGP code.
We provide explicit constructions for the other fault-tolerant logical gadgets for the case where we do not have the logical translation gadget.
We direct readers to Ref.~\cite{xu2024fast} for the case where the logical translation gadget is available, i.e., the square HGP codes are constructed from quasi-cyclic OGSC codes.  

By bounding the soundness and confinement factors for the $[[4,2,2]]$-concatenated HGP code, we prove that our concatenated HGP code is still amenable to the single-shot state preparation scheme for HGP codes proposed by Hong in Ref.~\cite{hong2024single}, giving us the ability to implement all of our logical gadgets in $O(1)$ logical cycles at the expense of additional space.
This trade-off comes from choosing between a single-shot state preparation scheme and a regular CSS code state preparation scheme and gives rise to the two possible space-time costs for each of the two different logical gadget constructions stated in Theorem~\ref{thm:clifford_logical_ops_concatenated}.

Because we can reproduce the logical gadgets in Ref.~\cite{xu2024fast} for our $[[4,2,2]]$-concatenated HGP code, we can use the GPPM scheme to perform the ``8-to-CCZ'' magic state distillation protocol and implement the non-Clifford gates in parallel.
Combining the magic state distillation and consumption protocol with the Clifford gate simulation, we obtain a fault-tolerant protocol for universal logical computation.

\section{Numerical simulations}
\label{sec:numerics}

To quantify the effectiveness of the adaptive QED + QEC scheme, we perform memory experiment simulations. We provide additional details on the codes, decoders, and simulation in Appendix~\ref{apx:sims}.

\subsection{$[[4,2,2]]$-concatenated decoding}
\label{sec:decoding}

Several other codes utilize concatenation where the lowest level code consists of the $[[4,2,2]]$ code: the $C_4/C_6$ code~\cite{c4c6Knill2005}, the $C_4$/Steane code~\cite{yoshida2024concatenatecodessavequbits}, many-hypercube code~\cite{goto2024manyhypercubecodeshighratequantum}, and the $[[4,2,2]]$-Toric code~\cite{criger2016}.
However, in almost all of these works, error correction is done using Knill-style syndrome extraction~\cite{knill2005}. This method is not compatible with an adaptive syndrome extraction protocol, as there is no opportunity to `short-circuit' the syndrome extraction. We instead stick to Shor-style syndrome extraction~\cite{Shor1996}, and as a consequence we are not able to use hard- and soft-decision decoding~\cite{c4c6Knill2005, goto2013}, symbol-MAP decoding, or level-by-level minimum distance decoding~\cite{goto2024manyhypercubecodeshighratequantum}. 

In this work, we decode the $[[4,2,2]]$-concatenated HGP code by first attempting to address errors in the inner $[[4,2,2]]$ code blocks before decoding the outer HGP code. Given that the inner code is only error-detecting, we are unable to applied tailored corrections; instead, we apply the same correction whenever an Iceberg code block detects an error. As a reminder, the generators and logical operators for the $[[4,2,2]]$ code are shown below.

\begin{align}
    S_X = X_1X_2X_3X_4 & \quad S_Z = Z_1Z_2Z_3Z_4 \\
    \overline{X}_1 = X_1X_2 & \quad \overline{Z}_1 = Z_2Z_4\\
    \overline{X}_2 = X_1X_3 & \quad \overline{Z}_2 = Z_3Z_4 
\end{align}

We now go through the process for correcting $Z$-type errors. When decoding the inner $[[4,2,2]]$ codes, we apply the correction $Z_4$ on each code block that has detected an error. The effect of this correction is:

\begin{align}
    \label{eq:corr1}
    Z_1 &\rightarrow Z_4 \rightarrow Z_1Z_4 = \overline{Z}_1 \overline{Z}_2 \cdot S_Z \\
    Z_2 &\rightarrow Z_4 \rightarrow Z_2Z_4 = \overline{Z}_1 \\
    Z_3 &\rightarrow Z_4 \rightarrow Z_3Z_4 = \overline{Z}_2 \\
    \label{eq:corr4}
    Z_4 &\rightarrow Z_4 \rightarrow I     
\end{align}

Hence we have successfully turned every single qubit physical error into zero, one, or two logical errors on the outer HGP code. Additionally, since $Z_4$ does not overlap with the $X$-type logical operators, it ensures that 1. the syndrome is consistent between decoding stages, and 2. errant corrections do not (immediately) affect the logical state of the $[[4,2,2]]$-concatenated HGP code. This provides a built-in robustness to measurement errors on the Iceberg code blocks.
The second decoding stage treats the residual logical errors on the inner $[[4,2,2]]$ code blocks as physical errors on the outer HGP code. Decoding the outer code provides a \textit{logical} correction, which we then translate into a physical correction to apply to the system. The decoding process for $X$-type errors is analogous, but we instead apply $X_1$ as the correction for the inner Iceberg codes. 

It has been shown that the decoding performance of concatenated codes can be improved by passing soft information between levels~\cite{Poulin_2006, meister2024efficientsoftoutputdecoderssurface}. In particular, observing a nonzero syndrome for an Iceberg block informs us that an error is much more likely to be located on those four physical qubits. After the application of the initial correction, Eqs.~\eqref{eq:corr1}-\eqref{eq:corr4}, this translates to an increased probability of there being a logical error on the two corresponding logical qubits.  
To take advantage of this extra information, we can update the channel probabilities given to the HGP decoder. Out of the four error scenarios, both $\overline{Z}_1$ and $\overline{Z}_2$ see an error in two of them; thus we can set the probability of having a logical error at 0.5 for any logical qubit of an Iceberg block detecting an error. The probabilities for all other logical qubits are set to $p^2$. 
The decoding performance could potentially be improved further by using an overlapping window~\cite{huang2024} or circuit-level decoder~\cite{wang2011, xu2023constantoverhead, gong2024lowlatencyiterativedecodingqldpc}. We leave it to future work to adapt the concatenated decoding scheme presented here to these methods.

\subsection{Memory experiments}

The memory experiments consist of a noiseless initialization of the code in the joint $\ket{\overline{0...0}}$ state, $r$ rounds of noisy syndrome extraction and correction, and a destructive measurement of the data qubits. From the destructive measurement, a final (noiseless) syndrome as well as the value of the logical qubits can be constructed. Decoding is considered a success if the logical qubits are measured in the $\ket{\overline{0...0}}$ state, otherwise we say a logical error has occurred. The logical error rate, $p_L$, is the fraction of samples where at least one logical qubit experiences a logical error. We also consider the logical error rate per round,
\begin{equation}
    \label{eq:ler_per_round}
    \epsilon_L = 1 - (1 - p_L(r))^{1/r},
\end{equation}
where $p_L(r)$ is the observed logical error rate at round $r$.

\begin{figure}
    \centering
    \includegraphics[width=0.7\linewidth]{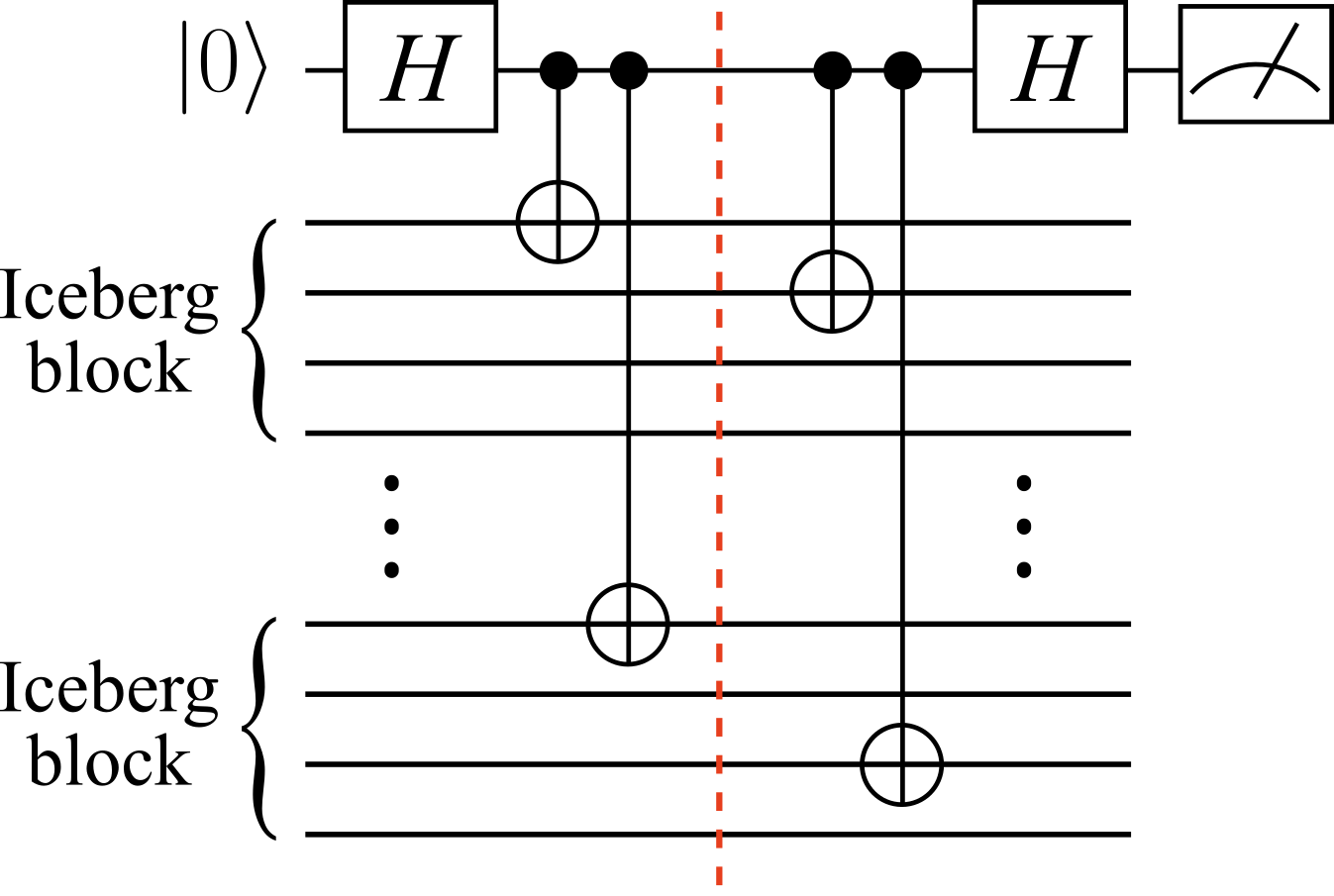}
    \caption{Circuit to measure an $X$-type concatenated HGP generator. A hook error that occurs in the middle of the circuit will propagate to a large number of detectable errors. }
    \label{fig:hook_errors}
\end{figure}

Our simulations use a circuit-level noise model that is parameterized by a noise strength $p$ and consists of the following: One-qubit Clifford gates are followed by a one-qubit depolarizing noise channel of strength $p/10$; two-qubit Clifford gates are followed by a two-qubit depolarizing noise channel of strength $p$; measurement results are flipped with probability $p$; qubit reset operations have probability $p$ of preparing the $\ket{1}$ state instead of the $\ket{0}$ state; and, idle qubits have a one-qubit depolarization noise channel of strength $p/10$ applied to them. This noise model, where single-qubit and idling errors are reduced, is physically relevant for ion-trap~\cite{ryananderson2021, moses2023, dasilva2024, reichardt2024} and neutral-atom~\cite{xu2023constantoverhead, Bluvstein_2023, reichardt2024_2} quantum computers.

To measure the stabilizer generators of the Iceberg code, we use the fault-tolerant syndrome extraction circuit from Ref.~\cite{Self_2024}. The concatenated and non-concatenated HGP generators are measured using a bare ancilla qubit and a circuit derived from an edge coloration of the Tanner graph~\cite{delfosse2021boundsstabilizermeasurementcircuits}, on which we
can apply accurate idling errors. 
Whereas HGP codes enjoy a robustness to hook errors~\cite{Dennis_2002} regardless of the CNOT gate ordering during syndrome extraction~\cite{manes2023, tan2024}, we have to be careful with the CNOT gate ordering for the concatenated HGP generators. For certain orderings, a single error in the middle of the circuit will propagate to a large number of logical Iceberg errors. This is especially detrimental in the adaptive scheme as these logical errors will not be detected by the Iceberg blocks, and so the overlapping HGP generators will not be measured.
As such, we employ the circuit shown in Fig.~\ref{fig:hook_errors}: By measuring the first qubit of each Iceberg logical before the second, we ensure that an error in the middle of the circuit will propagate to a large number of detectable errors.

\begin{figure*}[t]
    \centering
    \includegraphics[width=\linewidth]{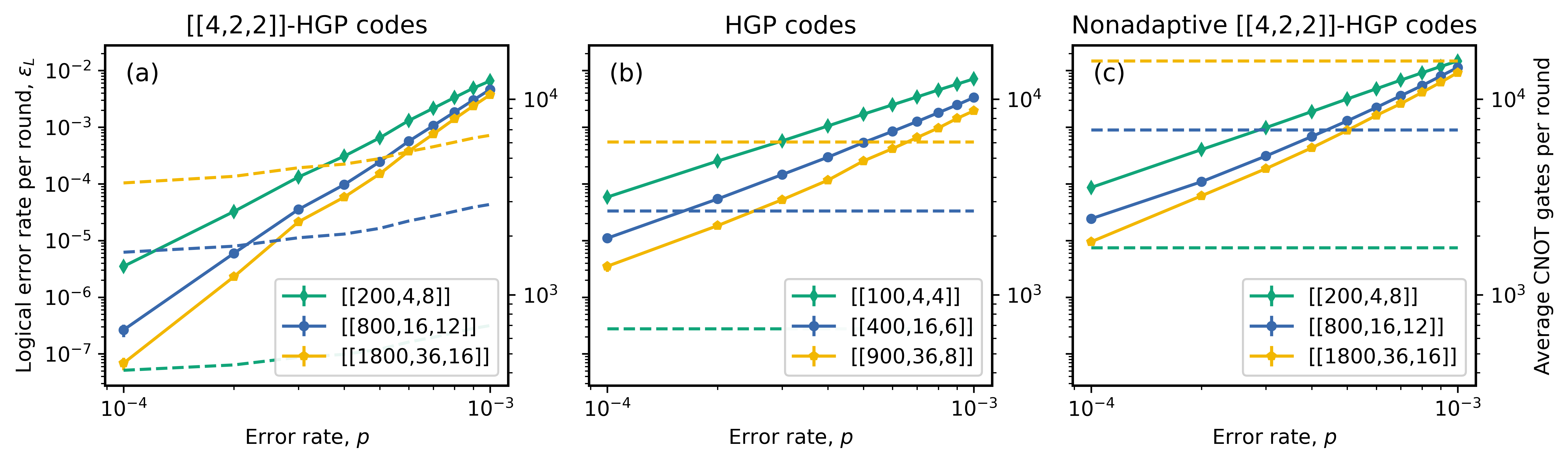}
    \caption{Logical error rate per round, $\epsilon_L$ as a function of the physical error rate, $p$. (a) $[[4,2,2]]$-concatenated HGP codes where the syndrome extraction circuit is adaptive. (b) Non-concatenated HPG codes using normal syndrome extraction. (c) $[[4,2,2]]$-concatenated HGP codes where the syndrome extraction circuit is not adaptive; that is, every generator is measured in every round. Error bars indicate the standard deviation on $\epsilon_L$, Eq.~\eqref{eq:ler_per_round_error}. The dashed lines and right y-axis show the number of CNOT gates in the syndrome extraction circuit averaged over the $r = 100$ rounds. }
    \label{fig:memory_circuit}
\end{figure*}

We initially found that the $[[4,2,2]]$-concatenated HGP codes implemented using the adaptive scheme were not single-shot~\cite{Bomb_n_2015}, see Fig.~\ref{fig:ft_ler_per_round}(b). Whereas the HGP codes exhibit a logical error rate per round that stabilizes with increasing $r$, we see no such behavior out of the concatenated codes. On the contrary, we no longer observe a pseudothreshold when increasing $r$.   
The main cause of this behavior is most likely an accumulation of logical errors in the Iceberg code blocks;
in a normal concatenated QEC scheme this is preventable as these logical errors are detected by the concatenated HGP generators in subsequent rounds. In the adaptive scheme, however, the logical errors are not detected, and the corresponding HGP generators are not measured, hence leading to an overwhelming build-up of errors.
To alleviate this, we measure the entire set of generators every $r'$ rounds to pick up residual logical errors. This \textit{unmasking}~\cite{Berthusen_2023} frequency is set at every 10 rounds for $p=0.1\%$ and scales inversely with $p$.

With this change we now observe single-shot performance out of the concatenated codes, with the logical error rate per round stabilizing by approximately round $r = 100$, see Fig.~\ref{fig:ft_ler_per_round}(a). 
As such, we perform $r = 100$ rounds of QEC in the following memory experiments. 
After each syndrome extraction, we attempt to guess a correction with only the syndrome information from that round. The specific decoder used depends on the current QEC round.
On rounds $r < 100$, we use only the Belief Propagation (BP)~\cite{mackay1997} decoder. Using the noiseless syndrome constructed from the destructive measurement of the data qubits, we instead use belief propagation and localized statics decoding (BP-LSD)~\cite{Roffe_LDPC_Python_tools_2022, hillmann2024localizedstatisticsdecodingparallel} to apply a final correction before verifying the logical state.
The decoding process for the $[[4,2,2]]$-concatenated HGP codes is the same except for in that we follow Sec.~\ref{sec:decoding} to first convert physical errors into logical errors.

\subsubsection{Quantum expander codes}

We first investigate a family of quantum expander codes~\cite{Leverrier_2015}, square hypergraph product codes where $H$ is a random $(3,4)$-regular classical code. 
Fig.~\ref{fig:memory_circuit} displays the logical error rate per round $\epsilon_L$ as a function of $p$ for the HGP codes and the corresponding $[[4,2,2]]$-concatenated HGP codes.
The dashed lines and right y-axis of Fig.~\ref{fig:memory_circuit} show the number of CNOT gates in the syndrome extraction circuit averaged over the $r = 100$ rounds.

At high error rates, the adaptive scheme is outperformed by the non-concatenated HGP codes.
Additionally, we find a lower pseudothreshold when using the adaptive scheme, see Fig.~\ref{fig:pseudothreshold}. 
However, we see over an order of magnitude performance improvement using the adaptive scheme at low error rates. This is the regime in which we expect the $[[4,2,2]]$-concatenated codes and the adaptive scheme to work the best; when the errors are infrequent, the overlapping HGP generators are not measured as often, hence leading to fewer applied gates and less induced circuit noise.
The non-adaptive concatenated codes perform slightly worse than their non-concatenated counterparts and significantly worse than their adaptive counterparts while using twice as many qubits and more CNOT gates. 

All Iceberg-concatenated qLDPC codes have a minimum CNOT gates per round of $2n$, i.e. just measuring the two weight-$n$ Iceberg generators. 
For $(3,4)$-regular quantum expander codes, which are $(4,7)$-qLDPC, this means we achieve a maximum $\sim 2 \times$ reduction in CNOT gate count. 
Using qLDPC codes with higher weight generators would provide better potential savings, but a large stabilizer weight poses problems for physical implementations~\cite{hastings2023quantumweightreduction, tan2024}. 
Note that despite the reduced CNOT count, the circuit depth is often not lower for the $[[4,2,2]]$-concatenated codes. With all-to-all connectivity and an edge coloring of the Tanner graph, the syndrome extraction for a $(\Delta_V, \Delta_C)$-qLDPC code can be done in depth $2\Delta_C +3$~\cite{delfosse2021boundsstabilizermeasurementcircuits}. 
The corresponding $[[4,2,2]]$-concatenated codes are $(2\Delta_V+1, 2\Delta_C)$-qLDPC and require a circuit depth of $4\Delta_C+8$ to measure all stabilizer generators. The additional 5 layers come from the adaptive scheme and the separate measurements of the Iceberg generators.
Hence except on rounds where no Iceberg blocks detect an error, the concatenated codes require deeper syndrome extraction circuits. 
This is true even when a single Iceberg block detects an error, in which case the remaining blocks must sit idle until syndrome extraction has completed; however it may be feasible for an intelligent, real-time circuit compiler~\cite{fang2023dynamicquantumcircuitcompilation} to recognize this and perform gates on the unused blocks at the same time.

However, this circuit-depth analysis only holds on quantum computers which can implement arbitrarily many two-qubit gates in parallel: for architectures unable to achieve maximum parallelization, the CNOT count directly impacts the QEC cycle time. 
In particular, trapped-ion quantum computers typically have many more qubits than gating abilities, with some architectures containing a small number of zones in which gates can be performed~\cite{Pino2021, moses2023} and others requiring completely serial execution. 
An additional scenario in which CNOT count might play an outsized role is if we needed to first compile to a device without all-to-all connectivity, where the implementation of nonlocal two-qubit gates would introduce overhead and reduce parallelization.
We further discuss this latter setting in Sec.~\ref{sec:2dlocal}.

\subsubsection{La-cross codes}

\begin{figure}[t]
    \centering
    \includegraphics[width=\linewidth]{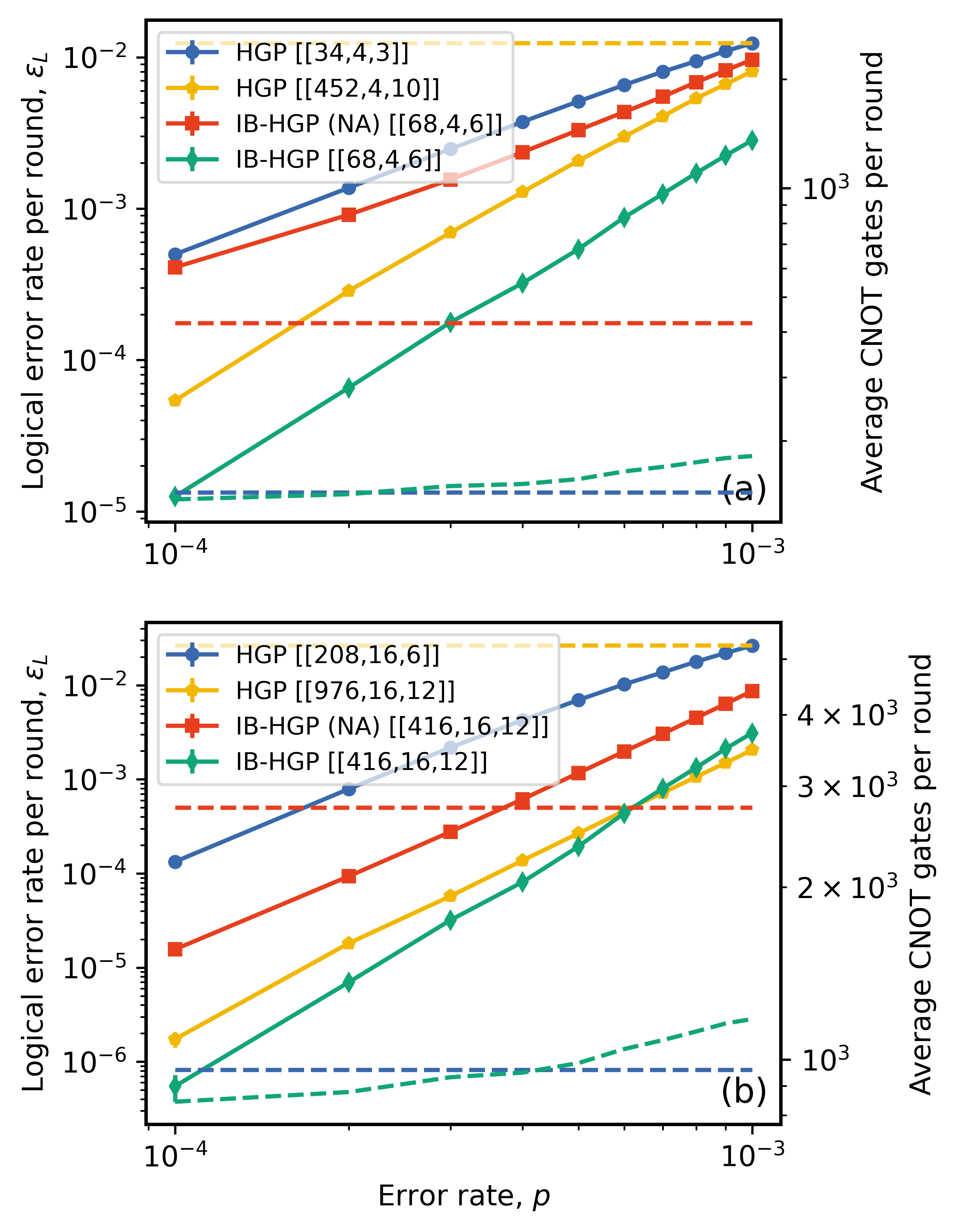}
    \caption{Logical error rate per round $\epsilon_L$ as a function of the physical error rate, $p$ for a family of (a) $z = 2$ La-cross codes and the corresponding $[[4,2,2]]$-concatenated La-cross codes (b) $z = 4$ (concatenated) La-cross codes. The $[[4,2,2]]$-concatenated codes are labeled by IB-HGP. We also show a concatenated code which was measured using a non-adaptive syndrome extraction circuit, labeled as IB-HGP (NA). The dashed lines and right y-axis show the number of CNOT gates in the syndrome extraction circuit averaged over the $r = 100$ rounds. }
    \label{fig:lacross}
\end{figure}

We also investigate several families of small La-cross codes~\cite{pecorari2024highratequantumldpccodes}, hypergraph product codes where $H$ is a circulant matrix. See Appendix~\ref{apx:lacross} for the specific code construction. Fig.~\ref{fig:lacross} displays the logical error rate per round $\epsilon_L$ as a function of the physical error rate $p$ for the La-cross codes and the corresponding $[[4,2,2]]$-concatenated La-cross codes. 

\setlength{\tabcolsep}{0.4em} 
\renewcommand{\arraystretch}{1.2}
\begin{table}
    \centering
    \begin{tabular}{||c|cc|ccc||}
        \toprule
        Code     & $\overline{q}$ & $\overline{w}$  & $p$ & $\overline{q}_{adapt}$ & 
        $\overline{w}_{adapt}$  \\ 
        \midrule
        $[[100,4,4]]$  & 6.72 & 7 &-&-&-   \\
        $[[400,16,6]]$ & 6.72 & 7 &-&-&-   \\
        \hline
        $[[200,4,8]]$  & 8.72 & 8.90 & $10^{-3}$ & 3.62 & 5.76 \\
        -  & - & - & $10^{-4}$ & 2.06 & 4.08 \\
        \hline \hline
        $[[208,16,6]]$ & 4.61 & 5.0 &-&-&- \\
        $[[400,16,8]]$ & 5.04 & 5.25 &-&-&- \\
        \hline
        $[[416,16,12]]$ & 6.62 & 6.88 & $10^{-3}$ & 2.83 & 4.86 \\
        -  & - & - & $10^{-4}$ & 2.03 & 4.04 \\
        \bottomrule
    \end{tabular}
    \caption[Average check weight, $\overline{w}$, and average number of generators each qubit participates in, $\overline{q}$, for several HGP and concatenated-HGP codes.]{Average check weight, $\overline{w}$, and average number of generators each qubit participates in, $\overline{q}$, for several HGP and $[[4,2,2]]$-concatenated HGP codes. We also show average check weight $\overline{w}_{adapt}$ and average qubit degree $\overline{q}_{adapt}$ obtained from averaging over $r=100$ rounds of adaptive syndrome extraction at two physical error rates, $p$.}
    \label{table:qwr}
\end{table}

For these codes, the adaptive scheme outperforms the non-concatenated codes by nearly an order of magnitude across the entire range of physical error rates. 
We were not able to scale up the $k=4$ non-concatenated La-cross codes to match the performance of the $[[4,2,2]]$-concatenated codes. Larger, $k=16$ non-concatenated La-cross were able to nearly match the performance of the concatenated codes and the adaptive scheme; however, they required over twice as many physical qubits and $\sim 5 \times$ more CNOT gates.
We again see that the non-adaptive codes perform comparatively worse, indicating that the adaptive scheme itself provides benefits to the logical error rate. 

\begin{figure*}
    \centering
    \includegraphics[width=\linewidth]{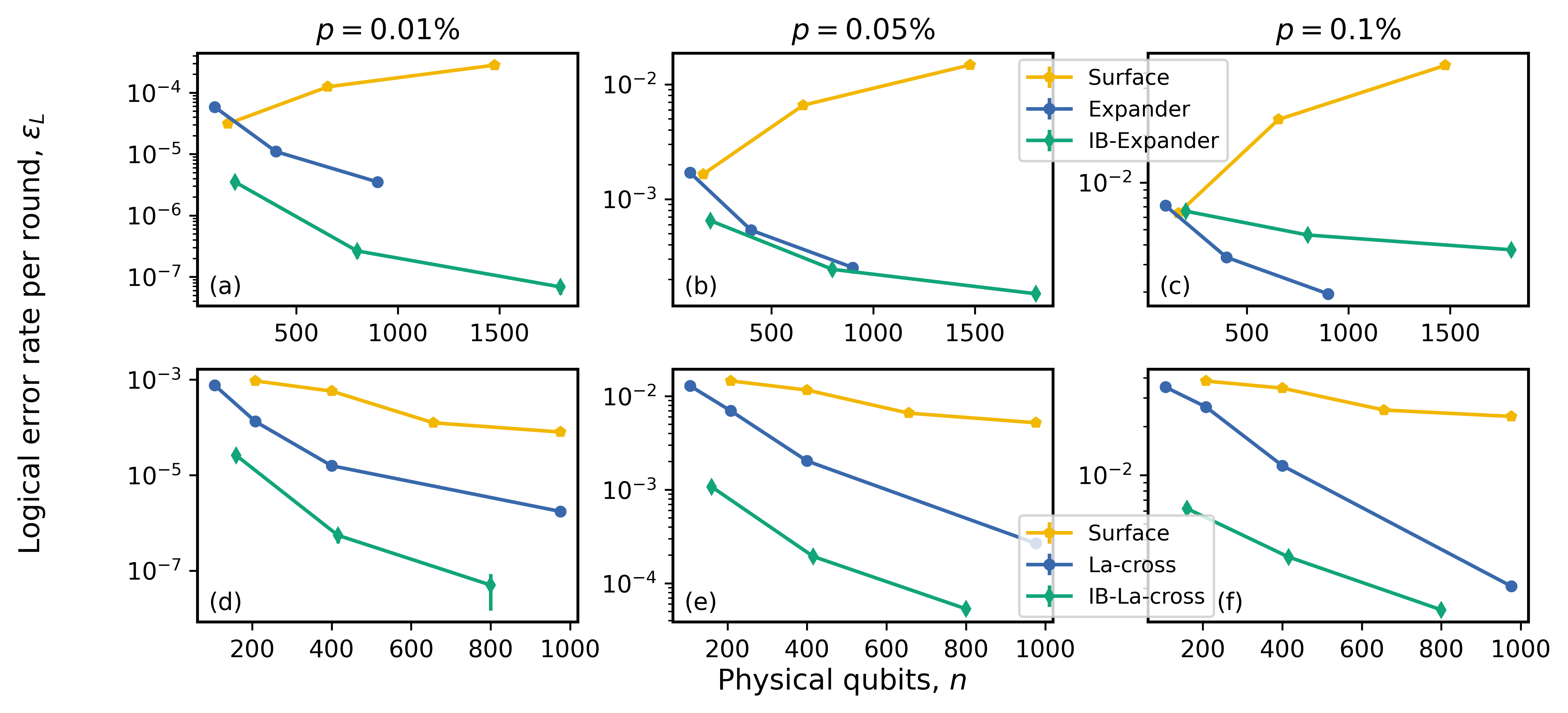}
    \caption{Logical error rate per round, $\epsilon_L$ as a function of the number of physical qubits. 
    Panels (a)-(c) display quantum expander codes and $[[4,2,2]]$-concatenated quantum expander codes (labeled IB-Expander). As a point of reference, we also include the performance of multiple surface code patches encoding $k = 4,~16,~36$ logical qubits. Using $4 \times [[41,1,5]], ~16\times[[41,1,5]]$, and $36\times[[41,1,5]]$ yield blocklengths similar to the IB-expander codes.
    Panels (d)-(f) display $z=4$ La-cross codes and $[[4,2,2]]$-concatenated La-cross codes (labeled IB-La-cross). Also shown for comparison are $k=16$ copies of the surface code, the distances of which were chosen to yield codes that approximately matched the blocklength of the IB-La-cross codes. Specifically, we used 16 patches of the $[[13,1,3]],~[[25,1,4]],~[[41,1,5]$, and $[[61,1,6]]$ surface codes, respectively.}
    \label{fig:overhead}
\end{figure*}

To illustrate how adaptive syndrome extraction can function as a quantum weight reduction scheme in the space-time picture, we display in Table~\ref{table:qwr} the average check weight $\overline{w}$ and the average number of generators each qubit participates in $\overline{q}$. 
When implemented using non-adaptive syndrome extraction, the $[[4,2,2]]$-concatenated HGP codes have higher weight checks and its qubits are involved in more checks than the corresponding non-concatenated codes.
However, if we use adaptive syndrome extraction and average over many rounds, we find large reductions in $\overline{q}$ and $\overline{w}$.
At low error rates, only the two Iceberg generators are measured, yielding $\overline{q}_{adapt}$ and $\overline{w}_{adapt}$ that approach 2 and 4, respectively. This is even better than notable topological codes like the surface code~\cite{bravyi1998quantum, Kitaev_2003}, with $\overline{q} = \overline{w} = 4$ or the color code~\cite{Bombin_2006}, with $\overline{q} = \overline{w} = 6$. Note, however, that the maximum stabilizer weight for the $[[4,2,2]]$-concatenated codes is still much larger at 14; it is just that when averaged over many syndrome extraction rounds, a majority of the generators that get measured come from the Iceberg code blocks.

To further emphasize the potential savings achievable by using the adaptive scheme, we plot in Fig.~\ref{fig:overhead} the logical error rate per round $\epsilon_L$ versus the number of physical qubits at three different physical error rates, $p = ~0.01\%, ~0.05\%, ~0.1\%$. The data shown here is the same as that displayed in Fig.~\ref{fig:memory_circuit} and Fig.~\ref{fig:lacross}.
Panels (a)-(c) display the performance of quantum expander codes. In this representation, it is clear that the adaptive scheme does not provide any benefits above $p = 0.05\%$; however, at $p = 0.01\%$ we observe a tradeoff of fewer logical qubits but lower logical error rates.
As a point of reference, we also include the performance of $r = 100$ rounds of syndrome extraction for surface codes encoding $k = 4, ~16$, and $36$ logical qubits while using approximately the same number of physical qubits as the $[[4,2,2]]$-concatenated quantum expander codes.
To obtain the logical error rate for $k$ surface code patches, we calculate
\begin{equation}
    p_{L,k} = 1 - (1 - p_{L,1})^k,
\end{equation}
where $p_{L,1}$ is the logical error rate for a single patch. We can then compute $\epsilon_L$ using Eq.~\eqref{eq:ler_per_round}.
Quantum expander codes have a constant encoding rate which means that the surface codes are not allowed to scale, hence leading to an increasing $\epsilon_L$ when encoding more logical qubits. 
Panels (d)-(f) display the performance of $z=4$ La-cross codes encoding $k=16$ logical qubits. Here we now see the $[[4,2,2]]$-concatenated codes and the adaptive scheme producing lower logical error rates while also using fewer physical qubits across the entire range of physical error rates. Also shown for comparison are $k=16$ copies of the surface code, the distances of which were chosen to yield codes that approximately matched the blocklength of the $[[4,2,2]]$-concatenated La-cross codes. We note that La-cross codes were already shown to outperform surface codes in Ref.~\cite{pecorari2024highratequantumldpccodes}.


\section{Applications}
\label{sec:applications}

In this section, we propose specific use-cases where $[[4,2,2]]$-concatenated codes and the adaptive scheme might further outperform their non-concatenated, non-adaptive counterparts.

\subsection{2D-local implementations of nonlocal codes}
\label{sec:2dlocal}

One potential application of $[[4,2,2]]$-concatenated codes and the adaptive scheme is in implementations of nonlocal qLDPC codes on restricted architectures. 
There is a close relationship between \textit{locality} and the parameters and properties of a quantum error correcting code.
After embedding the Tanner graph of a quantum code into $\mathbb{Z}^2$, a code is considered 2D-local if the syndrome extraction circuits can be implemented using only gates between neighboring qubits. Codes such as the surface code~\cite{bravyi1998quantum, Kitaev_2003} and color code~\cite{Bombin_2006} are 2D-local; as such, they are leading candidates for architectures like superconducting qubits which are currently limited to local gates only.
However, with ease of implementation comes restrictions in the asymptotic rate and distance of the code, as was shown in Refs.~\cite{Bravyi_Terhal_2009, Bravyi_Poulin_Terhal_2010}. 
To get around these bounds, long-range interactions must be introduced~\cite{Baspin_2022}, posing difficulties for hardware without long-range gates. Several recent works have provided methods for implementing nonlocal qLDPC codes on 2D-local architectures~\cite{delfosse2021boundsstabilizermeasurementcircuits, pattison2023, choe2024faulttolerantlyrealizequantumcircuit, berthusen20242dlocalimplementationquantum}, but it remains an open question as to the most effective methods and codes for this task.

First notice that $[[4,2,2]]$-concatenated codes have an embedding into $\mathbb{Z}^2$ which guarantees \textit{some} locality regardless of the outer code choice, see Fig.~\ref{fig:example-code}. 
Placing an ancilla qubit in each Iceberg codeblock would allow all blocks to be measured with only 2D-local interactions.
This, along with the adaptive scheme, may facilitate 2D-local implementations when the outer code is very nonlocal.
Previous work has indicated that implementing nonlocal, high-rate qLDPC codes in 2D appears to be challenging~\cite{delfosse2021boundsstabilizermeasurementcircuits} without mitigating the additional cost of implementing nonlocal gates~\cite{pattison2023, choe2024faulttolerantlyrealizequantumcircuit, berthusen20242dlocalimplementationquantum}. 
While the $[[4,2,2]]$-concatenated generators would inherit this nonlocality, the adaptive scheme ensures that the usage of these nonlocal checks would decrease at low physical error rates. 
Indeed, for QEC rounds in which only the Iceberg generators are measured, we obtain a syndrome extraction circuit that is completely local. 
This behavior might make naive 2D-local implementations of nonlocal codes feasible.

Certain codes have a Tanner graph structure which allows for convenient embeddings into $\mathbb{Z}^2$, such as the La-cross codes~\cite{pecorari2024highratequantumldpccodes} studied here, the related generalized bicycle codes~\cite{kovalev2013, bravyi2023highthreshold}, and long-range-enhanced surface codes~\cite{Hong_2024}. One such embedding for La-cross codes is shown in Fig.~\ref{fig:embedding}, which yields an embedded generator structure consisting of some local connections and some nonlocal connections. Codes like these may be easier to implement in 2D due to their limited nonlocality.
Unfortunately, if we attempt to recreate this embedding with the $[[4,2,2]]$-concatenated La-cross code, much of this locality is destroyed by our mapping of physical HGP qubits to logical Iceberg qubits.
Instead we can use a mapping that assigns a single physical HGP qubit to each Iceberg block, with the other logical qubit acting as a gauge qubit~\cite{criger2016}. 
With this one-to-one mapping between physical HGP qubits and Iceberg blocks, we can obtain essentially the same locality (although spread out by the Iceberg blocks themselves) for the concatenated generators.
This is in addition to the Iceberg generators, which would continue to be local.

We note that using $[[4,2,2]]$-concatenated codes as a way to reduce the cost of implementing nonlocal codes was introduced by Criger and Terhal in Ref.~\cite{criger2016}, where they investigated the $[[4,2,2]]$-concatenated toric code. This code also has an embedding into $\mathbb{Z}^2$ in which there are both local and nonlocal checks.
The authors discussed a superconducting transmon qubit architecture and predicted that these code constructions would perform well in settings where both short-range, high-fidelity and long-range, low-fidelity gates are available.  

\begin{figure}[t]
    \centering
    \includegraphics[width=0.8\linewidth]{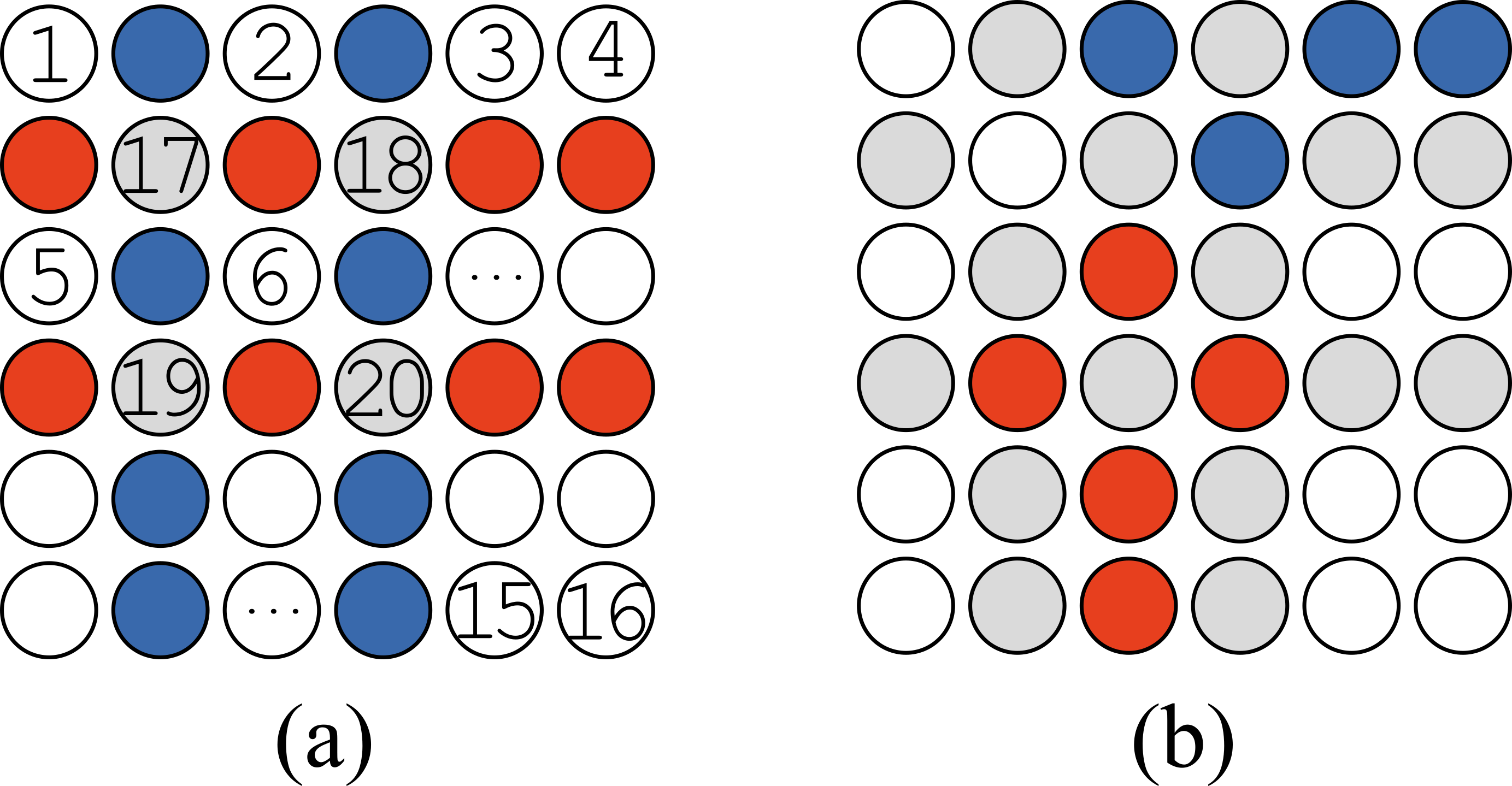}
    \caption{(a) Layout of the data and check qubits of a $[[20,4,2]]$ La-cross code. It can be interpreted as an interleaving of the $L$ (white) and $R$ (gray) sector qubits with qubits used for syndrome extraction readout. Red and blue circles represent $X$-type and $Z$-type check readout qubits, respectively. (b) Example $X$-type (red) and $Z$-type (blue) stabilizer generators. Remaining data and ancilla check qubits are colored white and gray, respectively. The structure of La-cross codes produce generators which have both local and nonlocal connections.  }
    \label{fig:embedding}
\end{figure}

\subsection{Modular QPU architectures}

A related application is implementing qLDPC codes on a modular system of quantum processing units (QPUs). The current generation of quantum computers have around $10^2$-$10^3$ physical qubits with which small experiments and calculations can be performed;
however, scaling a single chip much beyond this size is difficult due to physical and practical limitations.
An alternative to this monolithic architecture is one where many smaller chips are linked together by means of, e.g., photonic interconnects~\cite{monroe2014}, to form a larger, modular system. For these systems, the difficulty is now facilitating inter-module interactions and designing algorithms and QECCs which are well suited for the device.

A simple solution to this problem is to have a separate QECC for each module, and then the inter-module connections would only be used for interactions between code blocks. However, this effectively limits the achievable distance across the entire system, and so combining modules for the same QECC might be desirable~\cite{Strikis_2023}.
This setting fits well with the motivation of the adaptive scheme. The inter-module connections will be much slower and noisier than the intra-module connections, and as such, minimizing the usage of the inter-module connections would likely improve logical performance and speed up the computation. 
Concatenated codes again provide a reasonable approach: 
in each module suppose we put a copies of an $\ell$-level concatenated code $\cQ_\ell$. 
We can then form a final concatenated code from the logical qubits of the copies of $\cQ_\ell$, hence obtaining a code that spans over all modules.
The functionality of the adaptive scheme for this code would be similar to that described in Sec.~\ref{sec:adaptive}, with the main change being that error correction is attempted with the syndrome information from the first $\ell$ levels. In this case, the only time when the inter-module connections would be needed during syndrome extraction is when the lower concatenation levels were unable to fully correct the error.

\section{Discussion}
\label{sec:discussion}

In this work, we have presented adaptive syndrome extraction, a protocol which makes real-time decisions during syndrome extraction in order to minimize the cost of quantum error correction. 
Our protocol can be interpreted as a quantum weight reduction scheme in the space-time picture.
We considered a realization of this scheme through a concatenated code construction, where the logical qubits of many $[[4,2,2]]$ Iceberg code blocks are encoded into a square hypergraph product code.
Note, however, that the scheme is not limited to the code construction presented here;
the main necessity of a code used in adaptive syndrome extraction is the ability to cheaply determine whether a region or subset of qubits has an error, in which case more expensive error correction can take place. 
Concatenated codes provide a straightforward framework to accomplish this, but it is feasible that non-concatenated codes could be designed with this feature in mind. One such option could be the homological product of an error detecting code with a high-rate qLDPC code.
Or adaptive syndrome extraction may be possible even without any additional structure; instead, error detection could be facilitated by measuring a hitting set of stabilizer generators.

Through circuit-level simulations of repeated syndrome extraction, we found that $[[4,2,2]]$-concatenated HGP codes implemented using adaptive syndrome extraction achieve lower logical error rates than their non-concatenated counterparts while requiring fewer CNOT gates and using fewer qubits. 
Although under-performing their non-concatenated counterparts at high error rates, quantum expander codes and the adaptive scheme achieved over an order of magnitude improvement at low error rates while using $\sim 2\times$ fewer CNOT gates. These improved logical error rates come at the cost of fewer logical qubits for the same number of physical qubits. La-cross codes showed a more impressive comparative improvement: when using the adaptive scheme, we found an order of magnitude improvement in the logical error rate even at current experimental physical error rates. For non-concatenated La-cross codes, matching the performance of the adaptive scheme required over twice as many physical qubits and $\sim 5\times$ more CNOT gates.
Despite not reducing circuit depth per se, we described how this reduction in gate count could lead to shorter QEC cycles on certain quantum computing architectures.
We also showed how the scheme provides reductions in the check weight and qubit degree when averaged over many syndrome extraction rounds.

We also provided fault-tolerant gadgets that allowed us to achieve universal logical quantum computation with $[[4,2,2]]$-concatenated HGP codes.
Overall, we showed that our choice of code and syndrome extraction scheme reduces the CNOT gate count for error correction without increasing the space-time cost of logical computation.
In doing so, we showed how the concatenated codes can inherit logical gates from the underlying HGP codes, and we adapted the GPPM scheme of Ref.~\cite{xu2024fast} and the single-shot state preparation scheme of Ref.~\cite{hong2024single} to our logical computation scheme.
By developing a single-shot state preparation scheme for our $[[4,2,2]]$-concatenated HGP code, we showed how we can construct logical gadgets that take $O(1)$ logical cycles, giving us the opportunity to trade space for time in the logical computation protocol. Performance and overheads could potentially be further improved by using more efficient magic state distillation~\cite{zhu2024topological} or alternative codes~\cite{malcolm2025}.

Future research could investigate the potential of replacing the inner and outer codes of the $[[4,2,2]]$-concatenated HGP codes studied here.
It may be interesting to replace the $[[4,2,2]]$ code with larger Iceberg codes, or to use more concatenation layers akin to the many-hypercube~\cite{goto2024manyhypercubecodeshighratequantum} and related codes~\cite{c4c6Knill2005, yoshida2024concatenatecodessavequbits}.
Alternatively, the outer HGP code could be replaced with alternative qLDPC codes, such as bivariate bicycle codes~\cite{bravyi2023highthreshold}, or lifted product codes~\cite{Panteleev_2021, Panteleev_2022}, among others. In both of these cases, one would likely lose the logical-physical qubit mapping of Sec.~\ref{sec:assignment} and the resulting logical gates; however, there might use cases in which those choices are advantageous. 
Finally, there is an alternative concatenation scheme to that described by Procedure~\ref{proc:concat}: for input codes $\cQ_1, \cQ_2$ with parameters $[[n_1, k_1, d_1]]$ and $[[n_2, k_2, d_2]]$, respectively, we can obtain a concatenated code with parameters $[[n_1n_2, k_1k_2, d \ge d_1d_2]]$, see e.g. Section 3.5 of Ref.~\cite{gottesman1997stabilizercodesquantumerror}. The performance of such a code in the adaptive scheme would likely be comparable, but it may provide opportunities for different logical gates~\cite{criger2016}.

\section*{Acknowledgements}

We thank Ciar\'an Ryan-Anderson and Natalie Brown for helpful discussions at the start of the project. We thank Nadine Meister for answering questions about Ref.~\cite{meister2024efficientsoftoutputdecoderssurface}. We also thank Qian Xu and Guo (Jerry) Zheng for helpful discussions and answering questions about Ref.~\cite{xu2024fast}. We thank Yifan Hong for answering questions about Ref.~\cite{hong2024single}. SJST acknowledges funding and support from Joint
Center for Quantum Information and Computer Science (QuICS) Lanczos Graduate Fellowship,
MathQuantum Graduate Fellowship, and the National University of Singapore (NUS) Development Grant. EH is supported by the Fulbright Future Scholarship. DG is partially supported by the National Science Foundation (RQS QLCI grant OMA-2120757).

\section*{Data Availability}
The source code and data to generate the figures in the paper are available at \url{https://github.com/noahberthusen/adaptive_qec}. 

\bibliography{bibliography}

\begin{thebibliography}{93}%
\makeatletter
\providecommand \@ifxundefined [1]{%
 \@ifx{#1\undefined}
}%
\providecommand \@ifnum [1]{%
 \ifnum #1\expandafter \@firstoftwo
 \else \expandafter \@secondoftwo
 \fi
}%
\providecommand \@ifx [1]{%
 \ifx #1\expandafter \@firstoftwo
 \else \expandafter \@secondoftwo
 \fi
}%
\providecommand \natexlab [1]{#1}%
\providecommand \enquote  [1]{``#1''}%
\providecommand \bibnamefont  [1]{#1}%
\providecommand \bibfnamefont [1]{#1}%
\providecommand \citenamefont [1]{#1}%
\providecommand \href@noop [0]{\@secondoftwo}%
\providecommand \href [0]{\begingroup \@sanitize@url \@href}%
\providecommand \@href[1]{\@@startlink{#1}\@@href}%
\providecommand \@@href[1]{\endgroup#1\@@endlink}%
\providecommand \@sanitize@url [0]{\catcode `\\12\catcode `\$12\catcode `\&12\catcode `\#12\catcode `\^12\catcode `\_12\catcode `\%12\relax}%
\providecommand \@@startlink[1]{}%
\providecommand \@@endlink[0]{}%
\providecommand \url  [0]{\begingroup\@sanitize@url \@url }%
\providecommand \@url [1]{\endgroup\@href {#1}{\urlprefix }}%
\providecommand \urlprefix  [0]{URL }%
\providecommand \Eprint [0]{\href }%
\providecommand \doibase [0]{https://doi.org/}%
\providecommand \selectlanguage [0]{\@gobble}%
\providecommand \bibinfo  [0]{\@secondoftwo}%
\providecommand \bibfield  [0]{\@secondoftwo}%
\providecommand \translation [1]{[#1]}%
\providecommand \BibitemOpen [0]{}%
\providecommand \bibitemStop [0]{}%
\providecommand \bibitemNoStop [0]{.\EOS\space}%
\providecommand \EOS [0]{\spacefactor3000\relax}%
\providecommand \BibitemShut  [1]{\csname bibitem#1\endcsname}%
\let\auto@bib@innerbib\@empty
\bibitem [{\citenamefont {{Google Quantum AI}}\ and\ \citenamefont {Collaborators}(2024)}]{acharya2024}%
  \BibitemOpen
  \bibfield  {author} {\bibinfo {author} {\bibnamefont {{Google Quantum AI}}}\ and\ \bibinfo {author} {\bibnamefont {Collaborators}},\ }\bibfield  {title} {\bibinfo {title} {Quantum error correction below the surface code threshold},\ }\href {https://doi.org/10.1038/s41586-024-08449-y} {\bibfield  {journal} {\bibinfo  {journal} {Nature}\ } (\bibinfo {year} {2024})}\BibitemShut {NoStop}%
\bibitem [{\citenamefont {Reichardt}\ \emph {et~al.}(2024{\natexlab{a}})\citenamefont {Reichardt}, \citenamefont {Aasen}, \citenamefont {Chao}, \citenamefont {Chernoguzov}, \citenamefont {van Dam}, \citenamefont {Gaebler}, \citenamefont {Gresh}, \citenamefont {Lucchetti}, \citenamefont {Mills}, \citenamefont {Moses} \emph {et~al.}}]{reichardt2024}%
  \BibitemOpen
  \bibfield  {author} {\bibinfo {author} {\bibfnamefont {B.~W.}\ \bibnamefont {Reichardt}}, \bibinfo {author} {\bibfnamefont {D.}~\bibnamefont {Aasen}}, \bibinfo {author} {\bibfnamefont {R.}~\bibnamefont {Chao}}, \bibinfo {author} {\bibfnamefont {A.}~\bibnamefont {Chernoguzov}}, \bibinfo {author} {\bibfnamefont {W.}~\bibnamefont {van Dam}}, \bibinfo {author} {\bibfnamefont {J.~P.}\ \bibnamefont {Gaebler}}, \bibinfo {author} {\bibfnamefont {D.}~\bibnamefont {Gresh}}, \bibinfo {author} {\bibfnamefont {D.}~\bibnamefont {Lucchetti}}, \bibinfo {author} {\bibfnamefont {M.}~\bibnamefont {Mills}}, \bibinfo {author} {\bibfnamefont {S.~A.}\ \bibnamefont {Moses}}, \emph {et~al.},\ }\bibfield  {title} {\bibinfo {title} {Demonstration of quantum computation and error correction with a tesseract code},\ }\href {https://doi.org/10.48550/arXiv.2409.04628} {\bibfield  {journal} {\bibinfo  {journal} {arXiv preprint arXiv:2409.04628}\ } (\bibinfo {year} {2024}{\natexlab{a}})}\BibitemShut {NoStop}%
\bibitem [{\citenamefont {da~Silva}\ \emph {et~al.}(2024)\citenamefont {da~Silva}, \citenamefont {Ryan-Anderson}, \citenamefont {Bello-Rivas}, \citenamefont {Chernoguzov}, \citenamefont {Dreiling}, \citenamefont {Foltz}, \citenamefont {Frachon}, \citenamefont {Gaebler}, \citenamefont {Gatterman}, \citenamefont {Grans-Samuelsson}, \citenamefont {Hayes} \emph {et~al.}}]{dasilva2024}%
  \BibitemOpen
  \bibfield  {author} {\bibinfo {author} {\bibfnamefont {M.~P.}\ \bibnamefont {da~Silva}}, \bibinfo {author} {\bibfnamefont {C.}~\bibnamefont {Ryan-Anderson}}, \bibinfo {author} {\bibfnamefont {J.~M.}\ \bibnamefont {Bello-Rivas}}, \bibinfo {author} {\bibfnamefont {A.}~\bibnamefont {Chernoguzov}}, \bibinfo {author} {\bibfnamefont {J.~M.}\ \bibnamefont {Dreiling}}, \bibinfo {author} {\bibfnamefont {C.}~\bibnamefont {Foltz}}, \bibinfo {author} {\bibfnamefont {F.}~\bibnamefont {Frachon}}, \bibinfo {author} {\bibfnamefont {J.~P.}\ \bibnamefont {Gaebler}}, \bibinfo {author} {\bibfnamefont {T.~M.}\ \bibnamefont {Gatterman}}, \bibinfo {author} {\bibfnamefont {L.}~\bibnamefont {Grans-Samuelsson}}, \bibinfo {author} {\bibfnamefont {D.}~\bibnamefont {Hayes}}, \emph {et~al.},\ }\bibfield  {title} {\bibinfo {title} {Demonstration of logical qubits and repeated error correction with better-than-physical error rates},\ }\href {https://doi.org/10.48550/arXiv.2404.02280} {\bibfield  {journal} {\bibinfo  {journal} {arXiv
  preprint arXiv:2404.02280}\ } (\bibinfo {year} {2024})}\BibitemShut {NoStop}%
\bibitem [{\citenamefont {Reichardt}\ \emph {et~al.}(2024{\natexlab{b}})\citenamefont {Reichardt}, \citenamefont {Paetznick}, \citenamefont {Aasen}, \citenamefont {Basov}, \citenamefont {Bello-Rivas}, \citenamefont {Bonderson}, \citenamefont {Chao}, \citenamefont {van Dam}, \citenamefont {Hastings} \emph {et~al.}}]{reichardt2024_2}%
  \BibitemOpen
  \bibfield  {author} {\bibinfo {author} {\bibfnamefont {B.~W.}\ \bibnamefont {Reichardt}}, \bibinfo {author} {\bibfnamefont {A.}~\bibnamefont {Paetznick}}, \bibinfo {author} {\bibfnamefont {D.}~\bibnamefont {Aasen}}, \bibinfo {author} {\bibfnamefont {I.}~\bibnamefont {Basov}}, \bibinfo {author} {\bibfnamefont {J.~M.}\ \bibnamefont {Bello-Rivas}}, \bibinfo {author} {\bibfnamefont {P.}~\bibnamefont {Bonderson}}, \bibinfo {author} {\bibfnamefont {R.}~\bibnamefont {Chao}}, \bibinfo {author} {\bibfnamefont {W.}~\bibnamefont {van Dam}}, \bibinfo {author} {\bibfnamefont {M.~B.}\ \bibnamefont {Hastings}}, \emph {et~al.},\ }\bibfield  {title} {\bibinfo {title} {Logical computation demonstrated with a neutral atom quantum processor},\ }\href {https://doi.org/10.48550/arXiv.2411.11822} {\bibfield  {journal} {\bibinfo  {journal} {arXiv preprint arXiv:2411.11822}\ } (\bibinfo {year} {2024}{\natexlab{b}})}\BibitemShut {NoStop}%
\bibitem [{\citenamefont {Lacroix}\ \emph {et~al.}(2024)\citenamefont {Lacroix}, \citenamefont {Bourassa}, \citenamefont {Heras}, \citenamefont {Zhang}, \citenamefont {Bausch}, \citenamefont {Senior}, \citenamefont {Edlich}, \citenamefont {Shutty}, \citenamefont {Sivak}, \citenamefont {Bengtsson} \emph {et~al.}}]{lacroix2024scalinglogiccolorcode}%
  \BibitemOpen
  \bibfield  {author} {\bibinfo {author} {\bibfnamefont {N.}~\bibnamefont {Lacroix}}, \bibinfo {author} {\bibfnamefont {A.}~\bibnamefont {Bourassa}}, \bibinfo {author} {\bibfnamefont {F.~J.~H.}\ \bibnamefont {Heras}}, \bibinfo {author} {\bibfnamefont {L.~M.}\ \bibnamefont {Zhang}}, \bibinfo {author} {\bibfnamefont {J.}~\bibnamefont {Bausch}}, \bibinfo {author} {\bibfnamefont {A.~W.}\ \bibnamefont {Senior}}, \bibinfo {author} {\bibfnamefont {T.}~\bibnamefont {Edlich}}, \bibinfo {author} {\bibfnamefont {N.}~\bibnamefont {Shutty}}, \bibinfo {author} {\bibfnamefont {V.}~\bibnamefont {Sivak}}, \bibinfo {author} {\bibfnamefont {A.}~\bibnamefont {Bengtsson}}, \emph {et~al.},\ }\bibfield  {title} {\bibinfo {title} {Scaling and logic in the color code on a superconducting quantum processor},\ }\href {https://doi.org/10.48550/arXiv.2412.14256} {\bibfield  {journal} {\bibinfo  {journal} {arXiv preprint arXiv:2412.14256}\ } (\bibinfo {year} {2024})}\BibitemShut {NoStop}%
\bibitem [{\citenamefont {Xu}\ \emph {et~al.}(2018)\citenamefont {Xu}, \citenamefont {Beaudrap}, \citenamefont {O’Gorman},\ and\ \citenamefont {Benjamin}}]{Xu_2018}%
  \BibitemOpen
  \bibfield  {author} {\bibinfo {author} {\bibfnamefont {X.}~\bibnamefont {Xu}}, \bibinfo {author} {\bibfnamefont {N.~d.}\ \bibnamefont {Beaudrap}}, \bibinfo {author} {\bibfnamefont {J.}~\bibnamefont {O’Gorman}},\ and\ \bibinfo {author} {\bibfnamefont {S.~C.}\ \bibnamefont {Benjamin}},\ }\bibfield  {title} {\bibinfo {title} {An integrity measure to benchmark quantum error correcting memories},\ }\href {http://dx.doi.org/10.1088/1367-2630/aaa372} {\bibfield  {journal} {\bibinfo  {journal} {New Journal of Physics}\ }\textbf {\bibinfo {volume} {20}},\ \bibinfo {pages} {023009} (\bibinfo {year} {2018})}\BibitemShut {NoStop}%
\bibitem [{\citenamefont {Delfosse}\ \emph {et~al.}(2021{\natexlab{a}})\citenamefont {Delfosse}, \citenamefont {Beverland},\ and\ \citenamefont {Tremblay}}]{Delfosse_Beverland_Tremblay_2021}%
  \BibitemOpen
  \bibfield  {author} {\bibinfo {author} {\bibfnamefont {N.}~\bibnamefont {Delfosse}}, \bibinfo {author} {\bibfnamefont {M.~E.}\ \bibnamefont {Beverland}},\ and\ \bibinfo {author} {\bibfnamefont {M.~A.}\ \bibnamefont {Tremblay}},\ }\bibfield  {title} {\bibinfo {title} {Bounds on stabilizer measurement circuits and obstructions to local implementations of quantum {LDPC} codes},\ }\href {https://doi.org/10.48550/arXiv.2109.14599} {\bibfield  {journal} {\bibinfo  {journal} {arXiv preprint arXiv:2109.14599}\ } (\bibinfo {year} {2021}{\natexlab{a}})}\BibitemShut {NoStop}%
\bibitem [{\citenamefont {Baspin}\ \emph {et~al.}(2023)\citenamefont {Baspin}, \citenamefont {Fawzi},\ and\ \citenamefont {Shayeghi}}]{baspin2023lower}%
  \BibitemOpen
  \bibfield  {author} {\bibinfo {author} {\bibfnamefont {N.}~\bibnamefont {Baspin}}, \bibinfo {author} {\bibfnamefont {O.}~\bibnamefont {Fawzi}},\ and\ \bibinfo {author} {\bibfnamefont {A.}~\bibnamefont {Shayeghi}},\ }\bibfield  {title} {\bibinfo {title} {A lower bound on the overhead of quantum error correction in low dimensions},\ }\href {https://doi.org/10.48550/arXiv.2302.04317} {\bibfield  {journal} {\bibinfo  {journal} {arXiv preprint arXiv:2302.04317}\ } (\bibinfo {year} {2023})}\BibitemShut {NoStop}%
\bibitem [{\citenamefont {Brown}\ \emph {et~al.}(2023)\citenamefont {Brown}, \citenamefont {III}, \citenamefont {Granade}, \citenamefont {Heim}, \citenamefont {Wernli}, \citenamefont {Ryan-Anderson}, \citenamefont {Lucchetti}, \citenamefont {Paetznick}, \citenamefont {Roetteler}, \citenamefont {Svore},\ and\ \citenamefont {Chernoguzov}}]{brown2023}%
  \BibitemOpen
  \bibfield  {author} {\bibinfo {author} {\bibfnamefont {N.~C.}\ \bibnamefont {Brown}}, \bibinfo {author} {\bibfnamefont {J.~P.~C.}\ \bibnamefont {III}}, \bibinfo {author} {\bibfnamefont {C.}~\bibnamefont {Granade}}, \bibinfo {author} {\bibfnamefont {B.}~\bibnamefont {Heim}}, \bibinfo {author} {\bibfnamefont {S.}~\bibnamefont {Wernli}}, \bibinfo {author} {\bibfnamefont {C.}~\bibnamefont {Ryan-Anderson}}, \bibinfo {author} {\bibfnamefont {D.}~\bibnamefont {Lucchetti}}, \bibinfo {author} {\bibfnamefont {A.}~\bibnamefont {Paetznick}}, \bibinfo {author} {\bibfnamefont {M.}~\bibnamefont {Roetteler}}, \bibinfo {author} {\bibfnamefont {K.}~\bibnamefont {Svore}},\ and\ \bibinfo {author} {\bibfnamefont {A.}~\bibnamefont {Chernoguzov}},\ }\bibfield  {title} {\bibinfo {title} {Advances in compilation for quantum hardware -- a demonstration of magic state distillation and repeat-until-success protocols},\ }\href {https://doi.org/10.48550/arXiv.2310.12106} {\bibfield  {journal} {\bibinfo  {journal} {arXiv preprint
  arXiv:2310.12106}\ } (\bibinfo {year} {2023})}\BibitemShut {NoStop}%
\bibitem [{\citenamefont {Reichardt}(2020)}]{Reichardt_2020}%
  \BibitemOpen
  \bibfield  {author} {\bibinfo {author} {\bibfnamefont {B.~W.}\ \bibnamefont {Reichardt}},\ }\bibfield  {title} {\bibinfo {title} {Fault-tolerant quantum error correction for steane’s seven-qubit color code with few or no extra qubits},\ }\href {http://dx.doi.org/10.1088/2058-9565/abc6f4} {\bibfield  {journal} {\bibinfo  {journal} {Quantum Science and Technology}\ }\textbf {\bibinfo {volume} {6}},\ \bibinfo {pages} {015007} (\bibinfo {year} {2020})}\BibitemShut {NoStop}%
\bibitem [{\citenamefont {Ryan-Anderson}\ \emph {et~al.}(2021)\citenamefont {Ryan-Anderson}, \citenamefont {Bohnet}, \citenamefont {Lee}, \citenamefont {Gresh}, \citenamefont {Hankin}, \citenamefont {Gaebler}, \citenamefont {Francois}, \citenamefont {Chernoguzov}, \citenamefont {Lucchetti}, \citenamefont {Brown} \emph {et~al.}}]{ryananderson2021}%
  \BibitemOpen
  \bibfield  {author} {\bibinfo {author} {\bibfnamefont {C.}~\bibnamefont {Ryan-Anderson}}, \bibinfo {author} {\bibfnamefont {J.~G.}\ \bibnamefont {Bohnet}}, \bibinfo {author} {\bibfnamefont {K.}~\bibnamefont {Lee}}, \bibinfo {author} {\bibfnamefont {D.}~\bibnamefont {Gresh}}, \bibinfo {author} {\bibfnamefont {A.}~\bibnamefont {Hankin}}, \bibinfo {author} {\bibfnamefont {J.~P.}\ \bibnamefont {Gaebler}}, \bibinfo {author} {\bibfnamefont {D.}~\bibnamefont {Francois}}, \bibinfo {author} {\bibfnamefont {A.}~\bibnamefont {Chernoguzov}}, \bibinfo {author} {\bibfnamefont {D.}~\bibnamefont {Lucchetti}}, \bibinfo {author} {\bibfnamefont {N.~C.}\ \bibnamefont {Brown}}, \emph {et~al.},\ }\bibfield  {title} {\bibinfo {title} {Realization of real-time fault-tolerant quantum error correction},\ }\href {https://link.aps.org/doi/10.1103/PhysRevX.11.041058} {\bibfield  {journal} {\bibinfo  {journal} {Phys. Rev. X}\ }\textbf {\bibinfo {volume} {11}},\ \bibinfo {pages} {041058} (\bibinfo {year} {2021})}\BibitemShut {NoStop}%
\bibitem [{\citenamefont {Berthusen}\ \emph {et~al.}(2025)\citenamefont {Berthusen}, \citenamefont {Devulapalli}, \citenamefont {Schoute}, \citenamefont {Childs}, \citenamefont {Gullans}, \citenamefont {Gorshkov},\ and\ \citenamefont {Gottesman}}]{berthusen20242dlocalimplementationquantum}%
  \BibitemOpen
  \bibfield  {author} {\bibinfo {author} {\bibfnamefont {N.}~\bibnamefont {Berthusen}}, \bibinfo {author} {\bibfnamefont {D.}~\bibnamefont {Devulapalli}}, \bibinfo {author} {\bibfnamefont {E.}~\bibnamefont {Schoute}}, \bibinfo {author} {\bibfnamefont {A.~M.}\ \bibnamefont {Childs}}, \bibinfo {author} {\bibfnamefont {M.~J.}\ \bibnamefont {Gullans}}, \bibinfo {author} {\bibfnamefont {A.~V.}\ \bibnamefont {Gorshkov}},\ and\ \bibinfo {author} {\bibfnamefont {D.}~\bibnamefont {Gottesman}},\ }\bibfield  {title} {\bibinfo {title} {Toward a {2D} local implementation of quantum low-density parity-check codes},\ }\href {http://dx.doi.org/10.1103/PRXQuantum.6.010306} {\bibfield  {journal} {\bibinfo  {journal} {PRX Quantum}\ }\textbf {\bibinfo {volume} {6}} (\bibinfo {year} {2025})}\BibitemShut {NoStop}%
\bibitem [{\citenamefont {Hastings}(2016)}]{hastings2016weight}%
  \BibitemOpen
  \bibfield  {author} {\bibinfo {author} {\bibfnamefont {M.~B.}\ \bibnamefont {Hastings}},\ }\bibfield  {title} {\bibinfo {title} {Weight reduction for quantum codes},\ }\href {https://doi.org/10.48550/arXiv.1611.03790} {\bibfield  {journal} {\bibinfo  {journal} {arXiv preprint arXiv:1611.03790}\ } (\bibinfo {year} {2016})}\BibitemShut {NoStop}%
\bibitem [{\citenamefont {Hastings}(2023)}]{hastings2023quantumweightreduction}%
  \BibitemOpen
  \bibfield  {author} {\bibinfo {author} {\bibfnamefont {M.~B.}\ \bibnamefont {Hastings}},\ }\bibfield  {title} {\bibinfo {title} {On quantum weight reduction},\ }\href {https://doi.org/10.48550/arXiv.2102.10030} {\bibfield  {journal} {\bibinfo  {journal} {arXiv preprint arXiv:2102.10030}\ } (\bibinfo {year} {2023})}\BibitemShut {NoStop}%
\bibitem [{\citenamefont {Poulin}(2005)}]{poulin2005stabilizer}%
  \BibitemOpen
  \bibfield  {author} {\bibinfo {author} {\bibfnamefont {D.}~\bibnamefont {Poulin}},\ }\bibfield  {title} {\bibinfo {title} {Stabilizer formalism for operator quantum error correction},\ }\href {https://doi.org/10.1103/PhysRevLett.95.230504} {\bibfield  {journal} {\bibinfo  {journal} {Phys. Rev. Lett.}\ }\textbf {\bibinfo {volume} {95}},\ \bibinfo {pages} {230504} (\bibinfo {year} {2005})}\BibitemShut {NoStop}%
\bibitem [{\citenamefont {Bacon}(2006)}]{bacon2006operator}%
  \BibitemOpen
  \bibfield  {author} {\bibinfo {author} {\bibfnamefont {D.}~\bibnamefont {Bacon}},\ }\bibfield  {title} {\bibinfo {title} {Operator quantum error-correcting subsystems for self-correcting quantum memories},\ }\href {https://doi.org/10.1103/PhysRevA.73.012340} {\bibfield  {journal} {\bibinfo  {journal} {Phys. Rev. A}\ }\textbf {\bibinfo {volume} {73}},\ \bibinfo {pages} {012340} (\bibinfo {year} {2006})}\BibitemShut {NoStop}%
\bibitem [{\citenamefont {Bombin}(2010)}]{bombin2010topological}%
  \BibitemOpen
  \bibfield  {author} {\bibinfo {author} {\bibfnamefont {H.}~\bibnamefont {Bombin}},\ }\bibfield  {title} {\bibinfo {title} {Topological subsystem codes},\ }\href {https://link.aps.org/doi/10.1103/PhysRevA.81.032301} {\bibfield  {journal} {\bibinfo  {journal} {Phys. Rev. A}\ }\textbf {\bibinfo {volume} {81}},\ \bibinfo {pages} {032301} (\bibinfo {year} {2010})}\BibitemShut {NoStop}%
\bibitem [{\citenamefont {Baspin}\ and\ \citenamefont {Williamson}(2024)}]{baspin2024wire}%
  \BibitemOpen
  \bibfield  {author} {\bibinfo {author} {\bibfnamefont {N.}~\bibnamefont {Baspin}}\ and\ \bibinfo {author} {\bibfnamefont {D.}~\bibnamefont {Williamson}},\ }\bibfield  {title} {\bibinfo {title} {Wire codes},\ }\href {https://doi.org/10.48550/arXiv.2410.10194} {\bibfield  {journal} {\bibinfo  {journal} {arXiv preprint arXiv:2410.10194}\ } (\bibinfo {year} {2024})}\BibitemShut {NoStop}%
\bibitem [{\citenamefont {Bacon}\ \emph {et~al.}(2017)\citenamefont {Bacon}, \citenamefont {Flammia}, \citenamefont {Harrow},\ and\ \citenamefont {Shi}}]{bacon2017sparse}%
  \BibitemOpen
  \bibfield  {author} {\bibinfo {author} {\bibfnamefont {D.}~\bibnamefont {Bacon}}, \bibinfo {author} {\bibfnamefont {S.~T.}\ \bibnamefont {Flammia}}, \bibinfo {author} {\bibfnamefont {A.~W.}\ \bibnamefont {Harrow}},\ and\ \bibinfo {author} {\bibfnamefont {J.}~\bibnamefont {Shi}},\ }\bibfield  {title} {\bibinfo {title} {Sparse quantum codes from quantum circuits},\ }\href {https://doi.org/10.1109/TIT.2017.2663199} {\bibfield  {journal} {\bibinfo  {journal} {IEEE Transactions on Information Theory}\ }\textbf {\bibinfo {volume} {63}},\ \bibinfo {pages} {2464} (\bibinfo {year} {2017})}\BibitemShut {NoStop}%
\bibitem [{\citenamefont {Hastings}\ and\ \citenamefont {Haah}(2021)}]{hastings2021dynamically}%
  \BibitemOpen
  \bibfield  {author} {\bibinfo {author} {\bibfnamefont {M.~B.}\ \bibnamefont {Hastings}}\ and\ \bibinfo {author} {\bibfnamefont {J.}~\bibnamefont {Haah}},\ }\bibfield  {title} {\bibinfo {title} {Dynamically generated logical qubits},\ }\href {http://dx.doi.org/10.22331/q-2021-10-19-564} {\bibfield  {journal} {\bibinfo  {journal} {Quantum}\ }\textbf {\bibinfo {volume} {5}},\ \bibinfo {pages} {564} (\bibinfo {year} {2021})}\BibitemShut {NoStop}%
\bibitem [{\citenamefont {Gottesman}(2022)}]{gottesman2022opportunities}%
  \BibitemOpen
  \bibfield  {author} {\bibinfo {author} {\bibfnamefont {D.}~\bibnamefont {Gottesman}},\ }\bibfield  {title} {\bibinfo {title} {Opportunities and challenges in fault-tolerant quantum computation},\ }\href {https://doi.org/10.48550/arXiv.2210.15844} {\bibfield  {journal} {\bibinfo  {journal} {arXiv preprint arXiv:2210.15844}\ } (\bibinfo {year} {2022})}\BibitemShut {NoStop}%
\bibitem [{\citenamefont {Delfosse}\ and\ \citenamefont {Paetznick}(2023)}]{delfosse2023spacetime}%
  \BibitemOpen
  \bibfield  {author} {\bibinfo {author} {\bibfnamefont {N.}~\bibnamefont {Delfosse}}\ and\ \bibinfo {author} {\bibfnamefont {A.}~\bibnamefont {Paetznick}},\ }\bibfield  {title} {\bibinfo {title} {Spacetime codes of clifford circuits},\ }\href {https://doi.org/10.48550/arXiv.2304.05943} {\bibfield  {journal} {\bibinfo  {journal} {arXiv preprint arXiv:2304.05943}\ } (\bibinfo {year} {2023})}\BibitemShut {NoStop}%
\bibitem [{\citenamefont {de~la Fuente}(2024)}]{de2024dynamical}%
  \BibitemOpen
  \bibfield  {author} {\bibinfo {author} {\bibfnamefont {J.~C.~M.}\ \bibnamefont {de~la Fuente}},\ }\bibfield  {title} {\bibinfo {title} {Dynamical weight reduction of pauli measurements},\ }\href {https://doi.org/10.48550/arXiv.2410.12527} {\bibfield  {journal} {\bibinfo  {journal} {arXiv preprint arXiv:2410.12527}\ } (\bibinfo {year} {2024})}\BibitemShut {NoStop}%
\bibitem [{\citenamefont {Steane}(1996{\natexlab{a}})}]{Steane_1996}%
  \BibitemOpen
  \bibfield  {author} {\bibinfo {author} {\bibfnamefont {A.~M.}\ \bibnamefont {Steane}},\ }\bibfield  {title} {\bibinfo {title} {Simple quantum error-correcting codes},\ }\href {https://doi.org/10.1103/physreva.54.4741} {\bibfield  {journal} {\bibinfo  {journal} {Phys. Rev. A}\ }\textbf {\bibinfo {volume} {54}},\ \bibinfo {pages} {4741–4751} (\bibinfo {year} {1996}{\natexlab{a}})}\BibitemShut {NoStop}%
\bibitem [{\citenamefont {Gottesman}(1997)}]{gottesman1997stabilizercodesquantumerror}%
  \BibitemOpen
  \bibfield  {author} {\bibinfo {author} {\bibfnamefont {D.}~\bibnamefont {Gottesman}},\ }\bibfield  {title} {\bibinfo {title} {Stabilizer codes and quantum error correction},\ }\href {https://doi.org/10.48550/arXiv.quant-ph/9705052} {\bibfield  {journal} {\bibinfo  {journal} {arXiv preprint arXiv:quant-ph/9705052}\ } (\bibinfo {year} {1997})}\BibitemShut {NoStop}%
\bibitem [{\citenamefont {Self}\ \emph {et~al.}(2024)\citenamefont {Self}, \citenamefont {Benedetti},\ and\ \citenamefont {Amaro}}]{Self_2024}%
  \BibitemOpen
  \bibfield  {author} {\bibinfo {author} {\bibfnamefont {C.~N.}\ \bibnamefont {Self}}, \bibinfo {author} {\bibfnamefont {M.}~\bibnamefont {Benedetti}},\ and\ \bibinfo {author} {\bibfnamefont {D.}~\bibnamefont {Amaro}},\ }\bibfield  {title} {\bibinfo {title} {Protecting expressive circuits with a quantum error detection code},\ }\href {https://doi.org/10.1038/s41567-023-02282-2} {\bibfield  {journal} {\bibinfo  {journal} {Nature Physics}\ }\textbf {\bibinfo {volume} {20}},\ \bibinfo {pages} {219–224} (\bibinfo {year} {2024})}\BibitemShut {NoStop}%
\bibitem [{\citenamefont {Tillich}\ and\ \citenamefont {Zemor}(2014)}]{Tillich_2014}%
  \BibitemOpen
  \bibfield  {author} {\bibinfo {author} {\bibfnamefont {J.-P.}\ \bibnamefont {Tillich}}\ and\ \bibinfo {author} {\bibfnamefont {G.}~\bibnamefont {Zemor}},\ }\bibfield  {title} {\bibinfo {title} {Quantum {LDPC} codes with positive rate and minimum distance proportional to the square root of the blocklength},\ }\href {https://doi.org/10.1109/tit.2013.2292061} {\bibfield  {journal} {\bibinfo  {journal} {IEEE Transactions on Information Theory}\ }\textbf {\bibinfo {volume} {60}},\ \bibinfo {pages} {1193–1202} (\bibinfo {year} {2014})}\BibitemShut {NoStop}%
\bibitem [{\citenamefont {Xu}\ \emph {et~al.}(2024{\natexlab{a}})\citenamefont {Xu}, \citenamefont {Zhou}, \citenamefont {Zheng}, \citenamefont {Bluvstein}, \citenamefont {Ataides}, \citenamefont {Lukin},\ and\ \citenamefont {Jiang}}]{xu2024fast}%
  \BibitemOpen
  \bibfield  {author} {\bibinfo {author} {\bibfnamefont {Q.}~\bibnamefont {Xu}}, \bibinfo {author} {\bibfnamefont {H.}~\bibnamefont {Zhou}}, \bibinfo {author} {\bibfnamefont {G.}~\bibnamefont {Zheng}}, \bibinfo {author} {\bibfnamefont {D.}~\bibnamefont {Bluvstein}}, \bibinfo {author} {\bibfnamefont {J.}~\bibnamefont {Ataides}}, \bibinfo {author} {\bibfnamefont {M.~D.}\ \bibnamefont {Lukin}},\ and\ \bibinfo {author} {\bibfnamefont {L.}~\bibnamefont {Jiang}},\ }\bibfield  {title} {\bibinfo {title} {Fast and parallelizable logical computation with homological product codes},\ }\href {https://doi.org/10.48550/arXiv.2407.18490} {\bibfield  {journal} {\bibinfo  {journal} {arXiv preprint arXiv:2407.18490}\ } (\bibinfo {year} {2024}{\natexlab{a}})}\BibitemShut {NoStop}%
\bibitem [{\citenamefont {Hong}(2024)}]{hong2024single}%
  \BibitemOpen
  \bibfield  {author} {\bibinfo {author} {\bibfnamefont {Y.}~\bibnamefont {Hong}},\ }\bibfield  {title} {\bibinfo {title} {Single-shot preparation of hypergraph product codes via dimension jump},\ }\href {https://doi.org/10.48550/arXiv.2410.05171} {\bibfield  {journal} {\bibinfo  {journal} {arXiv preprint arXiv:2410.05171}\ } (\bibinfo {year} {2024})}\BibitemShut {NoStop}%
\bibitem [{\citenamefont {Calderbank}\ and\ \citenamefont {Shor}(1996)}]{shor1996_2}%
  \BibitemOpen
  \bibfield  {author} {\bibinfo {author} {\bibfnamefont {A.~R.}\ \bibnamefont {Calderbank}}\ and\ \bibinfo {author} {\bibfnamefont {P.~W.}\ \bibnamefont {Shor}},\ }\bibfield  {title} {\bibinfo {title} {Good quantum error-correcting codes exist},\ }\href {https://doi.org/10.1103/PhysRevA.54.1098} {\bibfield  {journal} {\bibinfo  {journal} {Phys. Rev. A}\ }\textbf {\bibinfo {volume} {54}},\ \bibinfo {pages} {1098} (\bibinfo {year} {1996})}\BibitemShut {NoStop}%
\bibitem [{\citenamefont {Steane}(1996{\natexlab{b}})}]{steane1996_2}%
  \BibitemOpen
  \bibfield  {author} {\bibinfo {author} {\bibfnamefont {A.}~\bibnamefont {Steane}},\ }\bibfield  {title} {\bibinfo {title} {Multiple particle interference and quantum error correction},\ }\href {https://doi.org/10.1098/rspa.1996.0136} {\bibfield  {journal} {\bibinfo  {journal} {Proceedings of the Royal Society of London. Series A: Mathematical, Physical and Engineering Sciences}\ }\textbf {\bibinfo {volume} {452}},\ \bibinfo {pages} {2551} (\bibinfo {year} {1996}{\natexlab{b}})}\BibitemShut {NoStop}%
\bibitem [{\citenamefont {Bravyi}\ and\ \citenamefont {Hastings}(2014)}]{Bravyi2014}%
  \BibitemOpen
  \bibfield  {author} {\bibinfo {author} {\bibfnamefont {S.}~\bibnamefont {Bravyi}}\ and\ \bibinfo {author} {\bibfnamefont {M.~B.}\ \bibnamefont {Hastings}},\ }\bibfield  {title} {\bibinfo {title} {Homological product codes},\ }in\ \href {https://doi.org/10.1145/2591796.2591870} {\emph {\bibinfo {booktitle} {Proceedings of the Forty-Sixth Annual ACM Symposium on Theory of Computing}}},\ \bibinfo {series and number} {STOC '14}\ (\bibinfo  {publisher} {Association for Computing Machinery},\ \bibinfo {address} {New York, NY, USA},\ \bibinfo {year} {2014})\ p.\ \bibinfo {pages} {273–282}\BibitemShut {NoStop}%
\bibitem [{\citenamefont {Sipser}\ and\ \citenamefont {Spielman}(1996)}]{Sipser_Spielman_1996}%
  \BibitemOpen
  \bibfield  {author} {\bibinfo {author} {\bibfnamefont {M.}~\bibnamefont {Sipser}}\ and\ \bibinfo {author} {\bibfnamefont {D.}~\bibnamefont {Spielman}},\ }\bibfield  {title} {\bibinfo {title} {Expander codes},\ }\href {https://doi.org/10.1109/18.556667} {\bibfield  {journal} {\bibinfo  {journal} {IEEE Transactions on Information Theory}\ }\textbf {\bibinfo {volume} {42}},\ \bibinfo {pages} {1710–1722} (\bibinfo {year} {1996})}\BibitemShut {NoStop}%
\bibitem [{\citenamefont {Leverrier}\ \emph {et~al.}(2015)\citenamefont {Leverrier}, \citenamefont {Tillich},\ and\ \citenamefont {Zemor}}]{Leverrier_2015}%
  \BibitemOpen
  \bibfield  {author} {\bibinfo {author} {\bibfnamefont {A.}~\bibnamefont {Leverrier}}, \bibinfo {author} {\bibfnamefont {J.-P.}\ \bibnamefont {Tillich}},\ and\ \bibinfo {author} {\bibfnamefont {G.}~\bibnamefont {Zemor}},\ }\bibfield  {title} {\bibinfo {title} {Quantum expander codes},\ }in\ \href {https://doi.org/10.1109/focs.2015.55} {\emph {\bibinfo {booktitle} {2015 IEEE 56th Annual Symposium on Foundations of Computer Science}}}\ (\bibinfo  {publisher} {IEEE},\ \bibinfo {year} {2015})\BibitemShut {NoStop}%
\bibitem [{\citenamefont {Quintavalle}\ and\ \citenamefont {Campbell}(2022)}]{quintavalle2022reshape}%
  \BibitemOpen
  \bibfield  {author} {\bibinfo {author} {\bibfnamefont {A.~O.}\ \bibnamefont {Quintavalle}}\ and\ \bibinfo {author} {\bibfnamefont {E.~T.}\ \bibnamefont {Campbell}},\ }\bibfield  {title} {\bibinfo {title} {Reshape: A decoder for hypergraph product codes},\ }\href {https://doi.org/10.1109/TIT.2022.3184108} {\bibfield  {journal} {\bibinfo  {journal} {IEEE Transactions on Information Theory}\ }\textbf {\bibinfo {volume} {68}} (\bibinfo {year} {2022})}\BibitemShut {NoStop}%
\bibitem [{\citenamefont {Quintavalle}\ \emph {et~al.}(2023)\citenamefont {Quintavalle}, \citenamefont {Webster},\ and\ \citenamefont {Vasmer}}]{quintavalle2023partitioning}%
  \BibitemOpen
  \bibfield  {author} {\bibinfo {author} {\bibfnamefont {A.~O.}\ \bibnamefont {Quintavalle}}, \bibinfo {author} {\bibfnamefont {P.}~\bibnamefont {Webster}},\ and\ \bibinfo {author} {\bibfnamefont {M.}~\bibnamefont {Vasmer}},\ }\bibfield  {title} {\bibinfo {title} {Partitioning qubits in hypergraph product codes to implement logical gates},\ }\href {http://dx.doi.org/10.22331/q-2023-10-24-1153} {\bibfield  {journal} {\bibinfo  {journal} {Quantum}\ }\textbf {\bibinfo {volume} {7}},\ \bibinfo {pages} {1153} (\bibinfo {year} {2023})}\BibitemShut {NoStop}%
\bibitem [{\citenamefont {Pryadko}\ \emph {et~al.}(2022)\citenamefont {Pryadko}, \citenamefont {Shabashov},\ and\ \citenamefont {Kozin}}]{Pryadko_2022}%
  \BibitemOpen
  \bibfield  {author} {\bibinfo {author} {\bibfnamefont {L.~P.}\ \bibnamefont {Pryadko}}, \bibinfo {author} {\bibfnamefont {V.~A.}\ \bibnamefont {Shabashov}},\ and\ \bibinfo {author} {\bibfnamefont {V.~K.}\ \bibnamefont {Kozin}},\ }\bibfield  {title} {\bibinfo {title} {{QDistRnd}: A {GAP} package for computing the distance of quantum error-correcting codes},\ }\href {https://doi.org/10.21105/joss.04120} {\bibfield  {journal} {\bibinfo  {journal} {Journal of Open Source Software}\ }\textbf {\bibinfo {volume} {7}},\ \bibinfo {pages} {4120} (\bibinfo {year} {2022})}\BibitemShut {NoStop}%
\bibitem [{\citenamefont {Knill}(2005{\natexlab{a}})}]{c4c6Knill2005}%
  \BibitemOpen
  \bibfield  {author} {\bibinfo {author} {\bibfnamefont {E.}~\bibnamefont {Knill}},\ }\bibfield  {title} {\bibinfo {title} {Quantum computing with realistically noisy devices},\ }\href {https://doi.org/10.1038/nature03350} {\bibfield  {journal} {\bibinfo  {journal} {Nature}\ }\textbf {\bibinfo {volume} {434}},\ \bibinfo {pages} {39} (\bibinfo {year} {2005}{\natexlab{a}})}\BibitemShut {NoStop}%
\bibitem [{\citenamefont {Yoshida}\ \emph {et~al.}(2024)\citenamefont {Yoshida}, \citenamefont {Tamiya},\ and\ \citenamefont {Yamasaki}}]{yoshida2024concatenatecodessavequbits}%
  \BibitemOpen
  \bibfield  {author} {\bibinfo {author} {\bibfnamefont {S.}~\bibnamefont {Yoshida}}, \bibinfo {author} {\bibfnamefont {S.}~\bibnamefont {Tamiya}},\ and\ \bibinfo {author} {\bibfnamefont {H.}~\bibnamefont {Yamasaki}},\ }\bibfield  {title} {\bibinfo {title} {Concatenate codes, save qubits},\ }\href {https://doi.org/10.48550/arXiv.2402.09606} {\bibfield  {journal} {\bibinfo  {journal} {arXiv preprint arXiv:2402.09606}\ } (\bibinfo {year} {2024})}\BibitemShut {NoStop}%
\bibitem [{\citenamefont {Goto}(2024)}]{goto2024manyhypercubecodeshighratequantum}%
  \BibitemOpen
  \bibfield  {author} {\bibinfo {author} {\bibfnamefont {H.}~\bibnamefont {Goto}},\ }\bibfield  {title} {\bibinfo {title} {High-performance fault-tolerant quantum computing with many-hypercube codes},\ }\href {http://dx.doi.org/10.1126/sciadv.adp6388} {\bibfield  {journal} {\bibinfo  {journal} {Science Advances}\ }\textbf {\bibinfo {volume} {10}} (\bibinfo {year} {2024})}\BibitemShut {NoStop}%
\bibitem [{\citenamefont {Criger}\ and\ \citenamefont {Terhal}(2016)}]{criger2016}%
  \BibitemOpen
  \bibfield  {author} {\bibinfo {author} {\bibfnamefont {B.}~\bibnamefont {Criger}}\ and\ \bibinfo {author} {\bibfnamefont {B.}~\bibnamefont {Terhal}},\ }\bibfield  {title} {\bibinfo {title} {Noise thresholds for the [[4, 2, 2]]-concatenated toric code},\ }\href {http://dx.doi.org/10.26421/QIC16.15-16} {\bibfield  {journal} {\bibinfo  {journal} {Quantum Information and Computation}\ }\textbf {\bibinfo {volume} {16}} (\bibinfo {year} {2016})}\BibitemShut {NoStop}%
\bibitem [{\citenamefont {Knill}(2005{\natexlab{b}})}]{knill2005}%
  \BibitemOpen
  \bibfield  {author} {\bibinfo {author} {\bibfnamefont {E.}~\bibnamefont {Knill}},\ }\bibfield  {title} {\bibinfo {title} {Scalable quantum computing in the presence of large detected-error rates},\ }\href {https://link.aps.org/doi/10.1103/PhysRevA.71.042322} {\bibfield  {journal} {\bibinfo  {journal} {Phys. Rev. A}\ }\textbf {\bibinfo {volume} {71}},\ \bibinfo {pages} {042322} (\bibinfo {year} {2005}{\natexlab{b}})}\BibitemShut {NoStop}%
\bibitem [{\citenamefont {Shor}(1996)}]{Shor1996}%
  \BibitemOpen
  \bibfield  {author} {\bibinfo {author} {\bibfnamefont {P.}~\bibnamefont {Shor}},\ }\bibfield  {title} {\bibinfo {title} {Fault-tolerant quantum computation},\ }in\ \href {https://doi.org/10.1109/SFCS.1996.548464} {\emph {\bibinfo {booktitle} {Proceedings of 37th Conference on Foundations of Computer Science}}}\ (\bibinfo {year} {1996})\ pp.\ \bibinfo {pages} {56--65}\BibitemShut {NoStop}%
\bibitem [{\citenamefont {Goto}\ and\ \citenamefont {Uchikawa}(2013)}]{goto2013}%
  \BibitemOpen
  \bibfield  {author} {\bibinfo {author} {\bibfnamefont {H.}~\bibnamefont {Goto}}\ and\ \bibinfo {author} {\bibfnamefont {H.}~\bibnamefont {Uchikawa}},\ }\bibfield  {title} {\bibinfo {title} {Fault-tolerant quantum computation with a soft-decision decoder for error correction and detection by teleportation},\ }\href {https://doi.org/10.1038/srep02044} {\bibfield  {journal} {\bibinfo  {journal} {Sci. Rep.}\ }\textbf {\bibinfo {volume} {3}} (\bibinfo {year} {2013})}\BibitemShut {NoStop}%
\bibitem [{\citenamefont {Poulin}(2006)}]{Poulin_2006}%
  \BibitemOpen
  \bibfield  {author} {\bibinfo {author} {\bibfnamefont {D.}~\bibnamefont {Poulin}},\ }\bibfield  {title} {\bibinfo {title} {Optimal and efficient decoding of concatenated quantum block codes},\ }\href {http://dx.doi.org/10.1103/PhysRevA.74.052333} {\bibfield  {journal} {\bibinfo  {journal} {Phys. Rev. A}\ }\textbf {\bibinfo {volume} {74}} (\bibinfo {year} {2006})}\BibitemShut {NoStop}%
\bibitem [{\citenamefont {Meister}\ \emph {et~al.}(2024)\citenamefont {Meister}, \citenamefont {Pattison},\ and\ \citenamefont {Preskill}}]{meister2024efficientsoftoutputdecoderssurface}%
  \BibitemOpen
  \bibfield  {author} {\bibinfo {author} {\bibfnamefont {N.}~\bibnamefont {Meister}}, \bibinfo {author} {\bibfnamefont {C.~A.}\ \bibnamefont {Pattison}},\ and\ \bibinfo {author} {\bibfnamefont {J.}~\bibnamefont {Preskill}},\ }\bibfield  {title} {\bibinfo {title} {Efficient soft-output decoders for the surface code},\ }\href {https://doi.org/10.48550/arXiv.2405.07433} {\bibfield  {journal} {\bibinfo  {journal} {arXiv preprint arXiv:2405.07433}\ } (\bibinfo {year} {2024})}\BibitemShut {NoStop}%
\bibitem [{\citenamefont {Huang}\ and\ \citenamefont {Puri}(2024)}]{huang2024}%
  \BibitemOpen
  \bibfield  {author} {\bibinfo {author} {\bibfnamefont {S.}~\bibnamefont {Huang}}\ and\ \bibinfo {author} {\bibfnamefont {S.}~\bibnamefont {Puri}},\ }\bibfield  {title} {\bibinfo {title} {Increasing memory lifetime of quantum low-density parity check codes with sliding-window noisy syndrome decoding},\ }\href {https://link.aps.org/doi/10.1103/PhysRevA.110.012453} {\bibfield  {journal} {\bibinfo  {journal} {Phys. Rev. A}\ }\textbf {\bibinfo {volume} {110}},\ \bibinfo {pages} {012453} (\bibinfo {year} {2024})}\BibitemShut {NoStop}%
\bibitem [{\citenamefont {Wang}\ \emph {et~al.}(2011)\citenamefont {Wang}, \citenamefont {Fowler},\ and\ \citenamefont {Hollenberg}}]{wang2011}%
  \BibitemOpen
  \bibfield  {author} {\bibinfo {author} {\bibfnamefont {D.~S.}\ \bibnamefont {Wang}}, \bibinfo {author} {\bibfnamefont {A.~G.}\ \bibnamefont {Fowler}},\ and\ \bibinfo {author} {\bibfnamefont {L.~C.~L.}\ \bibnamefont {Hollenberg}},\ }\bibfield  {title} {\bibinfo {title} {Surface code quantum computing with error rates over 1\%},\ }\href {https://link.aps.org/doi/10.1103/PhysRevA.83.020302} {\bibfield  {journal} {\bibinfo  {journal} {Phys. Rev. A}\ }\textbf {\bibinfo {volume} {83}},\ \bibinfo {pages} {020302} (\bibinfo {year} {2011})}\BibitemShut {NoStop}%
\bibitem [{\citenamefont {Xu}\ \emph {et~al.}(2024{\natexlab{b}})\citenamefont {Xu}, \citenamefont {Ataides}, \citenamefont {Pattison}, \citenamefont {Raveendran}, \citenamefont {Bluvstein}, \citenamefont {Wurtz}, \citenamefont {Vasic}, \citenamefont {Lukin}, \citenamefont {Jiang},\ and\ \citenamefont {Zhou}}]{xu2023constantoverhead}%
  \BibitemOpen
  \bibfield  {author} {\bibinfo {author} {\bibfnamefont {Q.}~\bibnamefont {Xu}}, \bibinfo {author} {\bibfnamefont {J.~P.~B.}\ \bibnamefont {Ataides}}, \bibinfo {author} {\bibfnamefont {C.~A.}\ \bibnamefont {Pattison}}, \bibinfo {author} {\bibfnamefont {N.}~\bibnamefont {Raveendran}}, \bibinfo {author} {\bibfnamefont {D.}~\bibnamefont {Bluvstein}}, \bibinfo {author} {\bibfnamefont {J.}~\bibnamefont {Wurtz}}, \bibinfo {author} {\bibfnamefont {B.}~\bibnamefont {Vasic}}, \bibinfo {author} {\bibfnamefont {M.~D.}\ \bibnamefont {Lukin}}, \bibinfo {author} {\bibfnamefont {L.}~\bibnamefont {Jiang}},\ and\ \bibinfo {author} {\bibfnamefont {H.}~\bibnamefont {Zhou}},\ }\bibfield  {title} {\bibinfo {title} {Constant-overhead fault-tolerant quantum computation with reconfigurable atom arrays},\ }\href {https://doi.org/10.1038/s41567-024-02479-z} {\bibfield  {journal} {\bibinfo  {journal} {Nat. Phys.}\ }\textbf {\bibinfo {volume} {20}},\ \bibinfo {pages} {1084} (\bibinfo {year} {2024}{\natexlab{b}})}\BibitemShut {NoStop}%
\bibitem [{\citenamefont {Gong}\ \emph {et~al.}(2024)\citenamefont {Gong}, \citenamefont {Cammerer},\ and\ \citenamefont {Renes}}]{gong2024lowlatencyiterativedecodingqldpc}%
  \BibitemOpen
  \bibfield  {author} {\bibinfo {author} {\bibfnamefont {A.}~\bibnamefont {Gong}}, \bibinfo {author} {\bibfnamefont {S.}~\bibnamefont {Cammerer}},\ and\ \bibinfo {author} {\bibfnamefont {J.~M.}\ \bibnamefont {Renes}},\ }\bibfield  {title} {\bibinfo {title} {Toward low-latency iterative decoding of {QLDPC} codes under circuit-level noise},\ }\href {https://doi.org/10.48550/arXiv.2403.18901} {\bibfield  {journal} {\bibinfo  {journal} {arXiv preprint arXiv:2403.18901}\ } (\bibinfo {year} {2024})}\BibitemShut {NoStop}%
\bibitem [{\citenamefont {Moses}\ \emph {et~al.}(2023)\citenamefont {Moses}, \citenamefont {Baldwin}, \citenamefont {Allman}, \citenamefont {Ancona}, \citenamefont {Ascarrunz}, \citenamefont {Barnes}, \citenamefont {Bartolotta}, \citenamefont {Bjork}, \citenamefont {Blanchard} \emph {et~al.}}]{moses2023}%
  \BibitemOpen
  \bibfield  {author} {\bibinfo {author} {\bibfnamefont {S.~A.}\ \bibnamefont {Moses}}, \bibinfo {author} {\bibfnamefont {C.~H.}\ \bibnamefont {Baldwin}}, \bibinfo {author} {\bibfnamefont {M.~S.}\ \bibnamefont {Allman}}, \bibinfo {author} {\bibfnamefont {R.}~\bibnamefont {Ancona}}, \bibinfo {author} {\bibfnamefont {L.}~\bibnamefont {Ascarrunz}}, \bibinfo {author} {\bibfnamefont {C.}~\bibnamefont {Barnes}}, \bibinfo {author} {\bibfnamefont {J.}~\bibnamefont {Bartolotta}}, \bibinfo {author} {\bibfnamefont {B.}~\bibnamefont {Bjork}}, \bibinfo {author} {\bibfnamefont {P.}~\bibnamefont {Blanchard}}, \emph {et~al.},\ }\bibfield  {title} {\bibinfo {title} {A race-track trapped-ion quantum processor},\ }\href {https://link.aps.org/doi/10.1103/PhysRevX.13.041052} {\bibfield  {journal} {\bibinfo  {journal} {Phys. Rev. X}\ }\textbf {\bibinfo {volume} {13}},\ \bibinfo {pages} {041052} (\bibinfo {year} {2023})}\BibitemShut {NoStop}%
\bibitem [{\citenamefont {Bluvstein}\ \emph {et~al.}(2024)\citenamefont {Bluvstein}, \citenamefont {Evered}, \citenamefont {Geim}, \citenamefont {Li}, \citenamefont {Zhou}, \citenamefont {Manovitz}, \citenamefont {Ebadi}, \citenamefont {Cain}, \citenamefont {Kalinowski}, \citenamefont {Hangleiter} \emph {et~al.}}]{Bluvstein_2023}%
  \BibitemOpen
  \bibfield  {author} {\bibinfo {author} {\bibfnamefont {D.}~\bibnamefont {Bluvstein}}, \bibinfo {author} {\bibfnamefont {S.~J.}\ \bibnamefont {Evered}}, \bibinfo {author} {\bibfnamefont {A.~A.}\ \bibnamefont {Geim}}, \bibinfo {author} {\bibfnamefont {S.~H.}\ \bibnamefont {Li}}, \bibinfo {author} {\bibfnamefont {H.}~\bibnamefont {Zhou}}, \bibinfo {author} {\bibfnamefont {T.}~\bibnamefont {Manovitz}}, \bibinfo {author} {\bibfnamefont {S.}~\bibnamefont {Ebadi}}, \bibinfo {author} {\bibfnamefont {M.}~\bibnamefont {Cain}}, \bibinfo {author} {\bibfnamefont {M.}~\bibnamefont {Kalinowski}}, \bibinfo {author} {\bibfnamefont {D.}~\bibnamefont {Hangleiter}}, \emph {et~al.},\ }\bibfield  {title} {\bibinfo {title} {Logical quantum processor based on reconfigurable atom arrays},\ }\href {http://dx.doi.org/10.1038/s41586-023-06927-3} {\bibfield  {journal} {\bibinfo  {journal} {Nature}\ }\textbf {\bibinfo {volume} {626}},\ \bibinfo {pages} {58} (\bibinfo {year} {2024})}\BibitemShut {NoStop}%
\bibitem [{\citenamefont {Delfosse}\ \emph {et~al.}(2021{\natexlab{b}})\citenamefont {Delfosse}, \citenamefont {Beverland},\ and\ \citenamefont {Tremblay}}]{delfosse2021boundsstabilizermeasurementcircuits}%
  \BibitemOpen
  \bibfield  {author} {\bibinfo {author} {\bibfnamefont {N.}~\bibnamefont {Delfosse}}, \bibinfo {author} {\bibfnamefont {M.~E.}\ \bibnamefont {Beverland}},\ and\ \bibinfo {author} {\bibfnamefont {M.~A.}\ \bibnamefont {Tremblay}},\ }\bibfield  {title} {\bibinfo {title} {Bounds on stabilizer measurement circuits and obstructions to local implementations of quantum {LDPC} codes},\ }\href {https://doi.org/10.48550/arXiv.2109.14599} {\bibfield  {journal} {\bibinfo  {journal} {arXiv preprint arXiv:2109.14599}\ } (\bibinfo {year} {2021}{\natexlab{b}})}\BibitemShut {NoStop}%
\bibitem [{\citenamefont {Dennis}\ \emph {et~al.}(2002)\citenamefont {Dennis}, \citenamefont {Kitaev}, \citenamefont {Landahl},\ and\ \citenamefont {Preskill}}]{Dennis_2002}%
  \BibitemOpen
  \bibfield  {author} {\bibinfo {author} {\bibfnamefont {E.}~\bibnamefont {Dennis}}, \bibinfo {author} {\bibfnamefont {A.}~\bibnamefont {Kitaev}}, \bibinfo {author} {\bibfnamefont {A.}~\bibnamefont {Landahl}},\ and\ \bibinfo {author} {\bibfnamefont {J.}~\bibnamefont {Preskill}},\ }\bibfield  {title} {\bibinfo {title} {Topological quantum memory},\ }\href {http://dx.doi.org/10.1063/1.1499754} {\bibfield  {journal} {\bibinfo  {journal} {Journal of Mathematical Physics}\ }\textbf {\bibinfo {volume} {43}},\ \bibinfo {pages} {4452–4505} (\bibinfo {year} {2002})}\BibitemShut {NoStop}%
\bibitem [{\citenamefont {Manes}\ and\ \citenamefont {Claes}(2023)}]{manes2023}%
  \BibitemOpen
  \bibfield  {author} {\bibinfo {author} {\bibfnamefont {A.~G.}\ \bibnamefont {Manes}}\ and\ \bibinfo {author} {\bibfnamefont {J.}~\bibnamefont {Claes}},\ }\bibfield  {title} {\bibinfo {title} {Distance-preserving stabilizer measurements in hypergraph product codes},\ }\href {https://doi.org/10.48550/arXiv.2308.15520} {\bibfield  {journal} {\bibinfo  {journal} {arXiv preprint arXiv:2308.15520}\ } (\bibinfo {year} {2023})}\BibitemShut {NoStop}%
\bibitem [{\citenamefont {Tan}\ and\ \citenamefont {Stambler}(2024)}]{tan2024}%
  \BibitemOpen
  \bibfield  {author} {\bibinfo {author} {\bibfnamefont {S.~J.~S.}\ \bibnamefont {Tan}}\ and\ \bibinfo {author} {\bibfnamefont {L.}~\bibnamefont {Stambler}},\ }\bibfield  {title} {\bibinfo {title} {Effective distance of higher dimensional hgps and weight-reduced quantum {LDPC} codes},\ }\href {https://doi.org/10.48550/arXiv.2409.02193} {\bibfield  {journal} {\bibinfo  {journal} {arXiv preprint arXiv:2409.02193}\ } (\bibinfo {year} {2024})}\BibitemShut {NoStop}%
\bibitem [{\citenamefont {Bombín}(2015)}]{Bomb_n_2015}%
  \BibitemOpen
  \bibfield  {author} {\bibinfo {author} {\bibfnamefont {H.}~\bibnamefont {Bombín}},\ }\bibfield  {title} {\bibinfo {title} {Single-shot fault-tolerant quantum error correction},\ }\href {http://dx.doi.org/10.1103/PhysRevX.5.031043} {\bibfield  {journal} {\bibinfo  {journal} {Phys. Rev. X}\ }\textbf {\bibinfo {volume} {5}} (\bibinfo {year} {2015})}\BibitemShut {NoStop}%
\bibitem [{\citenamefont {Berthusen}\ and\ \citenamefont {Gottesman}(2024)}]{Berthusen_2023}%
  \BibitemOpen
  \bibfield  {author} {\bibinfo {author} {\bibfnamefont {N.}~\bibnamefont {Berthusen}}\ and\ \bibinfo {author} {\bibfnamefont {D.}~\bibnamefont {Gottesman}},\ }\bibfield  {title} {\bibinfo {title} {Partial syndrome measurement for hypergraph product codes},\ }\href {http://dx.doi.org/10.22331/q-2024-05-14-1345} {\bibfield  {journal} {\bibinfo  {journal} {Quantum}\ }\textbf {\bibinfo {volume} {8}},\ \bibinfo {pages} {1345} (\bibinfo {year} {2024})}\BibitemShut {NoStop}%
\bibitem [{\citenamefont {MacKay}\ and\ \citenamefont {Neal}(1997)}]{mackay1997}%
  \BibitemOpen
  \bibfield  {author} {\bibinfo {author} {\bibfnamefont {D.~J.~C.}\ \bibnamefont {MacKay}}\ and\ \bibinfo {author} {\bibfnamefont {R.~M.}\ \bibnamefont {Neal}},\ }\bibfield  {title} {\bibinfo {title} {Near shannon limit performance of low density parity check codes},\ }\href {https://doi.org/10.1049/el:19970362} {\bibfield  {journal} {\bibinfo  {journal} {Electronics Letters}\ }\textbf {\bibinfo {volume} {33}} (\bibinfo {year} {1997})}\BibitemShut {NoStop}%
\bibitem [{\citenamefont {Roffe}(2022)}]{Roffe_LDPC_Python_tools_2022}%
  \BibitemOpen
  \bibfield  {author} {\bibinfo {author} {\bibfnamefont {J.}~\bibnamefont {Roffe}},\ }\href {https://pypi.org/project/ldpc/} {\bibinfo {title} {{LDPC: Python tools for low density parity check codes}}} (\bibinfo {year} {2022})\BibitemShut {NoStop}%
\bibitem [{\citenamefont {Hillmann}\ \emph {et~al.}(2024)\citenamefont {Hillmann}, \citenamefont {Berent}, \citenamefont {Quintavalle}, \citenamefont {Eisert}, \citenamefont {Wille},\ and\ \citenamefont {Roffe}}]{hillmann2024localizedstatisticsdecodingparallel}%
  \BibitemOpen
  \bibfield  {author} {\bibinfo {author} {\bibfnamefont {T.}~\bibnamefont {Hillmann}}, \bibinfo {author} {\bibfnamefont {L.}~\bibnamefont {Berent}}, \bibinfo {author} {\bibfnamefont {A.~O.}\ \bibnamefont {Quintavalle}}, \bibinfo {author} {\bibfnamefont {J.}~\bibnamefont {Eisert}}, \bibinfo {author} {\bibfnamefont {R.}~\bibnamefont {Wille}},\ and\ \bibinfo {author} {\bibfnamefont {J.}~\bibnamefont {Roffe}},\ }\bibfield  {title} {\bibinfo {title} {Localized statistics decoding: A parallel decoding algorithm for quantum low-density parity-check codes},\ }\href {https://doi.org/10.48550/arXiv.2406.18655} {\bibfield  {journal} {\bibinfo  {journal} {arXiv preprint arXiv:2406.18655}\ } (\bibinfo {year} {2024})}\BibitemShut {NoStop}%
\bibitem [{\citenamefont {Fang}\ \emph {et~al.}(2023)\citenamefont {Fang}, \citenamefont {Zhang}, \citenamefont {Shi},\ and\ \citenamefont {Li}}]{fang2023dynamicquantumcircuitcompilation}%
  \BibitemOpen
  \bibfield  {author} {\bibinfo {author} {\bibfnamefont {K.}~\bibnamefont {Fang}}, \bibinfo {author} {\bibfnamefont {M.}~\bibnamefont {Zhang}}, \bibinfo {author} {\bibfnamefont {R.}~\bibnamefont {Shi}},\ and\ \bibinfo {author} {\bibfnamefont {Y.}~\bibnamefont {Li}},\ }\bibfield  {title} {\bibinfo {title} {Dynamic quantum circuit compilation},\ }\href {https://doi.org/10.48550/arXiv.2310.11021} {\bibfield  {journal} {\bibinfo  {journal} {arXiv preprint arXiv:2310.11021}\ } (\bibinfo {year} {2023})}\BibitemShut {NoStop}%
\bibitem [{\citenamefont {Pino}\ \emph {et~al.}(2021)\citenamefont {Pino}, \citenamefont {Dreiling}, \citenamefont {Figgatt}, \citenamefont {Gaebler}, \citenamefont {Moses}, \citenamefont {Allman}, \citenamefont {Baldwin}, \citenamefont {Foss-Feig}, \citenamefont {Hayes}, \citenamefont {Mayer}, \citenamefont {Ryan-Anderson},\ and\ \citenamefont {Neyenhuis}}]{Pino2021}%
  \BibitemOpen
  \bibfield  {author} {\bibinfo {author} {\bibfnamefont {J.~M.}\ \bibnamefont {Pino}}, \bibinfo {author} {\bibfnamefont {J.~M.}\ \bibnamefont {Dreiling}}, \bibinfo {author} {\bibfnamefont {C.}~\bibnamefont {Figgatt}}, \bibinfo {author} {\bibfnamefont {J.~P.}\ \bibnamefont {Gaebler}}, \bibinfo {author} {\bibfnamefont {S.~A.}\ \bibnamefont {Moses}}, \bibinfo {author} {\bibfnamefont {M.~S.}\ \bibnamefont {Allman}}, \bibinfo {author} {\bibfnamefont {C.~H.}\ \bibnamefont {Baldwin}}, \bibinfo {author} {\bibfnamefont {M.}~\bibnamefont {Foss-Feig}}, \bibinfo {author} {\bibfnamefont {D.}~\bibnamefont {Hayes}}, \bibinfo {author} {\bibfnamefont {K.}~\bibnamefont {Mayer}}, \bibinfo {author} {\bibfnamefont {C.}~\bibnamefont {Ryan-Anderson}},\ and\ \bibinfo {author} {\bibfnamefont {B.}~\bibnamefont {Neyenhuis}},\ }\bibfield  {title} {\bibinfo {title} {Demonstration of the trapped-ion quantum {CCD} computer architecture},\ }\href {https://doi.org/10.1038/s41586-021-03318-4} {\bibfield  {journal} {\bibinfo  {journal}
  {Nature}\ }\textbf {\bibinfo {volume} {592}},\ \bibinfo {pages} {209} (\bibinfo {year} {2021})}\BibitemShut {NoStop}%
\bibitem [{\citenamefont {Pecorari}\ \emph {et~al.}(2025)\citenamefont {Pecorari}, \citenamefont {Jandura}, \citenamefont {Brennen},\ and\ \citenamefont {Pupillo}}]{pecorari2024highratequantumldpccodes}%
  \BibitemOpen
  \bibfield  {author} {\bibinfo {author} {\bibfnamefont {L.}~\bibnamefont {Pecorari}}, \bibinfo {author} {\bibfnamefont {S.}~\bibnamefont {Jandura}}, \bibinfo {author} {\bibfnamefont {G.~K.}\ \bibnamefont {Brennen}},\ and\ \bibinfo {author} {\bibfnamefont {G.}~\bibnamefont {Pupillo}},\ }\bibfield  {title} {\bibinfo {title} {High-rate quantum ldpc codes for long-range-connected neutral atom registers},\ }\href {http://dx.doi.org/10.1038/s41467-025-56255-5} {\bibfield  {journal} {\bibinfo  {journal} {Nature Communications}\ }\textbf {\bibinfo {volume} {16}} (\bibinfo {year} {2025})}\BibitemShut {NoStop}%
\bibitem [{\citenamefont {Bravyi}\ and\ \citenamefont {Kitaev}(1998)}]{bravyi1998quantum}%
  \BibitemOpen
  \bibfield  {author} {\bibinfo {author} {\bibfnamefont {S.~B.}\ \bibnamefont {Bravyi}}\ and\ \bibinfo {author} {\bibfnamefont {A.~Y.}\ \bibnamefont {Kitaev}},\ }\bibfield  {title} {\bibinfo {title} {Quantum codes on a lattice with boundary},\ }\href {https://doi.org/10.48550/arXiv.quant-ph/9811052} {\bibfield  {journal} {\bibinfo  {journal} {arXiv prerpint arXiv:quant-ph/9811052}\ } (\bibinfo {year} {1998})}\BibitemShut {NoStop}%
\bibitem [{\citenamefont {Kitaev}(2003)}]{Kitaev_2003}%
  \BibitemOpen
  \bibfield  {author} {\bibinfo {author} {\bibfnamefont {A.}~\bibnamefont {Kitaev}},\ }\bibfield  {title} {\bibinfo {title} {Fault-tolerant quantum computation by anyons},\ }\href {https://doi.org/10.1016/s0003-4916(02)00018-0} {\bibfield  {journal} {\bibinfo  {journal} {Annals of Physics}\ }\textbf {\bibinfo {volume} {303}},\ \bibinfo {pages} {2} (\bibinfo {year} {2003})}\BibitemShut {NoStop}%
\bibitem [{\citenamefont {Bombin}\ and\ \citenamefont {Martin-Delgado}(2006)}]{Bombin_2006}%
  \BibitemOpen
  \bibfield  {author} {\bibinfo {author} {\bibfnamefont {H.}~\bibnamefont {Bombin}}\ and\ \bibinfo {author} {\bibfnamefont {M.~A.}\ \bibnamefont {Martin-Delgado}},\ }\bibfield  {title} {\bibinfo {title} {Topological quantum distillation},\ }\bibfield  {journal} {\bibinfo  {journal} {Phys. Rev. Lett.}\ }\textbf {\bibinfo {volume} {97}},\ \href {https://doi.org/10.1103/physrevlett.97.180501} {10.1103/physrevlett.97.180501} (\bibinfo {year} {2006})\BibitemShut {NoStop}%
\bibitem [{\citenamefont {Bravyi}\ and\ \citenamefont {Terhal}(2009)}]{Bravyi_Terhal_2009}%
  \BibitemOpen
  \bibfield  {author} {\bibinfo {author} {\bibfnamefont {S.}~\bibnamefont {Bravyi}}\ and\ \bibinfo {author} {\bibfnamefont {B.}~\bibnamefont {Terhal}},\ }\bibfield  {title} {\bibinfo {title} {A no-go theorem for a two-dimensional self-correcting quantum memory based on stabilizer codes},\ }\href {https://doi.org/10.1088/1367-2630/11/4/043029} {\bibfield  {journal} {\bibinfo  {journal} {New Journal of Physics}\ }\textbf {\bibinfo {volume} {11}},\ \bibinfo {pages} {043029} (\bibinfo {year} {2009})}\BibitemShut {NoStop}%
\bibitem [{\citenamefont {Bravyi}\ \emph {et~al.}(2010)\citenamefont {Bravyi}, \citenamefont {Poulin},\ and\ \citenamefont {Terhal}}]{Bravyi_Poulin_Terhal_2010}%
  \BibitemOpen
  \bibfield  {author} {\bibinfo {author} {\bibfnamefont {S.}~\bibnamefont {Bravyi}}, \bibinfo {author} {\bibfnamefont {D.}~\bibnamefont {Poulin}},\ and\ \bibinfo {author} {\bibfnamefont {B.}~\bibnamefont {Terhal}},\ }\bibfield  {title} {\bibinfo {title} {Tradeoffs for reliable quantum information storage in {2D} systems},\ }\href {https://link.aps.org/doi/10.1103/PhysRevLett.104.050503} {\bibfield  {journal} {\bibinfo  {journal} {Phys. Rev. Lett.}\ }\textbf {\bibinfo {volume} {104}},\ \bibinfo {pages} {050503} (\bibinfo {year} {2010})}\BibitemShut {NoStop}%
\bibitem [{\citenamefont {Baspin}\ and\ \citenamefont {Krishna}(2022)}]{Baspin_2022}%
  \BibitemOpen
  \bibfield  {author} {\bibinfo {author} {\bibfnamefont {N.}~\bibnamefont {Baspin}}\ and\ \bibinfo {author} {\bibfnamefont {A.}~\bibnamefont {Krishna}},\ }\bibfield  {title} {\bibinfo {title} {Quantifying nonlocality: How outperforming local quantum codes is expensive},\ }\href {https://doi.org/10.1103/PhysRevLett.129.050505} {\bibfield  {journal} {\bibinfo  {journal} {Phys. Rev. Lett.}\ }\textbf {\bibinfo {volume} {129}},\ \bibinfo {pages} {050505} (\bibinfo {year} {2022})}\BibitemShut {NoStop}%
\bibitem [{\citenamefont {Pattison}\ \emph {et~al.}(2023)\citenamefont {Pattison}, \citenamefont {Krishna},\ and\ \citenamefont {Preskill}}]{pattison2023}%
  \BibitemOpen
  \bibfield  {author} {\bibinfo {author} {\bibfnamefont {C.~A.}\ \bibnamefont {Pattison}}, \bibinfo {author} {\bibfnamefont {A.}~\bibnamefont {Krishna}},\ and\ \bibinfo {author} {\bibfnamefont {J.}~\bibnamefont {Preskill}},\ }\bibfield  {title} {\bibinfo {title} {Hierarchical memories: Simulating quantum ldpc codes with local gates},\ }\href {https://doi.org/10.48550/arXiv.2303.04798} {\bibfield  {journal} {\bibinfo  {journal} {arXiv preprint arXiv:2303.04798}\ } (\bibinfo {year} {2023})}\BibitemShut {NoStop}%
\bibitem [{\citenamefont {Choe}\ and\ \citenamefont {Koenig}(2024)}]{choe2024faulttolerantlyrealizequantumcircuit}%
  \BibitemOpen
  \bibfield  {author} {\bibinfo {author} {\bibfnamefont {S.~H.}\ \bibnamefont {Choe}}\ and\ \bibinfo {author} {\bibfnamefont {R.}~\bibnamefont {Koenig}},\ }\bibfield  {title} {\bibinfo {title} {How to fault-tolerantly realize any quantum circuit with local operations},\ }\href {https://doi.org/10.48550/arXiv.2402.13863} {\bibfield  {journal} {\bibinfo  {journal} {arXiv preprint arXiv:2402.13863}\ } (\bibinfo {year} {2024})}\BibitemShut {NoStop}%
\bibitem [{\citenamefont {Kovalev}\ and\ \citenamefont {Pryadko}(2013)}]{kovalev2013}%
  \BibitemOpen
  \bibfield  {author} {\bibinfo {author} {\bibfnamefont {A.~A.}\ \bibnamefont {Kovalev}}\ and\ \bibinfo {author} {\bibfnamefont {L.~P.}\ \bibnamefont {Pryadko}},\ }\bibfield  {title} {\bibinfo {title} {Quantum kronecker sum-product low-density parity-check codes with finite rate},\ }\href {https://doi.org/10.1103/PhysRevA.88.012311} {\bibfield  {journal} {\bibinfo  {journal} {Phys. Rev. A}\ }\textbf {\bibinfo {volume} {88}},\ \bibinfo {pages} {012311} (\bibinfo {year} {2013})}\BibitemShut {NoStop}%
\bibitem [{\citenamefont {Bravyi}\ \emph {et~al.}(2024)\citenamefont {Bravyi}, \citenamefont {Cross}, \citenamefont {Gambetta}, \citenamefont {Maslov}, \citenamefont {Rall},\ and\ \citenamefont {Yoder}}]{bravyi2023highthreshold}%
  \BibitemOpen
  \bibfield  {author} {\bibinfo {author} {\bibfnamefont {S.}~\bibnamefont {Bravyi}}, \bibinfo {author} {\bibfnamefont {A.~W.}\ \bibnamefont {Cross}}, \bibinfo {author} {\bibfnamefont {J.~M.}\ \bibnamefont {Gambetta}}, \bibinfo {author} {\bibfnamefont {D.}~\bibnamefont {Maslov}}, \bibinfo {author} {\bibfnamefont {P.}~\bibnamefont {Rall}},\ and\ \bibinfo {author} {\bibfnamefont {T.~J.}\ \bibnamefont {Yoder}},\ }\bibfield  {title} {\bibinfo {title} {High-threshold and low-overhead fault-tolerant quantum memory},\ }\href {https://doi.org/https://doi.org/10.1038/s41586-024-07107-7} {\bibfield  {journal} {\bibinfo  {journal} {Nature}\ }\textbf {\bibinfo {volume} {627}},\ \bibinfo {pages} {778} (\bibinfo {year} {2024})}\BibitemShut {NoStop}%
\bibitem [{\citenamefont {Hong}\ \emph {et~al.}(2024)\citenamefont {Hong}, \citenamefont {Marinelli}, \citenamefont {Kaufman},\ and\ \citenamefont {Lucas}}]{Hong_2024}%
  \BibitemOpen
  \bibfield  {author} {\bibinfo {author} {\bibfnamefont {Y.}~\bibnamefont {Hong}}, \bibinfo {author} {\bibfnamefont {M.}~\bibnamefont {Marinelli}}, \bibinfo {author} {\bibfnamefont {A.~M.}\ \bibnamefont {Kaufman}},\ and\ \bibinfo {author} {\bibfnamefont {A.}~\bibnamefont {Lucas}},\ }\bibfield  {title} {\bibinfo {title} {Long-range-enhanced surface codes},\ }\href {http://dx.doi.org/10.1103/PhysRevA.110.022607} {\bibfield  {journal} {\bibinfo  {journal} {Phys. Rev. A}\ }\textbf {\bibinfo {volume} {110}} (\bibinfo {year} {2024})}\BibitemShut {NoStop}%
\bibitem [{\citenamefont {Monroe}\ \emph {et~al.}(2014)\citenamefont {Monroe}, \citenamefont {Raussendorf}, \citenamefont {Ruthven}, \citenamefont {Brown}, \citenamefont {Maunz}, \citenamefont {Duan},\ and\ \citenamefont {Kim}}]{monroe2014}%
  \BibitemOpen
  \bibfield  {author} {\bibinfo {author} {\bibfnamefont {C.}~\bibnamefont {Monroe}}, \bibinfo {author} {\bibfnamefont {R.}~\bibnamefont {Raussendorf}}, \bibinfo {author} {\bibfnamefont {A.}~\bibnamefont {Ruthven}}, \bibinfo {author} {\bibfnamefont {K.~R.}\ \bibnamefont {Brown}}, \bibinfo {author} {\bibfnamefont {P.}~\bibnamefont {Maunz}}, \bibinfo {author} {\bibfnamefont {L.-M.}\ \bibnamefont {Duan}},\ and\ \bibinfo {author} {\bibfnamefont {J.}~\bibnamefont {Kim}},\ }\bibfield  {title} {\bibinfo {title} {Large-scale modular quantum-computer architecture with atomic memory and photonic interconnects},\ }\href {https://doi.org/10.1103/PhysRevA.89.022317} {\bibfield  {journal} {\bibinfo  {journal} {Phys. Rev. A}\ }\textbf {\bibinfo {volume} {89}},\ \bibinfo {pages} {022317} (\bibinfo {year} {2014})}\BibitemShut {NoStop}%
\bibitem [{\citenamefont {Strikis}\ and\ \citenamefont {Berent}(2023)}]{Strikis_2023}%
  \BibitemOpen
  \bibfield  {author} {\bibinfo {author} {\bibfnamefont {A.}~\bibnamefont {Strikis}}\ and\ \bibinfo {author} {\bibfnamefont {L.}~\bibnamefont {Berent}},\ }\bibfield  {title} {\bibinfo {title} {Quantum low-density parity-check codes for modular architectures},\ }\bibfield  {journal} {\bibinfo  {journal} {PRX Quantum}\ }\textbf {\bibinfo {volume} {4}},\ \href {https://doi.org/10.1103/prxquantum.4.020321} {10.1103/prxquantum.4.020321} (\bibinfo {year} {2023})\BibitemShut {NoStop}%
\bibitem [{\citenamefont {Zhu}(2025)}]{zhu2024topological}%
  \BibitemOpen
  \bibfield  {author} {\bibinfo {author} {\bibfnamefont {G.}~\bibnamefont {Zhu}},\ }\bibfield  {title} {\bibinfo {title} {A topological theory for q{LDPC}: non-clifford gates and magic state fountain on homological product codes with constant rate and beyond the $n^{1/3}$ distance barrier},\ }\href {https://doi.org/10.48550/arXiv.2501.19375} {\bibfield  {journal} {\bibinfo  {journal} {arXiv preprint arXiv:2501.19375}\ } (\bibinfo {year} {2025})}\BibitemShut {NoStop}%
\bibitem [{\citenamefont {Malcolm}\ \emph {et~al.}(2025)\citenamefont {Malcolm}, \citenamefont {Glaudell}, \citenamefont {Fuentes}, \citenamefont {Chandra}, \citenamefont {Schotte}, \citenamefont {DeLisle}, \citenamefont {Haenel}, \citenamefont {Ebrahimi}, \citenamefont {Roffe}, \citenamefont {Quintavalle}, \citenamefont {Beale}, \citenamefont {Lee-Hone},\ and\ \citenamefont {Simmons}}]{malcolm2025}%
  \BibitemOpen
  \bibfield  {author} {\bibinfo {author} {\bibfnamefont {A.~J.}\ \bibnamefont {Malcolm}}, \bibinfo {author} {\bibfnamefont {A.~N.}\ \bibnamefont {Glaudell}}, \bibinfo {author} {\bibfnamefont {P.}~\bibnamefont {Fuentes}}, \bibinfo {author} {\bibfnamefont {D.}~\bibnamefont {Chandra}}, \bibinfo {author} {\bibfnamefont {A.}~\bibnamefont {Schotte}}, \bibinfo {author} {\bibfnamefont {C.}~\bibnamefont {DeLisle}}, \bibinfo {author} {\bibfnamefont {R.}~\bibnamefont {Haenel}}, \bibinfo {author} {\bibfnamefont {A.}~\bibnamefont {Ebrahimi}}, \bibinfo {author} {\bibfnamefont {J.}~\bibnamefont {Roffe}}, \bibinfo {author} {\bibfnamefont {A.~O.}\ \bibnamefont {Quintavalle}}, \bibinfo {author} {\bibfnamefont {S.~J.}\ \bibnamefont {Beale}}, \bibinfo {author} {\bibfnamefont {N.~R.}\ \bibnamefont {Lee-Hone}},\ and\ \bibinfo {author} {\bibfnamefont {S.}~\bibnamefont {Simmons}},\ }\bibfield  {title} {\bibinfo {title} {Computing efficiently in {QLDPC} codes},\ }\href {https://doi.org/10.48550/arXiv.2502.07150} {\bibfield  {journal}
  {\bibinfo  {journal} {arXiv preprint arXiv:2502.07150}\ } (\bibinfo {year} {2025})}\BibitemShut {NoStop}%
\bibitem [{\citenamefont {Panteleev}\ and\ \citenamefont {Kalachev}(2021)}]{Panteleev_2021}%
  \BibitemOpen
  \bibfield  {author} {\bibinfo {author} {\bibfnamefont {P.}~\bibnamefont {Panteleev}}\ and\ \bibinfo {author} {\bibfnamefont {G.}~\bibnamefont {Kalachev}},\ }\bibfield  {title} {\bibinfo {title} {Degenerate quantum {LDPC} codes with good finite length performance},\ }\href {https://doi.org/10.22331/q-2021-11-22-585} {\bibfield  {journal} {\bibinfo  {journal} {Quantum}\ }\textbf {\bibinfo {volume} {5}},\ \bibinfo {pages} {585} (\bibinfo {year} {2021})}\BibitemShut {NoStop}%
\bibitem [{\citenamefont {Panteleev}\ and\ \citenamefont {Kalachev}(2022)}]{Panteleev_2022}%
  \BibitemOpen
  \bibfield  {author} {\bibinfo {author} {\bibfnamefont {P.}~\bibnamefont {Panteleev}}\ and\ \bibinfo {author} {\bibfnamefont {G.}~\bibnamefont {Kalachev}},\ }\bibfield  {title} {\bibinfo {title} {Quantum ldpc codes with almost linear minimum distance},\ }\href {https://doi.org/10.1109/tit.2021.3119384} {\bibfield  {journal} {\bibinfo  {journal} {IEEE Transactions on Information Theory}\ }\textbf {\bibinfo {volume} {68}},\ \bibinfo {pages} {213–229} (\bibinfo {year} {2022})}\BibitemShut {NoStop}%
\bibitem [{\citenamefont {Grospellier}\ \emph {et~al.}(2021)\citenamefont {Grospellier}, \citenamefont {Grouès}, \citenamefont {Krishna},\ and\ \citenamefont {Leverrier}}]{Grospellier_2021}%
  \BibitemOpen
  \bibfield  {author} {\bibinfo {author} {\bibfnamefont {A.}~\bibnamefont {Grospellier}}, \bibinfo {author} {\bibfnamefont {L.}~\bibnamefont {Grouès}}, \bibinfo {author} {\bibfnamefont {A.}~\bibnamefont {Krishna}},\ and\ \bibinfo {author} {\bibfnamefont {A.}~\bibnamefont {Leverrier}},\ }\bibfield  {title} {\bibinfo {title} {Combining hard and soft decoders for hypergraph product codes},\ }\href {http://dx.doi.org/10.22331/q-2021-04-15-432} {\bibfield  {journal} {\bibinfo  {journal} {Quantum}\ }\textbf {\bibinfo {volume} {5}},\ \bibinfo {pages} {432} (\bibinfo {year} {2021})}\BibitemShut {NoStop}%
\bibitem [{\citenamefont {Kovalev}\ \emph {et~al.}(2018)\citenamefont {Kovalev}, \citenamefont {Prabhakar}, \citenamefont {Dumer},\ and\ \citenamefont {Pryadko}}]{Kovalev_2018}%
  \BibitemOpen
  \bibfield  {author} {\bibinfo {author} {\bibfnamefont {A.~A.}\ \bibnamefont {Kovalev}}, \bibinfo {author} {\bibfnamefont {S.}~\bibnamefont {Prabhakar}}, \bibinfo {author} {\bibfnamefont {I.}~\bibnamefont {Dumer}},\ and\ \bibinfo {author} {\bibfnamefont {L.~P.}\ \bibnamefont {Pryadko}},\ }\bibfield  {title} {\bibinfo {title} {Numerical and analytical bounds on threshold error rates for hypergraph-product codes},\ }\href {http://dx.doi.org/10.1103/PhysRevA.97.062320} {\bibfield  {journal} {\bibinfo  {journal} {Phys. Rev. A}\ }\textbf {\bibinfo {volume} {97}} (\bibinfo {year} {2018})}\BibitemShut {NoStop}%
\bibitem [{\citenamefont {Tremblay}\ \emph {et~al.}(2022)\citenamefont {Tremblay}, \citenamefont {Delfosse},\ and\ \citenamefont {Beverland}}]{Tremblay_2022}%
  \BibitemOpen
  \bibfield  {author} {\bibinfo {author} {\bibfnamefont {M.~A.}\ \bibnamefont {Tremblay}}, \bibinfo {author} {\bibfnamefont {N.}~\bibnamefont {Delfosse}},\ and\ \bibinfo {author} {\bibfnamefont {M.~E.}\ \bibnamefont {Beverland}},\ }\bibfield  {title} {\bibinfo {title} {Constant-overhead quantum error correction with thin planar connectivity},\ }\href {http://dx.doi.org/10.1103/PhysRevLett.129.050504} {\bibfield  {journal} {\bibinfo  {journal} {Phys. Rev. Lett.}\ }\textbf {\bibinfo {volume} {129}} (\bibinfo {year} {2022})}\BibitemShut {NoStop}%
\bibitem [{\citenamefont {Gidney}(2021)}]{Gidney_2021}%
  \BibitemOpen
  \bibfield  {author} {\bibinfo {author} {\bibfnamefont {C.}~\bibnamefont {Gidney}},\ }\bibfield  {title} {\bibinfo {title} {Stim: a fast stabilizer circuit simulator},\ }\href {http://dx.doi.org/10.22331/q-2021-07-06-497} {\bibfield  {journal} {\bibinfo  {journal} {Quantum}\ }\textbf {\bibinfo {volume} {5}},\ \bibinfo {pages} {497} (\bibinfo {year} {2021})}\BibitemShut {NoStop}%
\bibitem [{\citenamefont {Pattison}(2024)}]{Pattison_2024}%
  \BibitemOpen
  \bibfield  {author} {\bibinfo {author} {\bibfnamefont {C.}~\bibnamefont {Pattison}},\ }\href {https://github.com/qldpc/exp_ldpc} {\bibinfo {title} {{QEC} utilities for practical realizations of general q{LDPC} codes}} (\bibinfo {year} {2024})\BibitemShut {NoStop}%
\bibitem [{\citenamefont {Gottesman}(1998)}]{gottesman1998theory}%
  \BibitemOpen
  \bibfield  {author} {\bibinfo {author} {\bibfnamefont {D.}~\bibnamefont {Gottesman}},\ }\bibfield  {title} {\bibinfo {title} {Theory of fault-tolerant quantum computation},\ }\href {https://doi.org/10.1103/PhysRevA.57.127} {\bibfield  {journal} {\bibinfo  {journal} {Phys. Rev. A}\ }\textbf {\bibinfo {volume} {57}},\ \bibinfo {pages} {127} (\bibinfo {year} {1998})}\BibitemShut {NoStop}%
\bibitem [{\citenamefont {Sayginel}\ \emph {et~al.}(2024)\citenamefont {Sayginel}, \citenamefont {Koutsioumpas}, \citenamefont {Webster}, \citenamefont {Rajput},\ and\ \citenamefont {Browne}}]{sayginel2024fault}%
  \BibitemOpen
  \bibfield  {author} {\bibinfo {author} {\bibfnamefont {H.}~\bibnamefont {Sayginel}}, \bibinfo {author} {\bibfnamefont {S.}~\bibnamefont {Koutsioumpas}}, \bibinfo {author} {\bibfnamefont {M.}~\bibnamefont {Webster}}, \bibinfo {author} {\bibfnamefont {A.}~\bibnamefont {Rajput}},\ and\ \bibinfo {author} {\bibfnamefont {D.~E.}\ \bibnamefont {Browne}},\ }\bibfield  {title} {\bibinfo {title} {Fault-tolerant logical clifford gates from code automorphisms},\ }\href {https://doi.org/10.48550/arXiv.2409.18175} {\bibfield  {journal} {\bibinfo  {journal} {arXiv preprint arXiv:2409.18175}\ } (\bibinfo {year} {2024})}\BibitemShut {NoStop}%
\bibitem [{\citenamefont {Campbell}(2019)}]{campbell2019theory}%
  \BibitemOpen
  \bibfield  {author} {\bibinfo {author} {\bibfnamefont {E.~T.}\ \bibnamefont {Campbell}},\ }\bibfield  {title} {\bibinfo {title} {A theory of single-shot error correction for adversarial noise},\ }\href {http://dx.doi.org/10.1088/2058-9565/aafc8f} {\bibfield  {journal} {\bibinfo  {journal} {Quantum Science and Technology}\ }\textbf {\bibinfo {volume} {4}},\ \bibinfo {pages} {025006} (\bibinfo {year} {2019})}\BibitemShut {NoStop}%
\bibitem [{\citenamefont {Quintavalle}\ \emph {et~al.}(2021)\citenamefont {Quintavalle}, \citenamefont {Vasmer}, \citenamefont {Roffe},\ and\ \citenamefont {Campbell}}]{quintavalle2021single}%
  \BibitemOpen
  \bibfield  {author} {\bibinfo {author} {\bibfnamefont {A.~O.}\ \bibnamefont {Quintavalle}}, \bibinfo {author} {\bibfnamefont {M.}~\bibnamefont {Vasmer}}, \bibinfo {author} {\bibfnamefont {J.}~\bibnamefont {Roffe}},\ and\ \bibinfo {author} {\bibfnamefont {E.~T.}\ \bibnamefont {Campbell}},\ }\bibfield  {title} {\bibinfo {title} {Single-shot error correction of three-dimensional homological product codes},\ }\href {https://link.aps.org/doi/10.1103/PRXQuantum.2.020340} {\bibfield  {journal} {\bibinfo  {journal} {PRX Quantum}\ }\textbf {\bibinfo {volume} {2}},\ \bibinfo {pages} {020340} (\bibinfo {year} {2021})}\BibitemShut {NoStop}%
\bibitem [{\citenamefont {Anshu}\ \emph {et~al.}(2023)\citenamefont {Anshu}, \citenamefont {Breuckmann},\ and\ \citenamefont {Nirkhe}}]{anshu2023nlts}%
  \BibitemOpen
  \bibfield  {author} {\bibinfo {author} {\bibfnamefont {A.}~\bibnamefont {Anshu}}, \bibinfo {author} {\bibfnamefont {N.~P.}\ \bibnamefont {Breuckmann}},\ and\ \bibinfo {author} {\bibfnamefont {C.}~\bibnamefont {Nirkhe}},\ }\bibfield  {title} {\bibinfo {title} {{NLTS} hamiltonians from good quantum codes},\ }in\ \href {https://doi.org/10.1145/3564246.3585114} {\emph {\bibinfo {booktitle} {Proceedings of the 55th Annual ACM Symposium on Theory of Computing}}},\ \bibinfo {series and number} {STOC 2023}\ (\bibinfo  {publisher} {Association for Computing Machinery},\ \bibinfo {address} {New York, NY, USA},\ \bibinfo {year} {2023})\ p.\ \bibinfo {pages} {1090–1096}\BibitemShut {NoStop}%
\bibitem [{\citenamefont {Breuckmann}\ and\ \citenamefont {Burton}(2024)}]{breuckmann2024fold}%
  \BibitemOpen
  \bibfield  {author} {\bibinfo {author} {\bibfnamefont {N.~P.}\ \bibnamefont {Breuckmann}}\ and\ \bibinfo {author} {\bibfnamefont {S.}~\bibnamefont {Burton}},\ }\bibfield  {title} {\bibinfo {title} {Fold-transversal clifford gates for quantum codes},\ }\href {http://dx.doi.org/10.22331/q-2024-06-13-1372} {\bibfield  {journal} {\bibinfo  {journal} {Quantum}\ }\textbf {\bibinfo {volume} {8}},\ \bibinfo {pages} {1372} (\bibinfo {year} {2024})}\BibitemShut {NoStop}%
\bibitem [{\citenamefont {Chao}\ and\ \citenamefont {Reichardt}(2018)}]{chao2018fault}%
  \BibitemOpen
  \bibfield  {author} {\bibinfo {author} {\bibfnamefont {R.}~\bibnamefont {Chao}}\ and\ \bibinfo {author} {\bibfnamefont {B.~W.}\ \bibnamefont {Reichardt}},\ }\bibfield  {title} {\bibinfo {title} {Fault-tolerant quantum computation with few qubits},\ }\href {https://doi.org/10.1038/s41534-018-0085-z} {\bibfield  {journal} {\bibinfo  {journal} {npj Quantum Information}\ }\textbf {\bibinfo {volume} {4}},\ \bibinfo {pages} {42} (\bibinfo {year} {2018})}\BibitemShut {NoStop}%
\end{thebibliography}%

\appendix

\section{Hypergraph product codes}
\label{apx:codes}





\subsection{Quantum expander codes}


Quantum expander codes~\cite{Leverrier_2015} are families of hypergraph product codes where the base classical code(s) are randomly generated LDPC codes. We focus on square quantum expander codes where the Tanner graph of the single base classical code $H$ is randomly generated and $(\Delta_V, \Delta_C)$-regular, meaning that each bit is involved in $\Delta_V$ checks, and each check is supported on $\Delta_C$ bits. Random classical codes generated in this way are good expanders with high probability~\cite{Sipser_Spielman_1996}, and they have asymptotic parameters scaling like $[n, \Theta(n), \Theta(n)]$. The resulting HGP code, is $(\Delta_C, \Delta_V + \Delta_C)$-qLDPC (but not regular), and has parameters scaling like $[[n, O(n), O(\sqrt{n})]]$. Using the configuration model and edge swapping~\cite{Grospellier_2021}, we generate random $(3,4)$-regular classical codes which yield a family of hypergraph product codes with rate $k/n \ge 1/25$ and $[[4,2,2]]$-concatenated HGP codes with rate $k/n \ge 1/50$. 

\subsection{La-cross codes}
\label{apx:lacross}

La-cross codes~\cite{pecorari2024highratequantumldpccodes} are families of hypergraph product codes where the parity check matrix of the base classical code is circulant. That is, the parity check matrix consists of cyclic shifts of a $n$-dimensional vector, $(c_0, c_1, ..., c_{n-1}) \in \mathbb{F}_2^n$. The full parity check matrix $H \in \mathbb{F}_2^{n \times n}$ is then:
\begin{equation}
    H = \begin{pmatrix}
c_0 & c_1 & c_2 & \cdots & c_{n-1} \\
c_{n-1} & c_0 & c_1 &  &  \\
c_{n-2} & c_{n-1} & c_0 & & \vdots \\
\vdots &  &  & \ddots &  \\
c_1 & c_2 & c_3 & \cdots & c_0 \\
\end{pmatrix}
\end{equation}

The seed vector can be represented with a degree-$n$ polynomial of the form $h(x) = \sum_{i=0}^{n-1} c_i x^i$. In this work, we follow Ref.~\cite{pecorari2024highratequantumldpccodes} and consider codes whose polynomials take the form $h(x) = 1+x+x^z$. By taking the first $n-z$ rows of $H$, this yields a new parity check matrix $H' \in \mathbb{F}_2^{(n-z) \times n}$. The corresponding classical code has parameters $[n, z, d]$ and the square HGP code, HGP$(H',H')$ has parameters $[[n^2 + (n-z)^2, z^2, d]]$. For a fixed $z$, La-cross codes have a rate $k/n$ that goes to zero with increasing blocklength.

\subsection{Surface codes}

We note for completeness that surface codes~\cite{bravyi1998quantum, Kitaev_2003} can be interpreted as the hypergraph product of two repetition codes. The following $(n-1) \times n$ parity check matrix,
\begin{equation}
\renewcommand\arraystretch{1}
    H = \begin{pmatrix} 
    1 & 1 & &&&  \\
     & 1 & 1 &  &&& \\
     & & & \ddots && & \\
     &  &     &&1 & 1 
    \end{pmatrix}
\end{equation}
corresponds to the repetition code with parameters $[n,1,n]$. Taking the hypergraph product of two repetition codes thus yields a quantum code with parameters $[[n^2, 1, n]]$ and is equivalent to the surface code.

\section{Numerical simulations}
\label{apx:sims}

In this appendix, we provide additional details on the decoders and simulation for the memory experiments presented in the main text.

Previous works have investigated HGP codes formed from random $(3,4)$-regular classical codes in the code capacity~\cite{Kovalev_2018, Grospellier_2021}, phenemological~\cite{Grospellier_2021}, and circuit-level~\cite{delfosse2021boundsstabilizermeasurementcircuits, Tremblay_2022, xu2023constantoverhead} noise models. 
La-cross codes have also been studied in a modified circuit-level noise model~\cite{pecorari2024highratequantumldpccodes}.
We follow Refs.~\cite{delfosse2021boundsstabilizermeasurementcircuits, Tremblay_2022} in that while we have circuit-level noise, we decode using a phenomenological decoder instead of a circuit-level decoder~\cite{wang2011, xu2023constantoverhead, gong2024lowlatencyiterativedecodingqldpc}. This choice leads to potentially worse performance but is much easier to simulate. Decoding is then done by passing the parity check matrix and the observed syndrome to a decoder, which attempts to output a valid correction. To be able to handle syndrome and circuit-level noise, we follow Ref.~\cite{Grospellier_2021} to modify the Tanner graph/parity check matrix: for each check node, we add a single bit node which represents the possibility of a syndrome error. Note that when applying the correction to the system these syndrome errors are ignored. With this change to the parity check matrix, decoding is functionally the same in both the phenomenological and circuit-level noise models. 
We use two different decoders based on the current QEC round. For intermediate rounds we use only the `product-sum' BP~\cite{mackay1997} variant, which consists of 30 iterations of serial~\cite{Roffe_LDPC_Python_tools_2022} message passing. Whether or not BP converges to a valid solution, we apply the guessed correction and continue the simulation. The final noiseless syndrome is decoded using belief propagation plus localized statistics decoding (BP-LSD)~\cite{hillmann2024localizedstatisticsdecodingparallel}, which is based on belief propagation plus ordered statistics decoding (BP-OSD)~\cite{Panteleev_2021}. We again allow only 30 iterations of product-sum BP with serial scheduling, but if a valid correction is not obtained, order-4 combination sweep LSD is used. The process for decoding the $[[4,2,2]]$-concatenated HGP codes is identical after converting the physical errors into logical errors.

\begin{figure}
    \centering
    \includegraphics[width=0.9\linewidth]{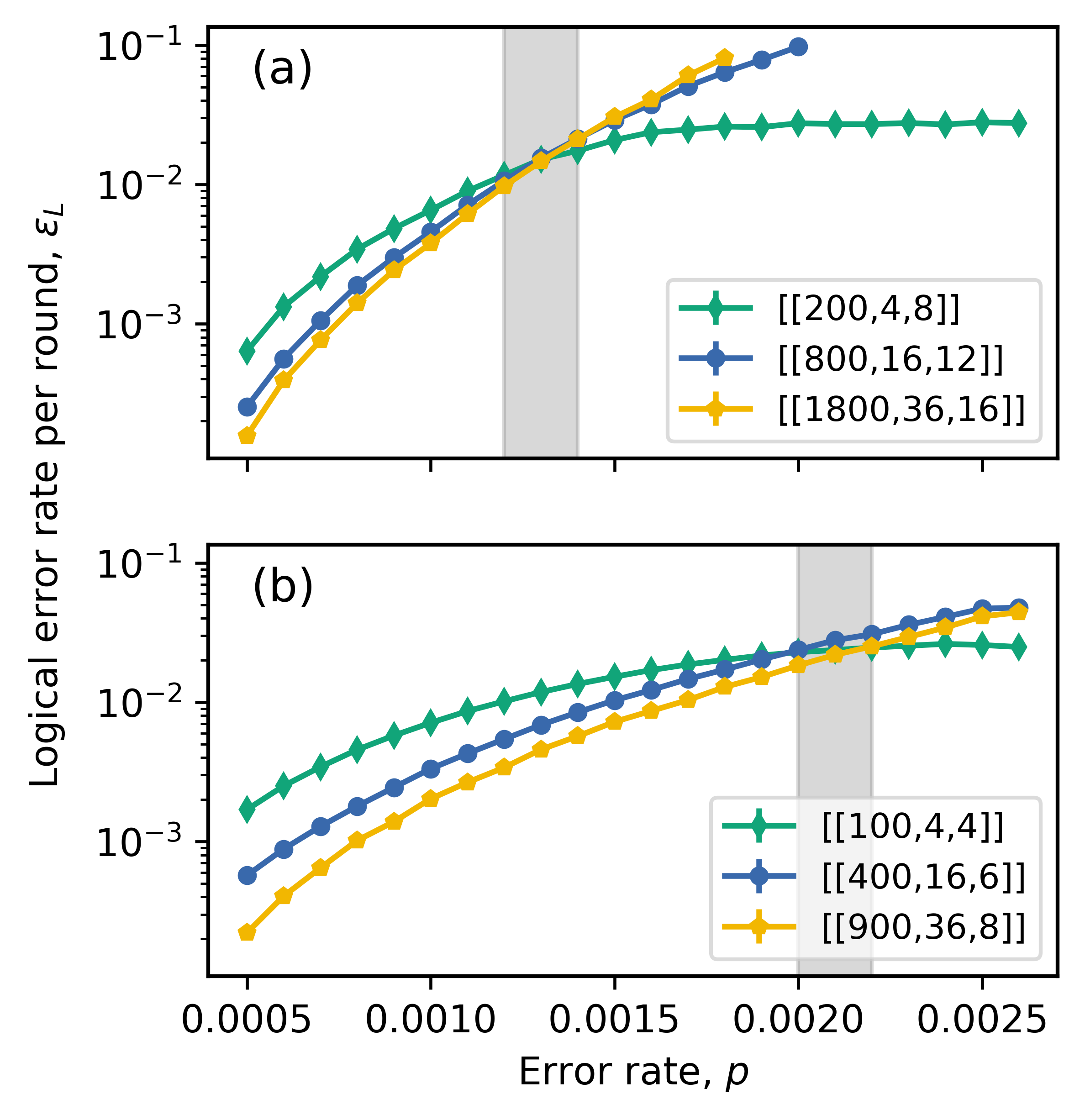}
    \caption{Logical error rate per round $\epsilon_L$ as a function of the physical error rate $p$ around the pseudothreshold. (a) $[[4,2,2]]$-concatenated quantum expander codes using the adaptive scheme exhibit a threshold around 0.13\%. (b) Non-concatenated, non-adaptive quantum expander codes display a higher pseudothreshold around 0.21\%, which is comparable to that the value of 0.23\% reported in Ref~\cite{delfosse2021boundsstabilizermeasurementcircuits} for a similar noise model and decoder. Again, $r = 100$ rounds were used as the total simulation time.}
    \label{fig:pseudothreshold}
\end{figure}

To perform the circuit-level memory experiments, we use \texttt{Stim}~\cite{Gidney_2021}. Since the adaptive scheme requires us to make decisions in real-time based on the outcomes of stabilizer generator measurements, we have to use the slower \texttt{TableauSimulator} which allows us to access the measurement results. 
The syndrome extraction circuits for the HGP codes and the concatenated HPG generators are obtained from an edge coloration of the Tanner graph of the parity check matrix, with which we can apply accurate idling error, see Ref.~\cite{delfosse2021boundsstabilizermeasurementcircuits} for details. 
We use an implementation of the bipartite graph edge coloration algorithm by Pattison~\cite{Pattison_2024}. 
It is possible to record the executed circuit and generate a detector error model which could then be used to perform circuit-level decoding; however, we found that generating the detector error model and corresponding parity check matrix yielded simulation times on the order of seconds per shot. Nonetheless, making the change to a space-time circuit-level decoder would likely improve the performance of the scheme~\cite{xu2023constantoverhead}. 

From the memory experiments as described in the main text, we obtain the logical error rate $p_L$, which is the probability of having a logical error on \textit{any} of the logical qubits. The standard deviation on $p_L$ is the standard error when sampling from a binomial distribution, $\sigma_{p_L} = \sqrt{p_L (1-p_L) / N}$, where $N$ is the number of collected shots. From $p_L$, we can calculate the logical error rate per round $\epsilon_L = 1 - (1- p_L(r))^{1/r}$, where $p_L(r)$ is the observed logical error rate at round $r$. The standard deviation on $\epsilon_L$, which is shown as the error bars in the plots here and in Sec.~\ref{sec:numerics} is then given as: 
 \begin{equation}
    \label{eq:ler_per_round_error}
     \sigma_{\epsilon_L} = \left( \frac{\partial \epsilon_L}{\partial p_L}\right) \sigma_{p_L} = \frac{1}{r} (1 - p_L)^{1/r-1} \sigma_{p_L}.
 \end{equation}

 \begin{figure}
    \centering
    \includegraphics[width=0.9\linewidth]{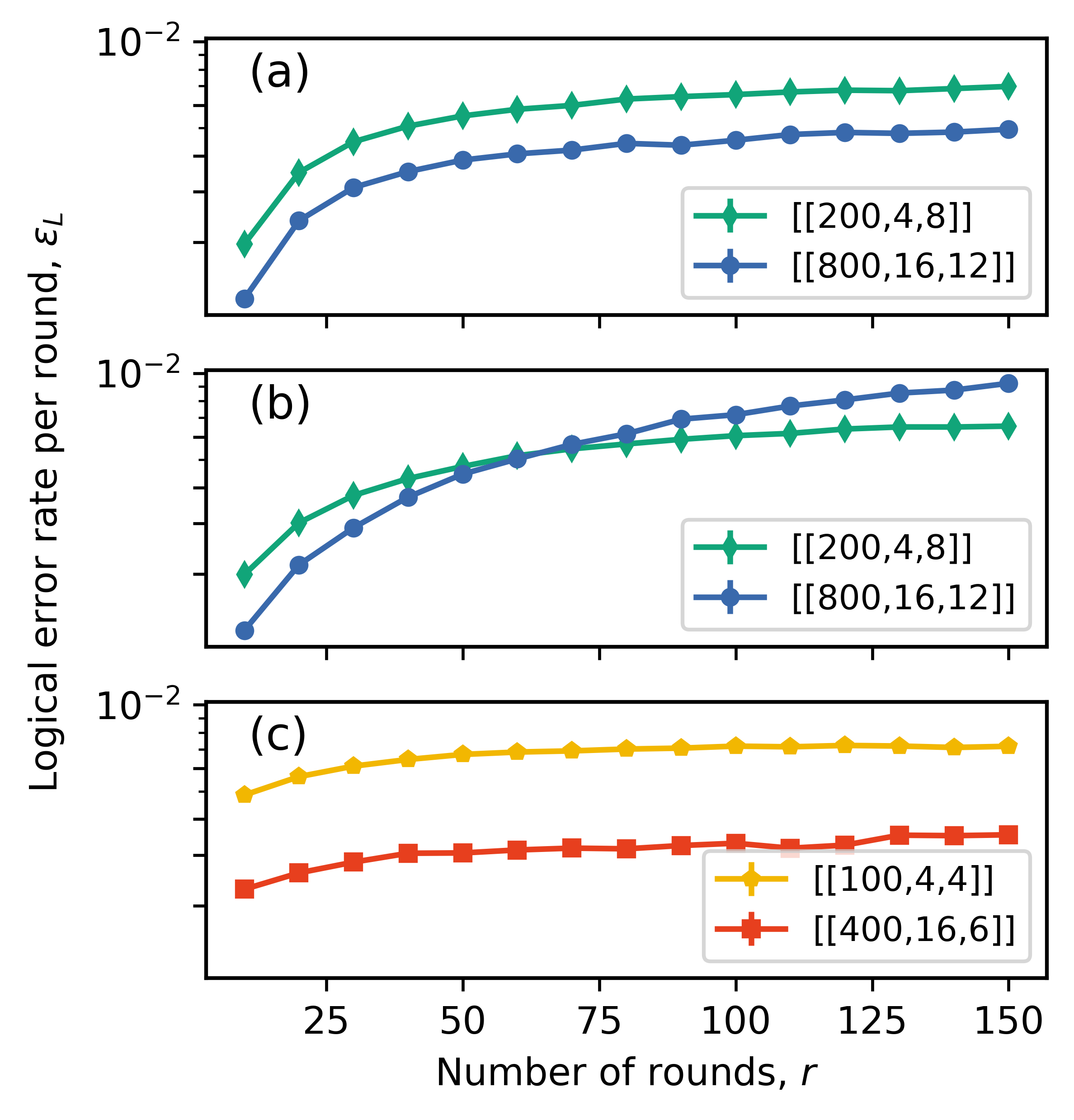}
    \caption{Logical error rate per round $\epsilon_L$ as a function of the number of rounds of quantum error correction. The physical error rate $p$ is fixed at 0.1\%. (a) $[[4,2,2]]$-concatenated quantum expander codes using the adaptive scheme where the entire set of generators is measured every 10 rounds. (b) $[[4,2,2]]$-concatenated HGP codes using the adaptive scheme where generators are never fully unmasked. (c) Normal quantum expander codes. }
    \label{fig:ft_ler_per_round}
\end{figure}

Initially, we were not able to obtain single-shot performance out of the $[[4,2,2]]$-quantum expander codes. For single-shot decoding, we would expect to see the logical error rate stabilize when increasing the number of rounds, as is the case for normal quantum expander codes, Fig.~\ref{fig:ft_ler_per_round}(c); however, the $[[4,2,2]]$-concatenated HGP codes exhibit an increasing logical error rate per round. To remedy this, we measure the entire set of stabilizer generators every 10 rounds. This effect of this change is shown in Fig.~\ref{fig:ft_ler_per_round}(a), where we now see evidence that the $[[4,2,2]]$-concatenated codes are single-shot.

\section{Logical operators for the $[[4,2,2]]$ code}
\label{sec:Iceberg_logical_ops}

\subsection{Logical Clifford operators}

We first discuss the different notions of gate transversality before introducing the various logical Clifford operators obtainable when we restrict ourselves to the various extents of transversality.
Traditionally, transversal logical gates demand that every physical operation acts on one physical qubit in a code block.
They form the foundation for fault-tolerant gates in quantum computation~\cite{gottesman1998theory}.
However, it is possible to relax this restrictive definition of transversality to what is known as \emph{SWAP-transversal} where we allow qubit permutations as well as single-qubit Clifford gates~\cite{sayginel2024fault}.
This notion of transversality is not inherently fault-tolerant except on platforms that have the ability to physically shuttle qubits, such as ion-traps and neutral atom arrays.

In the recent work of Sayginel \emph{et al.}~\cite{sayginel2024fault}, they derived a set of Clifford logical operators as shown in Table~\ref{tab:clifford_transversal_logical_ops_baby_Iceberg}.
\begin{table}[H]
\centering
    \begin{tabular}{||c|c|c||}
    \toprule
    Gate  & Circuit & Type \\
    \midrule
    $\overline{H}_1 \overline{H}_2 \overline{\emph{SWAP}}_{1,2}$ & $H_1 H_2 H_3 H_4$ & \text{Transversal} \\
    \hline
    $\overline{CZ}$ & $S_1^\dagger S_2^\dagger S_3 S_4$ & \text{Transversal} \\
    \hline
    $\overline{CNOT}_{1,2}$ & $\emph{SWAP}_{2,3}$ & $\emph{SWAP}$-transversal \\
    \hline
    $\overline{CNOT}_{2,1}$ & $\emph{SWAP}_{2,4}$ & $\emph{SWAP}$-transversal \\
    \hline
    $\overline{\emph{SWAP}}_{1,2}$ & $\emph{SWAP}_{3,4}$ & \emph{SWAP}-transversal \\
    \bottomrule    
    \end{tabular}
    \caption{\emph{SWAP}-transversal logical Clifford gates of the $[[4,2,2]]$ code.}
    \label{tab:clifford_transversal_logical_ops_baby_Iceberg}
\end{table}
Because the logical Hadamard operator does not target individual logical qubits of the $[[4,2,2]]$ code, the set of \emph{SWAP}-transversal logical Clifford gates in Table~\ref{tab:clifford_transversal_logical_ops_baby_Iceberg} does not give us the full logical Clifford group. By including the gates shown in Table~\ref{tab:clifford_logical_ops_baby_Iceberg2}, we are able to obtain a full set of logical Clifford operators; however, these gates are not transversal or \emph{SWAP}-transversal and as such are not inherently fault-tolerant.


\begin{table}[H]
\centering
    \begin{tabular}{||c|c|c||}
    \toprule
    \textbf{Gate}  & \textbf{Circuit} & \textbf{Type} \\
    \midrule
    $\overline{S}_1$ & $S_1 S_3 CZ_{1,3}$ & \text{General Clifford} \\
    \hline
    $\overline{S}_2$ & $S_1 S_4 CZ_{1,4}$ & \text{General Clifford} \\
    \hline
    $\overline{\sqrt{X}}_{1}$ & $\sqrt{X}_1 \sqrt{X}_4 C(X, X)_{1,4}$ & \text{General Clifford}\\
    \hline
    $\overline{\sqrt{X}}_{2}$ & $\sqrt{X}_1 \sqrt{X}_3 C(X, X)_{1,3}$ & \text{General Clifford}\\
    \bottomrule    
    \end{tabular}
    \caption{Logical Clifford gates of the $[[4,2,2]]$ code via embedded code technique.}
    \label{tab:clifford_logical_ops_baby_Iceberg2}
\end{table}

\subsection{State preparation}
\label{sec:state_prep_Iceberg_codes}
In this section, we describe how to prepare some of the common logical states for the $[[4,2,2]]$ code that we use for our logical computation protocol.
Because the $[[4,2,2]]$ code is a CSS code, it is trivial to initialize it in the logical states $\ket{\overline{00}}$ and $\ket{\overline{++}}$.
Thus, in this section, we focus mainly on describing how we might be able to initialize it in the $\ket{\overline{0+}}$ state.
This task is less trivial because the $[[4,2,2]]$ code only has global Hadamard as a transversal logical operator instead of one that targets a single logical qubit as shown in Section~\ref{sec:Iceberg_logical_ops}.

In Fig.~\ref{fig:0+_state_prep}, we provide a quantum circuit for transforming the logical $\ket{\overline{00}}$ state into $\ket{\overline{0+}}$ by using a joint XX measurement and a single Pauli operator for correction.
Note that the joint XX measurement can be done fault-tolerantly with an additional pair of maximally-entangled ancilla qubits.

\begin{figure}[H]
	\tikzset{
noisy/.style={starburst,fill=yellow,draw=red,line
width=1pt}
}
	\centering
 \begin{quantikz}
 \lstick[4]{$\ket{\overline{00}}$} & & &\rstick[4]{$\ket{\overline{0+}}$} \\
  & &\gate{Z^{a}} & \\
  &\meter[2, style={fill=red!20}, label style={inner sep=1pt}]{a}&\gate{Z^{a}} &\\
  & & &
 \end{quantikz}
					\caption{Physical quantum circuit for transforming the $\ket{\overline{00}}$ state of the $[[4,2,2]]$ Iceberg code into $\ket{\overline{0+}}$.
                    The four qubit lines correspond to the physical qubits of the $[[4,2,2]]$ code.
                    The red measurement with measurement output $a$ is a joint $XX$ measurement.
                    \label{fig:0+_state_prep}
				 }
\end{figure}
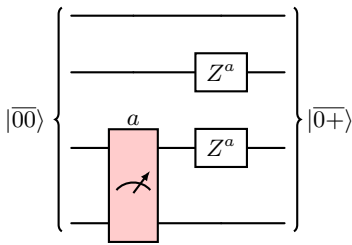

It is easy to see that the quantum circuit in Fig.~\ref{fig:0+_state_prep} gives us the desired transformation.
The logical $\ket{\overline{00}}$ and $\ket{\overline{01}}$ states are given by 
\[\ket{\overline{00}} = \frac{\ket{0000} + \ket{1111}}{\sqrt{2}}, \quad \ket{\overline{01}} = \frac{\ket{0011} + \ket{1100}}{\sqrt{2}}.\]
Since $\ket{\overline{0+}} = \frac{1}{\sqrt{2}}\left(\ket{\overline{00}} + \ket{\overline{01}}\right)$, performing a XX measurement on the third and fourth qubits of $\ket{\overline{00}}$ and getting a measurement eigenvalue of $+1\,(a = 0)$ gives us $\ket{\overline{0+}}$.
When the eigenvalue is $-1$, we apply $Z_2Z_3$ to flip the phase to get the desired $\ket{\overline{0+}}$ state.

The $\ket{\overline{0+}}$ is particularly useful for logical qubit teleportation for our $[[4,2,2]]$ Iceberg code.
We prepare the $\ket{\overline{0+}}$ state several times in the different logical gadgets that we construct in the subsequent sections.

\section{Clifford Logical Operators for $[[4,2,2]]$-concatenated HGP codes}
\label{sec:clifford_logical_ops_concatenated_HGP_code}

In this appendix, we describe how we can obtain the full Clifford group for logical computation on a $[[4,2,2]]$-concatenated square HGP code. 
We follow the concatenation procedure described in Procedure~\ref{proc:concat} and assign logical qubits following Sec.~\ref{sec:assignment}.
To achieve the full Clifford group for the resulting concatenated code
we need the following logical gadgets:
\begin{enumerate}
    \item $\ket{\overline{0}}/\ket{\overline{+}}$ state preparation
    \item Z/X basis measurements
    \item Concatenated grid Pauli product measurements (cGPPMs)
    \item Inter-block CNOTs
    \item H-SWAP
    \item CZ-S
    \item Logical translation
\end{enumerate}
In the subsequent sections, we discuss how each of the logical gadgets above can be constructed for our concatenated code.
Because the logical translation gadget is only available for a special subclass of hypergraph product codes, we leave it until last, before we describe how we construct it.
In fact, the first six logical gadgets can all be implemented without the logical translation gadget.
However, we can achieve better space-time cost savings if we have a logical translation gadget.
Thus, we discuss the different costs of implementing the different logical gadgets both with and without the logical translation gadget in the subsequent sections.

\subsection{$\ket{\overline{0}}/\ket{\overline{+}}$ state preparation}
It is easy to check that our concatenated code is still a CSS code.
Thus, the standard state preparation procedure would still work for our concatenated code~\cite{Dennis_2002}.
For the sake of concreteness, we state how we can prepare the $\ket{\overline{0}}^{\otimes k}$ logical state.
\begin{enumerate}
    \item Initialize all 2n physical qubits to be in the $\ket{0}$ state.
    \item Measure all the $X$ stabilizer generators for our concatenated code.
    \item Repeat step 2 an additional $d-1$ times.
    \item Decode the $d$ sets of syndromes collectively and apply the resulting $Z$ correction to obtain the desired $\ket{\overline{0}}^{\otimes k}$ logical state.
\end{enumerate}
By exchanging $X$ and $Z$ as well as initializing the 2n physical qubits to be in the $\ket{+}$ state, we can prepare the $\ket{\overline{+}}^{\otimes k}$ logical state.

Alternatively, we can consider a single-shot $\ket{\overline{0}}/\ket{\overline{+}}$ state preparation gadget introduced by Hong in Ref.~\cite{hong2024single}.
We now proceed to show why our concatenated code is amenable to such a state preparation gadget.
To do that, we need to show that the homological product of the concatenated code and some classical code has good soundness.
In addition, it would be helpful for us to prove that the concatenated code has linear confinement to show that it is amenable to single-shot decoding.

To begin, we first define soundness and what it means for a quantum code to have good soundness.

\begin{definition}[{\cite[Soundness]{campbell2019theory}}]
Let $t$ be an integer and $f:
\mathbb{Z} \to \mathbb{R}$ be some function called the soundness function
with $f(0) = 0$. 
Given some set of Pauli checks $M$, we say
it is $(t, f)$-sound if for all Pauli errors $E$ with $\left|\sigma(E)\right|$ =
$x < t$, it follows that there exists an $E'$ with $\sigma(E') =
\sigma(E)$ such that $\text{wt}(E') \leq f(x)$.
\end{definition}

\begin{definition}[{\cite[Good Soundness]{campbell2019theory}}]
Consider an infinite
check family $\mathcal{M}_n$. 
We say the family has good soundness
if each $\mathcal{M}_n$ is $(t, f)$-sound where:
\begin{enumerate}
    \item $t$ grows with $n$ such that $t \geq an^b$
for some positive
constants $a, b$. 
That is, $t = \Omega\left(n^b
\right)$ with $b > 0$;
\item and $f(x)$ is some polynomial that is monotonically
increasing with $x$ and independent of $n$.
\end{enumerate}
\end{definition}

Having stated the above definitions, we now proceed to prove that the our concatenated quantum code has the following soundness:

\begin{lemma}[Soundness of ``Transposed'' Concatenated Code]\label{lem:soundness_transposed_concatenated_code}
    Let the chain complex corresponding to the square HGP code be 
    \[C_{-1} \xrightarrow[\delta_{-1} = H_X^\top]{} C_0 \xrightarrow[\delta_0 = H_Z]{} C_1.\]
    In addition, let the chain complex corresponding to our concatenated code be 
    \[\tilde{C}_{-1} \xrightarrow[\tilde{\delta}_{-1} = \tilde{H}_X^\top]{} \tilde{C}_0 \xrightarrow[\tilde{\delta}_0 = \tilde{H}_Z]{} \tilde{C}_1\]
    Then, the maps $\tilde{\delta}_0^\top$ and $\tilde{\delta}_{-1}$ are $(t, f)$-sound with $f(x) = \frac{9x^2}{4} + x$ and $t = d/3$ where $d$ is the distance of the HGP code.
\end{lemma}
\begin{proof}
    From Lemma~5 in Ref.~\cite{campbell2019theory}, we know that the HGP code that we use as our outer code has stabilizer check matrices $H_X$ and $H_Z$ such that $H_X^\top$ and $H_Z^\top$ are $(t,f)$-sound with $f(x) = x^2/4$ and $t = d$. 
    For the subsequent part of the proof, we focus on analyzing $\tilde{H}_X^\top$ and the argument would also hold for $\tilde{H}_Z^\top$.

    First of all, notice that $\tilde{H}_X^\top \in \F_2^{2n \times ((n-k) + n/2)}$ and $H_X^\top \in \F_2^{n \times (n-k)}$.
    Denote the set of $(n-k)$ columns in $H_X^\top$ and $\tilde{H}_X^\top$ as $A$ and the set of $n/2$ columns in $\tilde{H}_X^\top$ as $B$.
    We refer to the stabilizer generators of $\tilde{H}_X$ that correspond to the columns in $A$ as the copied generators.
    The generators of $\tilde{H}_X$ that correspond to the columns in $B$ are referred to as the new generators.
    Denote $\tilde{H}_X^\top|_A$ as the restriction of $\tilde{H}_X^\top$ to the columns in $A$.
    
    For any error configuration $e \in \F_2^{(n-k) + n/2}$ that lies completely within the support of $A$, we have
    \[\left|H_X^\top e\right| \leq  \left|\tilde{H}_X^\top e\right| \leq 2\left|H_X^\top e\right| \]
    from the fact that the copied generators have twice the weight of the original generators in $H_X$ because each physical qubit checked by a generator in $H_X$ is now a logical qubit that corresponds to two physical qubits in the $2n$ physical qubits of the concatenated code. 
    For every $e \in \F_2^{n-k}$ that generates a syndrome in $\tilde{H}_X^\top|_A$ of weight at most $d$, we have $\left|H_X^\top e\right| \leq \left|\tilde{H}_X^\top e\right| \leq d$.
    From the $(d, x^2/4)$-soundness of $H_X^\top$, there exists some $e' \in \F_2^{n-k}$ such that $H_X^\top e = H_X^\top e'$ and $|e'| \leq \left|H_X^\top e\right|^2/4$.
    Now, we claim that $e'$ satisfies $\tilde{H}_X^\top|_A e = \tilde{H}_X^\top|_A e'$.
    To show that, we first note that there exists an injective homomorphism $\gamma:C_0 \to \tilde{C}_0$ since every physical qubit $i \in [n]$ is mapped to some unique pair of physical qubits $(i, j)$ for $i, j \in [n]$ according to the logical operator basis for the $[[4,2,2]]$ code. 
    Thus, 
    \[\tilde{H}_X^\top|_A e = \left(\gamma \circ H_X^\top\right) e = \gamma H_X^\top e' = \tilde{H}_X^\top|_A e'.\]    
    Since $|e'| \leq \left|H_X^\top e\right|^2/4 \leq \left|\tilde{H}_X^\top |_A e\right|^2/4$, this allows us to conclude that $\tilde{H}_X^\top|_A$ is $(d, x^2/4)$-sound. 

    Because $\tilde{H}_X^\top|_B$ consists of $n/2$ different columns with disjoint sets of 4 consecutive ones, it is easy to see that $\tilde{H}_X^\top|_B$ is $(n/2, x/4)$-sound.

    Now, we proceed to show that $\tilde{H}_X^\top$ is $\left(\frac{d}{3}, \frac{9x^2}{4} + x\right)$-sound using the arguments made regarding the soundness of $\tilde{H}_X^\top|_A$ and $\tilde{H}_X^\top|_B$.
    Consider an arbitrary error configuration $\tilde{e} \in \F_2^{(n-k) + n/2}$ such that $\left|\tilde{H}_X^\top \tilde{e}\right| = x \leq d/3$.
    Let $\tilde{e}|_A$ and $\tilde{e}|_B$ denote the restriction of $\tilde{e}$ to the indices in $A$ and $B$ respectively. 
    We now claim that $\left|\tilde{H}_X^\top|_A \tilde{e}|_A\right| \leq 3x \leq d$.
    To see that, note that the physical qubit with index 2 in each Iceberg code block only lies in the support of the new generators since the Iceberg code has logical operator basis $\overline{X}_1 = X_1 X_3$ and $\overline{X}_2 = X_1 X_4$. 
    Since we have 
    \[\tilde{H}_X^\top \tilde{e} = \tilde{H}_X^\top|_A \tilde{e}|_A + \tilde{H}_X^\top|_B \tilde{e}|_B \quad(\bmod{2}),\]
    when $\left[\tilde{H}_X^\top|_A \tilde{e}|_A\right]_i = \left[\tilde{H}_X^\top|_B \tilde{e}|_B\right]_i \neq 0$ for some index $i$ that corresponds to qubit 1, 3, or 4 in an Iceberg code block, there exists some $j$ corresponding to the qubit with index 2 in the same set of 4 Iceberg code qubits containing $i$ such that $\left[\tilde{H}_X^\top \tilde{e}\right]_j = \left[\tilde{H}_X^\top|_B \tilde{e}|_B\right]_j = 1$.
    Since $\left|\tilde{H}_X^\top \tilde{e}\right| = x$, we can have at most $x$ instances of $j$ in $x$ different Iceberg code blocks such that $\left[\tilde{H}_X^\top \tilde{e}\right]_j = \left[\tilde{H}_X^\top|_B \tilde{e}|_B\right]_j = 1$. For each of these instances of $j$, we can have at most 3 $i$'s that belong to the same 4 Iceberg code qubits that contains $j$ such that $\left[\tilde{H}_X^\top|_A \tilde{e}|_A\right]_i = \left[\tilde{H}_X^\top|_B \tilde{e}|_B\right]_i \neq 0$.
    Thus, $\left|\tilde{H}_X^\top|_A \tilde{e}|_A\right| \leq 3x \leq d$.
    From the soundness of $\tilde{H}_X^\top|_A$ and $\tilde{H}_X^\top|_B$, we can find an $\tilde{e}'$ that has restricted forms $\tilde{e}'|_A$ and $\tilde{e}'|_B$ such that $\left|\tilde{e}'|_A\right| \leq (3x)^2 / 4 =  9x^2/4$, $\left|\tilde{e}'|_B\right| \leq 4x/4 = x$ where $\tilde{H}_X^\top \tilde{e} = \tilde{H}_X^\top \tilde{e}'$.
    Thus, we obtain the lemma statement that $\tilde{H}_X^\top$ is $\left(\frac{d}{3}, \frac{9x^2}{4} + x\right)$-sound. 
\end{proof}
We note that the soundness factor can definitely be made tighter by performing a rigorous analysis similar to the one in the proof for Lemma~5 in Ref.~\cite{campbell2019theory}.
However, we do not need such a tight bound for our subsequent arguments.

Now, we are ready to show that the thickened concatenated code has good soundness.

\begin{lemma}[Good Soundness of Thickened Concatenated Code]\label{lem:good_soundness_of_thickened_concatenated_code}
Let the chain complex corresponding to our concatenated code be 
\[\tilde{C}_{-1}\xrightarrow[\tilde{\delta}_{-1} = \tilde{H}_X^\top ]{} \tilde{C}_0 \xrightarrow[\tilde{\delta}_0 = \tilde{H}_Z]{} \tilde{C}_1\]
such that $\tilde{\delta}_0^\top$ and $\tilde{\delta}_{-1}$ are $\left(\frac{d}{3}, \frac{9x^2}{4} + x\right)$-sound where $d$ is the distance of our concatenated code.
Applying the homological product on the concatenated code with a classical code gives us a new length-4 chain complex
\[\breve{C}_{-2}\xrightarrow[\breve{\delta}_{-2} ]{} \breve{C}_{-1}\xrightarrow[\breve{\delta}_{-1} = \breve{H}_X^\top ]{} \breve{C}_0 \xrightarrow[\breve{\delta}_0 = \breve{H}_Z]{} \tilde{C}_1\]
where the map $\breve{\delta}_{-1}^\top = \breve{H}_X$ is $\left(\frac{d}{3}, \frac{9x^3}{4}\right)$-sound.
\end{lemma}
\begin{proof}
    The proof for this lemma statement is effectively the same as the proof for Lemma~6 in Ref.~\cite{campbell2019theory} except that we use Lemma~\ref{lem:soundness_transposed_concatenated_code} from our paper instead of Lemma~6 from Ref.~\cite{campbell2019theory}.
\end{proof}

Next, we provide two important definitions regarding confinement.

\begin{definition}[{\cite[Confinement]{quintavalle2021single}}]
    Let $t > 0$ be an integer and $f:\mathbb{Z}\to \mathbb{R}$ be an increasing function.
    For a parity-check matrix $H \in \mathbb{F}_2^{m\times n}$, we say it is $(t,f)$-confined if for any Pauli errors $e$ with reduced weight $||e|| \leq t$, its syndrome $\sigma(e) = He$ obeys
    \[f(|\sigma(e)|) \geq ||e||.\]
\end{definition}

\begin{definition}[{\cite[Good Linear Confinement]{quintavalle2021single}}]
    Consider an infinite
check family $\mathcal{M}_n$. 
We say the family has good linear confinement
if each $\mathcal{M}_n$ is $(t, f)$-confined where:
\begin{enumerate}
    \item $t$ grows with $n$ such that $t \geq an^b$
for some positive
constants $a, b$. 
That is, $t \in \Omega\left(n^b
\right)$ with $b > 0$;
\item and $f(x)$ is some linear function that is monotonically
increasing with $x$ and independent of $n$.
\end{enumerate}
\end{definition}

When a code has good linear confinement, the code is single-shot against both adversarial and stochastic noise~\cite{quintavalle2021single}.
The intuition is that the linear confinement property guarantees that the low-energy state space of the code Hamiltonian
is partitioned into well-separated clusters~\cite{anshu2023nlts}.

Before we proceed to show that our concatenated code has good linear confinement, we first state the confinement property of good expander codes.

\begin{lemma}[{\cite[Linear Confinement of Expander Codes]{Leverrier_2015}}]\label{lem:expander_code_linear_confinement}
    For any quantum expander code with distance $d$, an arbitrary error $e$ with reduced weight $||e|| < d$ has a syndrome with weight bounded from below as $|\sigma(e)| \geq \frac{1}{3}||e||$.
\end{lemma}

Now, we proceed to show that our concatenated code has good linear confinement.

\begin{lemma}[Linear Confinement of Concatenated Code]\label{lem:linear_confinement_concatenated_code}
    Let the chain complex corresponding to our concatenated code be 
\[\tilde{C}_{-1}\xrightarrow[\tilde{\delta}_{-1} = \tilde{H}_X^\top ]{} \tilde{C}_0 \xrightarrow[\tilde{\delta}_0 = \tilde{H}_Z]{} \tilde{C}_1\]
where $d$ is the distance of our concatenated code.
If the outer code is $\left(d, 3x\right)$-confined, then the stabilizer matrices $\tilde{H}_X$ and $\tilde{H}_Z$ are $\left(\frac{d}{2}, 6x\right)$-confined. 
\end{lemma}
\begin{proof}
    Suppose we have a Pauli error $e$ that has reduced weight $||e|| < d/2$.
    To make things concrete, let $e'$ be any weight-reduced Pauli error such that $\left|e'\right| = ||e||$ and $\sigma(e) = \sigma\left(e'\right)$.
    
    Next, let us partition the length-$2n$ sequence $e'$ for our $[[2n, k, d]]$ concatenated code into sets of 4 qubits that belong to individual Iceberg code patches, i.e.
    \[P = \left\{\left\{\mathcal{P}_{4i+1}, \mathcal{P}_{4i+2}, \mathcal{P}_{4i+3}, \mathcal{P}_{4i+4}\right\}\right\}_{i = 0}^{\frac{n}{2} - 1},\]
    where $\mathcal{P}$ is some Pauli operator in the set $\{I, X, Y, Z\}$.    
    For an arbitrary $i \in \left\{0,1,\ldots,\frac{n}{2} -1\right\}$, $\left\{\mathcal{P}_{4i+1}, \mathcal{P}_{4i+2}, \mathcal{P}_{4i+3}, \mathcal{P}_{4i+4}\right\}$ can only have non-trivial support on 0, 1, or 2 physical qubits else we can apply the Iceberg code stabilizer generators to further reduce the weight of the error.    
    Let $\left\{e'_{1,j}\right\}_j \subseteq P$ be the set of $\left\{\mathcal{P}_{4i+1}, \mathcal{P}_{4i+2}, \mathcal{P}_{4i+3}, \mathcal{P}_{4i+4}\right\}$ sets that contain exactly 1 physical qubit Pauli error.
    Similarly, let $\left\{e'_{2,k}\right\}_k \subseteq P$ be the set of $\left\{\mathcal{P}_{4i+1}, \mathcal{P}_{4i+2}, \mathcal{P}_{4i+3}, \mathcal{P}_{4i+4}\right\}$ sets that contain exactly 2 physical qubit Pauli errors. 
    In addition, let $\sigma_{\textrm{new}}(e')$ and $\sigma_{\textrm{copied}}(e')$ correspond to the part of the syndrome that corresponds to the new and copied concatenated code generators respectively such that 
    \[\sigma(e') = \sigma_{\textrm{new}}(e') \oplus \sigma_{\textrm{copied}}(e').\]
    Construct a length-$n$ Pauli sequence $e'_{\textrm{logical}}$ from $e'$ with the following two steps.
    Firstly, convert each element in $\left\{e'_{2,k}\right\}_k$ into the corresponding logical Pauli operators of the Iceberg code patch.
    Secondly, convert each element of $\left\{e'_{1,j}\right\}_j$ into some (potentially trivial) logical Pauli operator of the Iceberg code patch depending on whether the element anti-commutes with $\overline{X}_1, \overline{X}_2, \overline{Z}_1, \overline{Z}_2$.
    For example, suppose $e'_{1,j} = X_1I_2I_3I_4$ and $e'_{1,j'} = I_1Z_2I_3I_4$, we then convert the former into $\overline{I}_1\overline{I}_2$ since it commutes with all the logical Pauli operators of the Iceberg code patch and the latter into $\overline{Z}_1\overline{I}_2$ since it only anti-commutes with $\overline{X}_1$.
    Finally, let $\varsigma$ denote the syndrome operator of the $\left[\left[n, k, d\right]\right]$ HGP code that was used to construct our concatenated code.
    
    

    Then, we have
    \begin{align}
        \left|\sigma(e)\right| &= \left|\sigma(e')\right| \\
        &= \left|\sigma_{\textrm{new}}\left(e'\right)\right| + \left|\sigma_{\textrm{copied}}(e')\right| \\
        &= \left|\left\{e'_{1, j}\right\}_j\right| + \left|\varsigma\left(e'_{\textrm{logical}}\right)\right| \\
        &\geq \left|\left\{e'_{1, j}\right\}_j\right| + \frac{1}{3} \left|\left|e'_{\textrm{logical}}\right|\right| \\
        &\geq \left|\left\{e'_{1, j}\right\}_j\right| + \frac{1}{3} \left|\left\{e'_{2,k}\right\}_k\right|\\
        &= \left|\left\{e'_{1, j}\right\}_j\right| + \frac{1}{6} \left(2\left|\left\{e'_{2,k}\right\}_k\right|\right) \\
        &\geq \frac{1}{6}\left(\left|\left\{e'_{1, j}\right\}_j\right| + 2\left|\left\{e'_{2,k}\right\}_k\right|\right) \\
        &= \frac{1}{6}|e'| \\
        &= \frac{1}{6}||e||\label{eq:concatenated_code_confinement}.
    \end{align}
    We note that the third equality comes from the fact that each element in $\left\{e'_{1,j}\right\}_j$ triggers the syndrome of a single new generator in the construction of our concatenation scheme.
    In addition, we have $\left|e'_{\textrm{logical}}\right| \leq 2|e'|$ since a single physical Pauli error can be converted into at most 2 logical errors when we construct $e'_{\textrm{logical}}$.
    Thus, we can use Lemma~\ref{lem:expander_code_linear_confinement} to obtain the latter term in the right hand side of the first inequality because $|e'| < d/2$ implies that $\left|e'_{\textrm{logical}}\right| < d$.
    The second inequality comes from discarding the logical errors that were generated from $\left\{e'_{1,j}\right\}_j$.
    It is not immediately obvious that $||e'_{\textrm{logical}}|| \geq \left|\left\{e'_{2,k}\right\}_k\right|$.
    This inequality arises because the stabilizer generators that are applied to reduce the weight of $e'_{\textrm{logical}}$ cannot reduce the number of Iceberg code blocks with two Pauli errors. 
    If so, this would violate the fact that $e'$ is a weight-reduced error for the concatenated code since we could have reduced the number of Iceberg code blocks with two Pauli errors by applying the corresponding copied generators to $e'$.
    Since \eqref{eq:concatenated_code_confinement} holds for $||e|| = |e'| < d/2$, we have shown that the stabilizer matrices of our concatenated code have $\left(\frac{d}{2}, 6x\right)$-confinement.  
\end{proof}

Since Lemma~\ref{lem:good_soundness_of_thickened_concatenated_code} guarantees us that the thickened concatenated code has good soundness and Lemma~\ref{lem:linear_confinement_concatenated_code} guarantees that our concatenated code has the single-shot property, we are able to use the single-shot logical state preparation scheme in Ref.~\cite{hong2024single}.
We now state the single-shot logical state preparation scheme with respect to our concatenated code for the sake of concreteness.

Suppose we have a classical repetition code with code parameters $[d, 1, d]$ where $d$ is the distance of our $\left[\left[\tilde{n}, \tilde{k}, d\right]\right]$ concatenated HGP code. 
Let $\breve{Q}$ be the $\left[\left[\breve{n}, \breve{k}, d\right]\right]$ code that emerges from thickening our concatenated HGP code in the $Z$-basis with the classical repetition code.
In the first stage of the single-shot $\ket{\overline{0}}^{\otimes \tilde{k}}$ state preparation protocol for our concatenated HGP code, we perform the following:
\begin{enumerate}
    \item Initialize all physical qubits in $\ket{0}^{\otimes \breve{n}}$.
    \item Measure all $X$-type check operators to obtain an initial $X$-syndrome before using the $X$-metachecks to obtain an $X$-metasyndrome.
    \item Decode the $X$-metasyndrome and repair the initial $X$-syndrome.
    \item Decode the repaired syndrome using $\tilde{H}_X$ of the concatenated HGP code to return to $\ket{\overline{0}}^{\otimes \breve{k}}$, which is the logical all-0s state for the thickened $\left[\left[\breve{n}, \breve{k}, \breve{d}\right]\right]$ code, up to some small residual $Z$ error.
\end{enumerate}
The second stage of the protocol proceeds as follows:
\begin{enumerate}
    \item Measure $Z$ on all physical qubits of the thickened code except on one of the boundary concatenated HGP codes and reconstruct a bulk $Z$-syndrome.
    \item Input the above bulk $Z$-measurement outcomes, $Z$-syndrome, and a suitable single-shot decoder for our concatenated code before iteratively decoding, inferring, and updating an $X$-correction vector from each layer of concatenated HGP code.
    \item Apply the final $X$-correction vector on the boundary concatenated HGP code to obtain a single layer of concatenated HGP code that is initialized to be in the $\ket{\overline{0}}^{\otimes \tilde{k}}$ state.
\end{enumerate}
By changing the basis for thickening and replacing $X$/$Z$, we can obtain the single-shot state preparation protocol for initializing the $\ket{\overline{+}}^{\otimes \tilde{k}}$ state for the concatenated HGP code.

\subsection{Z/X basis measurements}
To perform logical $Z/X$ basis measurements, we simply have to perform transversal measurements on all $n$ physical qubits of our concatenated HGP code in the $Z/X$ basis. 
In other words, just as it holds for any other CSS codes, $M_Z^{\otimes n}/M_X^{\otimes n}$ gives us the logical measurement $\overline{M}_Z^{\otimes k}/\overline{M}_X^{\otimes k}$ for our concatenated HGP code.

\subsection{Concatenated Grid Pauli Product Measurements (CGPPMs)}


In the original grid Pauli product measurement (GPPM) scheme in Ref.~\cite{xu2024fast}, corresponding physical qubits in the main HGP code and the ancilla HGP code are entangled via physical CNOT gates.
Our concatenated grid Pauli product measurement (CGPPM) scheme works effectively the same as these GPPMs except that we entangle the logical qubit of some Iceberg code with the physical qubit in the ancilla HGP code. 
To perform such a CNOT gate, we use the circuit shown in Fig.~\ref{fig:logical-physical-CNOT} which performs a logical CNOT that is controlled on the first logical qubit of some Iceberg code block and acts on qubit $a$ of a HGP ancilla code block. 
In this paper, we refer to these CNOTs as Logical-Physical (LP) CNOTs.
This circuit can be adapted to the case where we control on a different logical qubit or the case where we decide to control on the physical qubit of the HGP code instead.
This particular CGPPM scheme allows us to reduce the overhead of the GPPM measurement by not having to concatenate the ancilla HGP codes with the Iceberg codes.
In addition, it also allows us to sidestep the challenging task of assigning Iceberg code logical qubits to the punctured or/and augmented ancilla HGP code blocks.
More details regarding how the GPPMs can be customized to the specific Pauli operators of interest can be found in Algorithm 2 of Ref.~\cite{xu2024fast}.

In addition, our CGPPM is compatible with the selective inter-block teleportation scheme described in Algorithm 3 in Ref.~\cite{xu2024fast} that allows the teleportation of any subset of the logical qubits of a square HGP code to the corresponding logical qubits of another identical square HGP code.
By performing the LP CNOTs described above, we can entangle the desired subset of the logical qubits of the original concatenated HGP code block to an ancilla HGP code block before entangling it to the final concatenated HGP code block with another set of LP CNOTs.
We then proceed to perform CGPPM on the ancilla HGP code block and the original concatenated HGP code block to perform the desired selective teleportation.
The details can be easily derived from Algorithm 3 and Fig.~11(a) in Ref.~\cite{xu2024fast}.

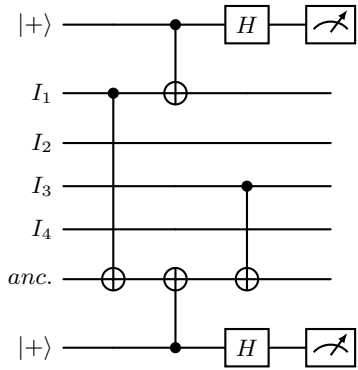
\begin{figure}[H]
	\tikzset{
noisy/.style={starburst,fill=yellow,draw=red,line
width=1pt}
}
	\centering
 \begin{quantikz}
 \lstick{$\ket{+}$} & &\ctrl{1} &\gate{H} &\meter{} \\
 \lstick{$I_1$} &\ctrl{4} &\targ{} & & \\
 \lstick{$I_2$} & & & &  \\
 \lstick{$I_3$} & & &\ctrl{2} &  \\
 \lstick{$I_4$} & & & & \\
 \lstick{$anc.$} &\targ{} &\targ{} &\targ{} &  \\
 \lstick{$\ket{+}$} & &\ctrl{-1} & \gate{H}&\meter{}
 \end{quantikz}
					\caption{By performing physical CNOTs that are controlled on the first and third qubit of the $[[4,2,2]]$ code and then acting on the qubit $anc.$ of the ancilla HGP code, we perform a logical CNOT that is controlled on the first logical qubit of the $[[4,2,2]]$ code and acting on the physical qubit $a$ of the ancilla HGP code. 
                    The $\ket{+}$ states are included as flag qubits to catch the $Z$ error that may appear before the two physical CNOT gates which would propagate to become $Z_1Z_3 = \overline{Z}_1$ in the $[[4,2,2]]$ code. 
                    \label{fig:logical-physical-CNOT}
				 }
\end{figure}

\subsection{Inter-block CNOTs}

For the inter-block CNOTs across all logical qubits of two identical concatenated HGP code, we perform the standard transversal physical CNOTs across all pairs of physical qubits to achieve global logical CNOTs.

Unlike the inter-block CNOTs scheme designed for square HGP codes in Ref.~\cite{xu2024fast}, we are unable to do a set of inter-block CNOTs on an arbitrary subset of the logical qubits in the concatenated HGP code.
This is because doing that would require us to puncture or/and augment the classical codes that are used to construct the HGP code.
The resulting HGP code is not amenable to our current concatenation scheme because the missing columns or rows in the resulting HGP code, in general, remove the symmetry that we use to assign the logical qubits of our $[[4,2,2]]$ code.
While we might still be able to find a different assignment assuming that there are still an even number of physical qubits left in the resulting HGP code, we cannot do tranversal CNOTs on a subset of logical qubits because it will, in general, require us to entangle only one logical qubit of some $[[4,2,2]]$ code with another logical qubit of some other $[[4,2,2]]$ code.
It is unclear whether it is possible for us to do targeted CNOTs on the Iceberg code when the logical qubits of the Iceberg codes are in arbitrary quantum states, especially if the ancilla concatenated HGP code blocks are punctured.
Thus, we construct our inter-block CNOT gadget to perform a global logical CNOT for our concatenation scheme for the HGP code with $[[4,2,2]]$ Iceberg code.
For the rest of the paper, we work with a global inter-block CNOT gadget.

However, we note that it is possible to modify our concatenation scheme so that we utilize only one logical qubit per Iceberg code for every physical qubit of the HGP code.
In other words, we have a concatenated HGP code that has a subsystem structure.
In this case, the concatenated HGP code can inherit the targeted logical CNOT capability from the HGP code since its structure closely resembles that of the original square HGP code.
It is also possible to formulate an alternative inter-block CNOT scheme.
Consider the more restricted setting where one of the blocks is purely an ancilla block where all the logical qubits are either $\ket{\overline{+}}$ or $\ket{\overline{0}}$ states.
To be explicit, suppose we have a punctured concatenated HGP ancilla block that is supposed to be the target block for a transversal logical CNOT.
To construct such a punctured concatenated HGP ancilla block, we first layout the punctured HGP code in the usual geometric layout.
In this picture, there will be some empty rows and columns due to the punctures.
When we concatenate it with the Iceberg codes, we perform the same diagonal symmetric assignment of Iceberg logical qubits to the physical qubits of the punctured HGP code.
For those logical qubits of the Iceberg code blocks that are mapped to the punctured qubits of the HGP ancilla block, we can initialize them to be in the $\ket{\overline{+}}$ or $\ket{\overline{0}}$ state so the global transversal CNOTs do not entangle the logical qubits that are supposed to be punctured.
This can allow us to do inter-block CNOTs on our concatenated HGP codes while respecting the presence of the punctures in the base HGP ancilla blocks, allowing us to only entangle the logical qubits that are not punctured.

\subsection{Intra-Block CNOTs}
\label{sec:Intra-Block-CNOTs}
To implement intra-block CNOT gates between logical qubits in a single $[[4,2,2]]$-concatenated HGP code block, we can perform the circuit in Fig.~\ref{fig:intra-block-CNOT} with our CGPPM gadgets.
Using similar ideas as in Ref.~\cite{xu2024fast}, we can perform all intra-block CNOT gates across logical qubits that are either aligned in the same column or row in $O(\sqrt{k})$ logical cycles.
The intra-block CNOTs that acts on logical qubits that are not in the same column or row can be grouped into clusters where each cluster contains CNOTs that act on pairs of qubits that all reside in the same column and row.
For example, if a CNOT acts on two qubits that are in row 5 and column 10 respectively and another CNOT acts on two other qubits that are also in row 5 and column 10 respectively, these two CNOTs will be in the same cluster.
Each of these clusters can contain at most $\sqrt{k}$ CNOTs and each of these CNOTs will take $O(\sqrt{k})$ logical cycles to implement.
This implies that each cluster takes $O(k)$ logical cycles to implement.
Because these clusters can be implemented in parallel by carefully teleporting the right columns and rows to ancilla code blocks as argued in Appendix B of Ref.~\cite{xu2024fast}, measuring these CNOTs will take $O(k)$ logical cycles.
When the logical translation gadget is available, the number of logical cycles required to perform $O(k)$ intra-block CNOTs can be reduced to $O\left(k^{3/4}\right)$ as described in Appendix B of Ref.~\cite{xu2024fast}.

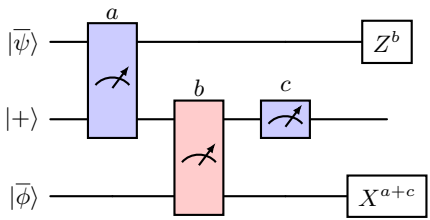
\begin{figure}
	\tikzset{
noisy/.style={starburst,fill=yellow,draw=red,line
width=1pt}
}
	\centering
 \begin{quantikz}
  \lstick{$\ket{\overline{\psi}}$}&\meter[2, style={fill=blue!20}, label style={inner sep=1pt}]{a} & & &\gate{Z^b} \\
 \lstick{$\ket{+}$} & &\meter[2, style={fill=red!20}, label style={inner sep=1pt}]{b} &\meter[style={fill=blue!20}]{c} &   \\
  \lstick{$\ket{\overline{\phi}}$}& & & &\gate{X^{a+c}}
 \end{quantikz}
					\caption{Logical quantum circuit for performing logical CNOT gate between intra-block logical qubits, $\ket{\overline{\psi}}, \ket{\overline{\phi}}$. The gate is controlled on $\ket{\overline{\psi}}$ and targets $\ket{\overline{\phi}}$.
                    The middle qubit line belongs to an ancilla qubit that is initialized in $\ket{\overline{+}}$.
                    The blue measurements with measurement outputs $a$ and $c$ are either $ZZ$ or single $Z$ measurements.
                    The red measurement with measurement output $b$ is an $XX$ measurement.
                    This figure is adapted from Fig.~11(c) in Ref.~\cite{xu2024fast}.
                    \label{fig:intra-block-CNOT}
				 }
\end{figure}

\subsection{H-SWAP}
\label{sec:H-SWAP}
For a square HGP code, the logical Hadamard on all logical qubits is performed by first applying the physical Hadamard on all physical qubits before swapping the twin physical qubits that lie on opposite sides of the principal diagonal.
In addition, this causes the logical qubits that lie on opposites of the principal diagonal in the logical grid to be swapped. More details about the H-SWAP gate for square HGP codes can be found in Ref.~\cite{quintavalle2023partitioning}. It falls under the larger fold-transversal framework of Ref.~\cite{breuckmann2024fold}.
For our $[[4,2,2]]$-concatenated HGP code, H-SWAP is achieved by performing the following:
\begin{enumerate}
    \item Perform physical Hadamard on all physical qubits of the concatenated code.
    \item For each $[[4,2,2]]$ code whose logical qubits are assigned to a pair of physical qubits that are adjacent to each other on the principal diagonal of the HGP code, perform $\mathrm{SWAP}_{3,4}$.
\end{enumerate}

Recall that $H_1 H_2 H_3 H_4$ implements $\overline{H}_1 \overline{H}_2 \overline{\mathrm{SWAP}}_{1,2}$ on a single $[[4,2,2]]$ Iceberg code as shown in Table~\ref{tab:clifford_transversal_logical_ops_baby_Iceberg}.
Because we have assigned the two logical qubits of every Iceberg code to the two twin qubits on opposite sides of the principal diagonal of the HGP code, step 1 would be equivalent to Hadamards and SWAPs on the twin qubits of the HGP code. 
Because we do not want the SWAPs for the qubits on the principal diagonal, we perform step 2 to reverse these swaps for the Iceberg codes that are on the principal diagonal of the HGP code.
Because H-SWAP also performs unwanted swaps on the twin logical qubits of the $[[4,2,2]]$-concatenated HGP code, we have to swap back these pairs of twin logical qubits in order to implement global logical Hadamard on the logical qubits of our concatenated code.
To reverse these swaps, we can first identify pairs of diagonal lines $L$ and $L'$ that are equidistant from the principal diagonal such that the two diagonal lines $L$ and $L'$ correspond to pairs of twin logical qubits.
Next, we teleport $L$ and $L'$ to two different ancilla $[[4,2,2]]$-concatenated HGP blocks that are initialized using Algorithm~4 in Ref.~\cite{xu2024fast} such that the appropriate logical qubits in the ancilla blocks are set to $\ket{\overline{+}}$ and the other logical qubits are set to $\ket{\overline{0}}$.
Next, for each of the two ancilla blocks, we perform a logical SWAP operation using three intra-block CNOTs with the cGPPM gadget so that the logical qubits corresponding to $L\,(L')$ are now swapped to the position of $L'\,(L)$ in the ancilla block.
Lastly, teleport the two diagonal lines of logical qubits back into the original code block.
Since there are $O(\sqrt{k})$ pairs of diagonal lines and the steps stated above take $O(\sqrt{k})$ logical cycles, reversing the unwanted swaps takes $O(k)$ logical cycles.
While the analysis above can be tightened further, showing that implementing global logical Hadamard takes $O(k)$ logical cycles is sufficient since it is not the bottle neck for the time cost for logical computation.

To perform a targeted Hadamard on a subset of logical qubits in a concatenated HGP code block, we simply have to teleport the subset of logical qubits to an ancilla $[[4,2,2]]$-concatenated HGP code block.
We then perform global transversal Hadamard as described above on all the logical qubits in the ancilla code block before teleporting the desired subset back to the original concatenated HGP code block. 
An alternative method to implement targeted Hadamard on $O(k)$ logical qubits would involve iteratively teleporting columns of the logical qubits of interest to an ancilla $[[4,2,2]]$-concatenated HGP code block before performing logical swaps to move them to the principal diagonal of the ancilla block.
Subsequently, we perform the H-SWAP logical gadget on the ancilla block.
Because the logical qubits of interest are on the principal diagonal, there is no need to reverse the unwanted logical swaps because the logical qubits on the principal diagonal are not impacted by the unwanted logical swaps.
Lastly, we swap the logical qubits out of the principal diagonal before teleporting them back to the original code block.
This process is repeated at most $\sqrt{k}$ times because the logical qubits of interest can have support on at most $\sqrt{k}$ columns. Since each iteration takes $O(\sqrt{k})$ logical cycles, it takes at most $O(k)$ logical cycles.
Therefore, to perform targeted Hadamard on $O(k)$ logical qubits, we require $O(k)$ logical cycles.
We note that performing targeted Hadamard on $O(k)$ logical qubits when given access to the logical translation gadget will take $O(\sqrt{k}\log k)$ logical cycles.
Readers can find the details for the implementation in Ref.~\cite{xu2024fast}.
We note that the H-SWAP gadget for the $[[4,2,2]]$-concatenated HGP code requires less non-locality because the physical non-local SWAP gates in the regular HGP code are replaced by the logical swaps for the $[[4,2,2]]$ Iceberg codes that can be implemented locally. 
However, recall that we have to reverse the unwanted logical swaps implemented by the H-SWAP logical gadget.
If we do not have the logical translation gadget, the logical swap implemented by the intra-block CNOTs may consume more non-local entangling gates.

\subsection{CZ-S}
\label{sec:CZ-S}
Using the same fold-transversal framework discussed in Ref.~\cite{breuckmann2024fold}, Quintavalle \emph{et al.} introduced a fold-transversal CZ-S logical gate for square HGP codes~\cite{quintavalle2023partitioning}. The gate applies a $S$ gate on left diagonal logical qubits, $S^\dagger$ on right diagonal logical qubits, and CZ between logical twin qubits.
The gate is performed by:
\begin{enumerate}
    \item Apply CZ gates on pairs of twin qubits
    \item Apply physical $S$ gates to the physical qubits which lie on the principal diagonal of the $L$ sector and $S^\dagger$ gates to the physical qubits which lie on the principal diagonal of the $R$ sector.
\end{enumerate} 


From Table~\ref{tab:clifford_transversal_logical_ops_baby_Iceberg}, observe that applying $S_1^\dagger S_2^\dagger S_3 S_4$ on each Iceberg code corresponding to twin qubits of the HGP code implements the desired logical CZ gate between those twin logical qubits.
The $\overline{S}$ and $\overline{S}^\dagger$ gates as listed in Table~\ref{tab:clifford_logical_ops_baby_Iceberg2} are not transversal because of the physical CZ gate and thus are not inherently fault-tolerant.
By introducing two additional flag qubits, we can use the circuit in Fig.~2(c) in Ref.~\cite{chao2018fault} to implement the physical CZ gate fault-tolerantly.
The explicit construction of the circuit for implementing $\overline{S}_1 = S_1 S_3 CZ_{1,3}$ is shown in Fig.~\ref{fig:cz-flag}. Note that $\overline{S}^\dagger_i$ can be obtained by simply applying $\overline{Z}_i$ after the circuit for $\overline{S}_i$.
Hence we have fault-tolerant implementations of the $\overline{S}, \overline{S}^\dagger$, and $\overline{CZ}$ gates for the $[[4,2,2]]$ code, with which we can perform the CZ-S logical gadget on the $[[4,2,2]]$-concatenated HGP code.


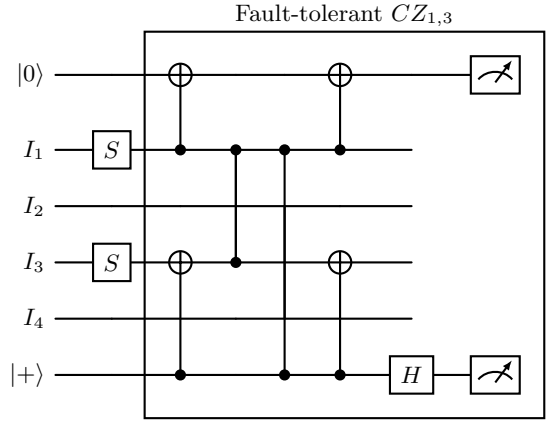
\begin{figure}
	\tikzset{
noisy/.style={starburst,fill=yellow,draw=red,line
width=1pt}
}
	\centering
 \begin{quantikz}
 \lstick{$\ket{0}$} & &\targ{}\gategroup[6, steps=6, style={inner sep=6pt}]{Fault-tolerant $CZ_{1,3}$} & & &\targ{} & &\meter{} \\
 \lstick{$I_1$} &\gate{S} &\ctrl{-1} &\ctrl{2} &\ctrl{3} &\ctrl{-1} &\\
 \lstick{$I_2$} & & & & & &\\
 \lstick{$I_3$} &\gate{S} &\targ{} &\ctrl{-2} & &\targ{} &\\
 \lstick{$I_4$} & & & & & & \\
 \lstick{$\ket{+}$} &&\ctrl{-2} & &\ctrl{-3} &\ctrl{-2} &\gate{H} &\meter{}   \end{quantikz}
					\caption{Fault-tolerant circuit for implementing $\overline{S}_1 = S_1 S_3 CZ_{1,3}$ on the $[[4,2,2]]$ code.
                    \label{fig:cz-flag}
				 }
\end{figure}

To perform targeted S gate on any subset of logical qubits in the concatenated HGP code block, we have to adopt a slightly different procedure. 
Because the CZ-S gadget only allows us to generate $\ket{i}$ states on the principal diagonal, we have to teleport them to the non-diagonal logical qubits in order to be able to apply $\overline{S}$ gates on any subset of logical qubits in the concatenated HGP code block.
To perform the teleportation, we can use the circuit shown in Fig.~\ref{fig:S-gate-teleportation}.
The explicit procedure for implementing logical $S$ gates on any subset of logical qubits would involve the following:
\begin{enumerate}
    \item Initialize the concatenated HGP code to be in the $\ket{\overline{0}}^{\otimes k}$ state. 
    We can teleport the logical qubits to another concatenated HGP code block before performing the initalization if necessary.
    \item Prepare $\sqrt{k}$ $\ket{\overline{i}}$ states on the principal diagonal of the concatenated HGP code using Algorithm 4 in Ref.~\cite{xu2024fast} and our CZ-S logical gadget without performing logical translation.
    \item Teleport the original off-diagonal logical qubits back to the concatenated HGP code block. 
    \item Use the circuit in Fig.~\ref{fig:S-gate-teleportation} to teleport the $\sqrt{k}$ S gates on the principal diagonal    onto the off-diagonal logical qubits that lie in the desired subset.   
    \item If there are excess $S$ gates on the diagonal, we can use a CGPPM to reset those logical qubits to $\ket{\overline{+}}$ so that we can teleport the original diagonal logical qubits back.
    If there are not enough $S$ gates on the diagonal, we can repeat Steps 1 - 5 by carefully teleporting the logical qubits with the completed $S$ gates away.
\end{enumerate}

Suppose we are interested in applying $O(k)$ logical $S$ gates.
Steps 1, 3, 4, and 5 take $O(\sqrt{k})$ logical cycles while step 2 takes $O(\log k)$ logical cycles.
We repeat steps 1 - 5 at most $O(\sqrt{k})$ times.
Thus, the entire procedure for applying $O(k)$ $S$ gates would take $O(k)$ logical cycles.
We note that our procedure is relatively slower than the one in Ref.~\cite{xu2024fast} for square HGP codes because we are not able to prepare $O(k)$ $\ket{\overline{i}}$ states in parallel without the logical translation gadget.
In the protocol stated in Ref.~\cite{xu2024fast} where they assume that the HGP code is constructed from quasi-cyclic classical codes, they are able to implement logical $S$ gates on $O(k)$ logical qubits in $O(\sqrt{k}\log k)$ logical cycles.

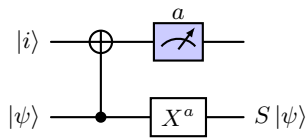
\begin{figure}
	\centering
 \begin{quantikz}
  \lstick{$\ket{i}$} &\targ{} &\meter[style={fill=blue!20}, label style={inner sep=1pt}]{a} & \\
 \lstick{$\ket{\psi}$} &\ctrl{-1} &\gate{X^a} & \rstick{$S\ket{\psi}$}
 \end{quantikz}
					\caption{Logical quantum circuit for teleporting the logical $S$ gate to the logical qubits that do not lie on the principal diagonal. The first qubit line corresponds to some logical qubit in the $\ket{i} \coloneqq S\ket{+}$ state that lies on the principal diagonal and the second qubit line corresponds to some logical qubit that is in some arbitrary $\ket{\psi}$ state and lies outside of the principal diagonal.
                    The blue measurement with measurement output $a$ is a single $Z$ measurement.  
                    \label{fig:S-gate-teleportation}
				 }
\end{figure}

\subsection{Logical translation}

Some quantum error correcting codes are able to implement logical gates simply by permuting their physical qubits. One such example of an automorphism gate~\cite{bravyi2023highthreshold, breuckmann2024fold} is the $\overline{SWAP}_{1,2}$ gate for the $[[4,2,2]]$ code.
Using a specific subclass of quasi-cyclic base classical LDPC codes, Ref.~\cite{xu2024fast} showed that the resulting HGP codes have automorphism gates which implement a translation of the logical qubits.
Because it is not known whether there exist asymptotically good classical LDPC codes satisfying the quasi-cyclic and one-generator-systematic circulant (OGSC) properties, we only included the construction of the logical translation gadget at the end and have discussed how we can implement the other gadgets without access to the logical translation gadget.

For our $[[4,2,2]]$-concatenated HGP code, permutation of the Iceberg logical qubits is a non-trivial task since we need to preserve the diagonal symmetry post-automorphism.
To be explicit, the Iceberg code blocks are initially assigned such that the two logical qubits are mapped to adjacent diagonal qubits or twin qubits of the HGP code.
Permuting the `physical' qubits of the $[[4,2,2]]$-concatenated HGP code would, in general, destroy this symmetry and prevent us from implementing logical gadgets post-permutation.
Thus, we need to make sure adjacent diagonal qubits and twin qubits are part of the same Iceberg block after permutation.

To accomplish this, we use the fact that we are able to teleport a single logical qubit out of the $[[4,2,2]]$ code, following Ref.~\cite{gottesman1998theory}. 
Suppose we are interested in teleporting away the second logical qubit of the Iceberg code that is in the $\ket{\overline{\psi\phi}}$ state.
To do so, we:
\begin{enumerate}
    \item Prepare an ancilla Iceberg code block in the $\ket{\overline{0+}}$ state following Sec.~\ref{sec:state_prep_Iceberg_codes}.
    \item Perform a transversal CNOT targeting $\ket{\overline{\psi\phi}}$ and controlled on the ancilla code block.
    \item Perform a non-destructive logical $Z$ measurement of the second logical qubit of $\ket{\overline{\psi\phi}}$.
\end{enumerate}
The resulting state of the ancilla Iceberg code block is then $\ket{\overline{0\phi}}$.
We can also perform an analogous procedure with $\ket{\overline{+0}}$ as the ancilla block to teleport the first logical qubit of $\ket{\overline{\psi\phi}}$.
Let us now discuss how we can produce $\ket{\overline{\psi \phi}}$ from $\ket{\overline{\psi 0}}$ and $\ket{\overline{0\phi}}$ using teleportation.
We can perform the following:
\begin{equation}
\nonumber
\ket{\overline{\psi 0}} \xrightarrow[]{H^{\otimes 4}} \ket{\overline{+}} \otimes \overline{H}\ket{\overline{\psi}} \xrightarrow[]{Tel.\, \ket{\overline{0+}}} \ket{\overline{0}} \otimes \overline{H}\ket{\overline{\psi}} \xrightarrow[]{H^{\otimes 4}} \ket{\overline{\psi +}}    
\end{equation}
where $Tel.\,\ket{\overline{0+}}$ teleports $\overline{H}\ket{\overline{\psi}}$ to an ancilla Iceberg code block.
We can perform the analogous operations to $\ket{\overline{0\phi}}$ to obtain $\ket{\overline{+\phi}}$.
Now, we perform a transversal logical CNOT that is controlled on $\ket{\overline{+\phi}}$ and targeting $\ket{\overline{\psi +}}$ before performing a non-destructive logical $Z$ measurement on the first logical qubit of $\ket{\overline{\psi +}}$.
This would implement the following transformation: $\ket{\overline{+\phi}} \xrightarrow[]{} \ket{\overline{\psi\phi}}.$

Now, we are ready to provide an explicit description of the logical translation gadget.
For each Iceberg code block $B$ in the concatenated HGP code, we prepare two ancilla Iceberg code blocks in $\ket{\overline{0+}}$ and $\ket{\overline{+0}}$ before teleporting the two logical qubits of $B$ into the two different ancilla code blocks.
We then assign the ancilla code blocks according to the desired automorphism gate before teleporting the new adjacent diagonal and twin qubits into the same Iceberg block.
For each Iceberg code block, we use at most three additional Iceberg code blocks to perform the permutation.
Because this process can be done in parallel for each Iceberg code block, the logical translation gadget requires $O(1)$ logical cycles.

\subsection{Combining the Logical Gadgets}
We restate the main theorem for the logical computation protocol for our $[[4,2,2]]$-concatenated HGP code for the readers' convenience.

\begin{theorem}[Clifford Gates for Concatenated HGP Code]
    A single layer of an ideal Clifford circuit on $k$ logical qubits with $\Theta(k)$ gates consisting of Hadamard, S, and CNOT gates can be simulated on a square HGP code concatenated with $[[4,2,2]]$ Iceberg code blocks at either one of the following space-time costs:
    \begin{enumerate}
        \item $O(k)$ space and $O\left(k^{3/2}\right)$ time
        \item $O\left(k^{3/2}\right)$ space and $O(k)$ time.
    \end{enumerate}
    If the square HGP code was constructed from OGSC quasi-cyclic base codes, then the concatenated code can simulate the layer with either one of the following space-time costs:
    \begin{enumerate}
        \item $O(k)$ space and $O\left(k^{5/4}\right)$ time
        \item $O\left(k^{3/2}\right)$ space and $O\left(k^{3/4}\right)$ time.
    \end{enumerate}
\end{theorem}

\begin{proof}
    The proof follows directly from the construction of the logical gadgets described in the appendix.
    The space-time overhead for the concatenated HGP code constructed with OGSC quasi-cyclic codes can be easily shown to be true using arguments from App.~B of Ref.~\cite{xu2024fast}.
    We note that the latter space-time cost for each of the construction in the theorem statement comes from the fact that we can perform all of our logical gadgets in $O(1)$ logical cycles at the expense of an increased space cost due to the single-shot state preparation.
\end{proof}
Similar to Ref.~\cite{xu2024fast}, we emphasize that the number of logical cycles is a rather loose upper bound that comes from a constructive compilation and the actual time costs of simulating the layer of $O(k)$ Clifford gates is likely to be significantly lower.
An efficient compilation of the Clifford gates would almost definitely allow us to improve the parallelism and drastically reduce the circuit depth.

\section{Universal Computation}
Using the magic state distillation and consumption protocols constructed in Section V-C-2 of Ref.~\cite{xu2024fast}, we are able to use the logical gadgets constructed above to perform the same ``8-to-CCZ'' magic state distillation protocol and implement the non-Clifford gates in parallel for $k$ logical qubits.
If we were to use a general square HGP code with good expansion, we can either perform the above with $O(k)$ space and $O\left(k^{3/2}\right)$ time or $O\left(k^{3/2}\right)$ space and $O\left(k\right)$ time.
In the event where we use a square HGP constructed from OGSC quasi-cyclic base codes, our protocol allows us to perform the above with $O(k)$ space and $O\left(k \log k\right)$ time or $O\left(k^{3/2}\right)$ space and $O\left(\sqrt{k}\log k\right)$ time.
Again, the less efficient space-time costs for the non-quasi-cyclic construction of HGP codes come from the fact that we are unable to use the logical translation gadget which would reduce the number of logical cycles for the CZ-S logical gadget from $O(k)$ to $O(\sqrt{k}\log k)$.

\end{document}